%% file: main.tex
\documentclass[11pt,a4paper]{article}

\newif\ifarxiv
\arxivtrue

\usepackage[T1]{fontenc}
\usepackage[ruled,  linesnumbered, vlined]{algorithm2e}

\makeatletter
\renewcommand{\nllabel}[1]
 {{\let\@currentlabel\algocf@currentlabel
  \let\@currentcounter\algocf@currentcounter
  \label{#1}}}%

\renewcommand{\algocf@nl@sethref}[1]{%
  \renewcommand{\theHAlgoLine}{\thealgocfproc.#1}%
  \hyper@refstepcounter{AlgoLine}%
  \gdef\algocf@currentlabel{#1}%
  \gdef\algocf@currentcounter{AlgoLine}%
 }%
\makeatother

\usepackage{float}

\usepackage[square, numbers]{natbib}

\usepackage{amsmath,amsfonts,amssymb}
\usepackage{amsthm}
\usepackage{thmtools}
\usepackage{thm-restate}
\usepackage{dsfont}
\usepackage[a4paper, margin=1in]{geometry}
\usepackage{graphicx}

\usepackage{enumerate}
\usepackage{enumitem}

\usepackage{upref}
\usepackage[acronym]{glossaries}
\usepackage[pagebackref,colorlinks=true,urlcolor=blue,linkcolor=blue,citecolor=green!50!black]{hyperref}
\usepackage[noabbrev, capitalize]{cleveref}
\hypersetup{colorlinks=true,allcolors=blue}
\usepackage{comment}

\usepackage{xcolor}
\usepackage{xparse}
\usepackage{stmaryrd}

\input{header/newcommand}

\newacronym{mpc}{MPC}{Massively Parallel Computation}

\title{A Scalable and Unified Framework to Weighted Rank Aggregation}
\date{}

\author{
  Amir Carmel\thanks{Weizmann Institute of Science, Rehovot, Israel. \href{mailto:amir6423@gmail.com}{\texttt{amir6423@gmail.com}}.} \and
  Debarati Das\thanks{Pennsylvania State University, University Park, PA, USA. \href{mailto:debaratix710@gmail.com}{\texttt{debaratix710@gmail.com}}.} \and
  Tien-Long Nguyen\thanks{Pennsylvania State University, University Park, PA, USA. \href{mailto:tfn5179@psu.edu}{\texttt{tfn5179@psu.edu}}.}
}

\begin{document}
\maketitle
\newcommand{\candidates}{\mathcal{C}}
\input{abstract}

\thispagestyle{empty}

\input{introduction}

\paragraph{Organization.} The remainder of the paper is organized as follows.
\cref{sec:framework} formalizes the general framework: we prove the universal \cref{property:universal} (\cref{lemma:property1}) and combine it with the metric-specific~\cref{property:metric.specific} to obtain~\cref{alg.space.efficient.general.framework} and its approximation guarantee (\cref{thm:generalframework}).
\cref{sec:description.of.scalable.cycle.removal} presents our new approximation algorithm $\scalableCycleRemoval$ for rank aggregation under the Ulam metric, describing each of its four steps: $\windowDecomposition$, $\localAggregation$, $\blockConcate$, and $\postprocess$. \cref{sec:analysis_for_the_algorithm_for_Ulam_distance} is dedicated to the analysis of $\scalableCycleRemoval$, culminating in the proof of \cref{lem:bound.gmedian.to.offlineOut}.
\cref{section:application} establishes the metric-specific \cref{property:metric.specific} for Hamming, Spearman's footrule, Kendall-tau, and Ulam distances, yielding our offline approximation results for the weighted setting. Finally, \cref{sec:mpc_implementation} presents the MPC implementations of all four local-solution algorithms and yields our main MPC results, \cref{thm:mpc.main.ulam.theorem,thm:mpc.main.theorem}.

\input{framework}

\input{Ulam/local_algorithm.tex}

\input{Ulam/analysis.tex}

\input{hammingSF}

\input{MPCimplementation}

\bibliographystyle{plainnat}
\bibliography{bibliography}

\end{document}

%% file: header/newcommand.tex
\theoremstyle{plain}
\newtheorem{theorem}{Theorem}[section]

\newtheorem{lemma}[theorem]{Lemma}

\newtheorem*{claim*}{Claim}
\newtheorem{corollary}[theorem]{Corollary}

\newtheorem{proposition}[theorem]{Proposition}

\newtheorem{property}{Property}
\makeatletter
\newenvironment{sameproperty}[1]{%
  \addtocounter{property}{-1}%
  \renewcommand{\theproperty}{\ref*{#1}}%
  \renewcommand{\theHproperty}{formal.\ref*{#1}}%
  \begin{property}%
}{%
  \end{property}%
}
\makeatother

\theoremstyle{definition}

\theoremstyle{remark}

\makeatletter
\newcommand{\setword}[2]{%
  \phantomsection
  #1\def\@currentlabel{\unexpanded{#1}}\label{#2}%
}
\makeatother

\newcommand{\opt}{\mathrm{OPT}}

\newcommand{\bigO}[1]{\mathcal{O}\left( #1 \right)}
\newcommand{\otilda}[1]{\widetilde{\mathcal{O}}\left(#1\right)}
\newcommand{\polylog}{\mathrm{polylog}}

\newcommand{\card}[1]{\left\lvert #1 \right \rvert}

\newcommand{\eqdef}{:=}

\newcommand{\ceil}[1]{\lceil{#1}\rceil}
\newcommand{\ceilenv}[1]{\left\lceil #1 \right\rceil}
\newcommand{\tuple}[1]{\mathbf{#1}}
\newcommand{\pair}[1]{\left\langle #1 \right\rangle}

\DeclareMathOperator{\costop}{cost} 
\DeclareMathOperator{\dist}{\texttt{dist}}
\NewDocumentCommand{\distE}{m g}{%
  \IfNoValueTF{#2}
    {\dist^{E}\left(#1\right)}%
    {\dist^{E}_{#2}\left(#1\right)}%
} 
\NewDocumentCommand{\distP}{m g}{%
  \IfNoValueTF{#2}
    {\dist^{P}\left(#1\right)}%
    {\dist^{P}_{#2}\left(#1\right)}%
} 
\NewDocumentCommand{\weight}{g}{%
    \IfNoValueTF{#1}
    {w}%
    {w\left(#1\right)}%
} 
\NewDocumentCommand{\kt}{m g g}{
  \IfNoValueTF{#2}{%
    \dist\left(#1\right)%
  }{
    \IfNoValueTF{#3}{%
      \dist_{#2}\left(#1\right)%
    }{
      \dist^{#3}_{#2}\left(#1\right)%
    }%
  }%
} 
\newcommand{\unali}[1]{I_{#1}} 
\newcommand{\cost}[2]{\costop(#1,#2)}
\newcommand{\costp}[1]{\cost{#1}{P}}
\NewDocumentCommand{\ed}{g}{%
    \IfNoValueTF{#1}
    {\mathrm{d}}%
    {\mathrm{ID}\left(#1\right)}%
} 
\newcommand{\lcs}{\mathrm{LCS}} 
\NewDocumentCommand{\aggregatePerm}{g}{%
    \IfNoValueTF{#1}
    {\texttt{AggregateToPermutation}}%
    {\texttt{AggregateToPermutation}\left(#1\right)}%
} 
\newcommand{\inpset}{P} 
\newcommand{\smallSubset}{Q} 
\NewDocumentCommand{\pickBetter}{g}{%
    \IfNoValueTF{#1}
    {\texttt{PickBetterPermutation}}%
    {\texttt{PickBetterPermutation}\left(#1\right)}%
} 

\DeclareMathOperator*{\argmin}{arg\,min}

\newcommand{\permutations}{\mathcal{S}_n}
\newcommand{\gmedian}{x^{*}} 
\newcommand{\lmedian}{y^{*}} 
\newcommand{\perm}{x} 

\newcommand{\slack}[3]{\Delta_{#1,#2}(#3)}
\newcommand{\slackopt}[2]{\Delta^*_{#1,#2}}
\newcommand{\metricspace}{\mathcal{X}}
\newcommand{\totalslack}[1]{\Delta(#1)}
\newcommand{\opttotalslack}{\Delta^*}
\newcommand{\wnorm}[1]{\|#1\|}

\newcommand{\indic}{\mathds{1}}
\newcommand{\distPara}{\delta} 
\newcommand{\ltime}[1]{t_{\ell}(#1)} 
\newcommand{\tdist}[1]{t(#1)} 

\newcommand{\scalableCycleRemoval}{\mathtt{ScalableMedianReconstruct}} 
\newcommand{\windowDecomposition}{\mathtt{Window\ Decomposition}} 
\newcommand{\localAggregation}{\mathtt{Block\ Reconstruction}} 
\newcommand{\blockConcate}{\mathtt{Block\ Composition\ via\ Dynamic\ Programming}} 
\newcommand{\postprocess}{\mathtt{PostProcessing}} 
\newcommand{\mpcSpaceRegime}{\varepsilon} 
\newcommand{\genStr}{z} 
\newcommand{\approxFact}{\alpha} 
\newcommand{\approxFactMPCED}{\lambda} 

\NewDocumentCommand{\windowSet}{g}{%
    \IfNoValueTF{#1}
    {\mathcal{W}}%
    {\mathcal{W}_{#1}}%
} 
\NewDocumentCommand{\sWindowSet}{g}{%
    \IfNoValueTF{#1}
    {\mathcal{S}}%
    {\mathcal{S}_{#1}}%
} 
\NewDocumentCommand{\windTupSet}{g}{%
    \IfNoValueTF{#1}
    {W}%
    {W_{#1}}%
} 
\NewDocumentCommand{\gap}{g}{%
    \IfNoValueTF{#1}
    {g}%
    {g_{#1}}%
} 
\newcommand{\edErr}{\rho}
\NewDocumentCommand{\upperBoundED}{g}{%
    \IfNoValueTF{#1}
    {\phi}%
    {\phi_{#1}}%
} 
\newcommand{\blSize}{\kappa} 
\NewDocumentCommand{\optStartPoint}{g}{%
    \IfNoValueTF{#1}
    {\alpha}%
    {\alpha_{#1}}%
} 
\NewDocumentCommand{\optEndPoint}{g}{%
    \IfNoValueTF{#1}
    {\beta}%
    {\beta_{#1}}%
} 
\NewDocumentCommand{\startPoint}{g}{%
    \IfNoValueTF{#1}
    {s}%
    {s_{#1}}%
} 
\NewDocumentCommand{\endPoint}{g}{%
    \IfNoValueTF{#1}
    {e}%
    {e_{#1}}%
} 
\NewDocumentCommand{\mpcGoodString}{g}{%
    \IfNoValueTF{#1}
    {\widetilde{y}}%
    {\widetilde{y}_{ #1 }}%
} 
\NewDocumentCommand{\offlineGoodString}{g}{%
    \IfNoValueTF{#1}
    {\widetilde{s}}%
    {\widetilde{s}_{ #1 }}%
} 
\NewDocumentCommand{\offlineOut}{g}{%
    \IfNoValueTF{#1}
    {y_{\mathrm{offline}}}%
    {y_{ \mathrm{offline}, #1 }}%
} 
\NewDocumentCommand{\offlineLocalOut}{g}{%
    \IfNoValueTF{#1}
    {s}%
    {s_{ #1 }}%
} 
\NewDocumentCommand{\concateStr}{g}{%
    \IfNoValueTF{#1}
    {s}%
    {s( #1 )}%
}
\NewDocumentCommand{\offlineIntS}{g}{%
    \IfNoValueTF{#1}
    {s}%
    {s_{ #1 }}%
} 
\NewDocumentCommand{\mpcOut}{g}{%
    \IfNoValueTF{#1}
    {y^{\mathrm{MPC}}}%
    {y^{\mathrm{MPC}}_{ #1 }}%
} 
\NewDocumentCommand{\mpcLocalOut}{g}{%
    \IfNoValueTF{#1}
    {y}%
    {y_{ #1 }}%
} 
\NewDocumentCommand{\mpcInp}{g}{%
    \IfNoValueTF{#1}
    {I}%
    {I_{#1}}%
} 
\NewDocumentCommand{\mpcRDist}{g}{%
    \IfNoValueTF{#1}
    {R_{D}}%
    {R_{#1}}%
} 
\NewDocumentCommand{\mpcSDist}{g}{%
    \IfNoValueTF{#1}
    {S_{D}}%
    {S_{#1}}%
} 
\NewDocumentCommand{\machineSet}{g}{%
    \IfNoValueTF{#1}
    {\mathcal{M}}%
    {\mathcal{M}_{#1}}%
}
\NewDocumentCommand{\machine}{m g}{%
    \IfNoValueTF{#2}
    {M_{#1}}%
    {M_{#1}^{#2}}%
} 
\NewDocumentCommand{\lInt}{g}{%
    \IfNoValueTF{#1}
    {L}%
    {L_{#1}}%
} 
\NewDocumentCommand{\nInt}{g}{
    \IfNoValueTF{#1}
    {N}%
    {N_{#1}}%
} 
\NewDocumentCommand{\anInt}{g}{
    \IfNoValueTF{#1}
    {AN}%
    {AN_{#1}}%
} 
\NewDocumentCommand{\sInt}{g}{
    \IfNoValueTF{#1}
    {S}%
    {S_{#1}}%
} 
\newcommand{\blIntS}[1]{\sigma_{#1}} 
\newcommand{\mpcLRound}{R_{\ell}} 
\newcommand{\mpcLSpace}{S_{\ell}} 

\newcommand{\mpcIdx}[1]{\mathrm{Idx}_{#1}} 
\NewDocumentCommand{\errGood}{g}{%
    \IfNoValueTF{#1}
    {\bar{G}}%
    {\bar{G}_{#1}}%
} 
\NewDocumentCommand{\appGood}{g}{%
    \IfNoValueTF{#1}
    {\widetilde{G}}%
    {\widetilde{G}_{#1}}%
} 
\newcommand{\specialChar}{*} 
\newcommand{\tp}{\mathtt{tp}} 

\definecolor{greenlantern}{rgb}{0.0, 0.5, 0.0}

%% file: abstract.tex
 \begin{abstract}
 
 The rank aggregation problem, seeks to combine multiple rank orderings of the same set of candidates into a single consensus ordering. Such problems arise in diverse domains, including web search, employment, college admissions, and voting. 
In this work we focus on the 1-median objective: given a set of $m$ rankings over $[n]$, the goal is to compute a ranking that minimizes the sum of its distances to all input rankings.
 
We study rank aggregation under several classical distance metrics: Ulam distance, Spearman’s footrule, Hamming distance, and Kendall-tau, as well as their weighted variants. Our contributions begin with a novel unified framework that identifies a key structural property: it suffices to focus on a small subset of rankings (of size three or five), where the corresponding local one-median provides a good approximation to the global median. This principle extends across these distance measures, yielding a general algorithmic framework for weighted rank aggregation.

Building on this, we present a new approximation algorithm for rank aggregation under the Ulam distance that scales in the Massively Parallel Computation (MPC) model. Our algorithm computes a $(2-\alpha)$-approximation, for a constant $\alpha>0$, to the $1$-median in a constant number of rounds, using local memory sublinear in $n$ (the size of a ranking) and total memory near linear in $n$.

We further design new MPC approximation algorithms for Spearman's footrule and for the \emph{element-weighted} variants of Hamming and Kendall-tau distances. For each metric, we obtain a $(2-\zeta)$-approximation, for a constant $\zeta>0$ (which may differ across metrics), to the $1$-median in a constant number of rounds, using local memory sublinear in $n$ and total memory linear or near-linear in $n$.

Moreover, for the Ulam distance, where computing the $1$-median is NP-hard 
[Fischer et al., ESA, 2025], we simplify and strengthen the analysis of Chakraborty et al. [ITCS 2023], obtaining an improved $1.968$-approximation 
that further extends to the weighted setting. 

 \end{abstract}

%% file: introduction.tex
\clearpage
\setcounter{page}{1}

\newif\iffirsttime
\firsttimetrue

\newif\ifsecondtime
\secondtimetrue

\newif\ifthirdtime
\thirdtimetrue

\section{Introduction}
Aggregating inconsistent information from diverse sources is a fundamental challenge across disciplines such as social choice theory and information retrieval. The task of reconciling potentially conflicting preferences into a single consensus ranking is known as rank aggregation. This well-studied problem traces back to the late 18th century, when Condorcet and Borda introduced early voting systems for elections with more than two candidates. In modern settings, rank aggregation continues to play a central role in applications ranging from sports and elections to search engines, databases, web evaluation systems, and statistics~\cite{DKNS01, rank1, KV10, kolde2012robust, chen2015spectral, SVWXY19, kuhlman2020rank}. The complexity of the problem is amplified when disagreements or cycles arise in the input rankings, reflecting the conflicts among different sources. Addressing these inconsistencies to produce a meaningful consensus makes rank aggregation a critical and widely impactful data aggregation task.

Given a set of $n$ items, a ranking can be represented by a permutation $\pi$ on $[n]$. To compare different rankings, several distance measures have been introduced, including Kendall-tau distance~\cite{kendall1938new,diaconis1977spearman,kemeny1959mathematics,MS07,ACN08,KV10} (also known as Kemeny distance in the context of rank aggregation), Hamming distance~\cite{cicirello2019classification}, Spearman’s footrule distance~\cite{spearman1961proof,spearman1906footrule,diaconis1977spearman,DKNS01,KV10}, and Ulam distance~\cite{CDK21, DBLP:conf/innovations/Chakraborty0K23,fischer2025hardness} (see~\cref{subsec:problemsetup} for formal definitions). Among the most studied aggregation frameworks are median rank aggregation (or simply rank aggregation)~\cite{kemeny1959mathematics,young1988condorcet,young1978consistent,DKNS01} and maximum rank aggregation~\cite{bachmaier2015hardness,popov2007multiple,biedl2009complexity}, which aim to find the median and the center ranking of a given set, respectively. In this paper, we consider the continuous version of these problems where the median/center can be any permutation from the input metric space. 

Rank aggregation is computationally tractable under Hamming distance and Spearman’s footrule, where exact solutions can be obtained using straightforward minimum-cost bipartite matching algorithms. In contrast, under Kendall-tau and Ulam distance, the problem is NP-hard~\cite{alon2006ranking,ACN08,fischer2025hardness}. Nevertheless, a folklore 2-approximation follows directly from the triangle inequality, and a series of works have pushed beyond this trivial approximation. In particular, there exists a PTAS for Kendall-tau~\cite{MS07} while for Ulam distance~\cite{CDK21, DBLP:conf/innovations/Chakraborty0K23} achieved a 1.999-approximation.

Traditional rank aggregation metrics, as discussed above, treat all errors uniformly, 
overlooking the fact that in real-world settings, some mistakes are more costly than 
others. In information retrieval, for instance, misplacing a highly relevant document 
should incur a much larger penalty than misplacing an irrelevant one. This motivates the use of element weights, where each item is assigned a relevance score, and mistakes involving high-weight elements are penalized more heavily~\cite{KV10, sevaux2005permutation}.

In many applications, rank aggregation must be performed on massive datasets, where not only is the number of rankings large, but also the size of each ranking. For instance, scholarly search engines such as Corpus rank over 200M papers~\cite{pomikalek2012building}, e-commerce platforms order millions of products~\cite{takanobu2019aggregating}, 
and genomic studies combine extensive datasets to prioritize disease-associated genes~\cite{aerts2006gene}. Handling such large-scale data naturally motivates the design of efficient parallel algorithms. 

In particular, this large-scale setting presents two primary computational challenges: handling a large number of input rankings $m$, and processing rankings of large size $n$ that may exceed the memory capacity of a single machine. We address the first challenge using sampling techniques, working with only a polylogarithmic-sized sample of the input rankings. The second challenge is more substantial and necessitates distributed parallel computation, naturally motivating the design of efficient parallel algorithms. In particular, when $n$ is large, even storing a single ranking on one machine may be infeasible, requiring rankings to be distributed across multiple machines.

A well-studied framework for this purpose is the Massively Parallel Computation (MPC) model~\cite{DBLP:conf/stoc/AndoniNOY14,DBLP:journals/jacm/BeameKS17,DBLP:conf/isaac/GoodrichSZ11,DBLP:conf/soda/KarloffSV10}. In this model, each machine has full access to its own local memory, but communication between machines occurs only between rounds. Thus, the round complexity of an algorithm becomes the central performance measure, since network communication is typically the main bottleneck in practice. The ultimate objective is to design constant-round algorithms, which are highly desirable for large-scale systems.
Despite its importance, rank aggregation has not been well studied in the massively parallel setting, and existing algorithms cannot be implemented directly in the MPC model. For instance, algorithms for Hamming or Spearman’s footrule rely on min-cost matching, but efficient implementations of matching in constant rounds with memory sublinear in the number of vertices are not known, necessitating new scalable approaches. In this work, we introduce a generalized framework that can be deployed in scalable settings and apply it to several classical rank aggregation measures.

\subsection{Problem formulation and preliminaries}\label{subsec:problemsetup}

Let $(\mathcal{X}, \dist)$ be an arbitrary metric space. Consider a finite set $P = \{p_1, p_2, \ldots, p_m\} \subseteq \mathcal{X}$. For a given $x \in \mathcal{X}$, we define the cost function as
\begin{equation*}
\costp{x} = \frac{1}{|P|} \sum_{i=1}^m \dist(x, p_i).
\end{equation*}
The \emph{1-median problem} seeks to find a point in $ \mathcal{X}$ that minimizes this cost function. Let $x^* \in \mathcal{X}$ be an optimal solution for $P$, that is, $x^* = \arg\min_{x \in \mathcal{X}} \costp{x}$, and denote $\opt = \costp{x^*}$ as the optimal cost. Note that $x^*$ may not be unique; when multiple optimal solutions exist, we fix an arbitrary choice.

While our framework is conceptually applicable to general metric spaces, in this paper we focus on the case where $\mathcal{X} = \permutations$ is the set of all permutations on $[n]$ elements, and $\dist$ is some distance metric over $\permutations$. This setting is also known as the \emph{rank aggregation problem}. Specifically, we study four classical distance measures on permutations: \emph{Spearman's footrule distance},  \emph{Hamming distance}, \emph{Kendall-tau distance}, and \emph{Ulam distance}, along with the element-weighted variants.

\paragraph{Spearman's footrule distance} 
measures the $\ell_1$ distance between the position vectors of two permutations:
\[
\dist_F (p,q) = \sum_{i=1}^n \Big|p[i]-q[i]\Big|
\]

\paragraph{Hamming distance}
measures the number of positions where two permutations differ:
\[
\dist_H(p,q) = \sum_{i=1}^n \indic_{p[i] \neq q[i]}.
\]
The weighted version penalizes mismatches proportionally to the average weight of the mismatched elements:
\[
\dist^w_H(p,q) = \sum_{i=1}^n \frac{w(p[i])+w(q[i])}{2}\indic_{p[i] \neq q[i]}.
\]

\paragraph{Kendall-tau distance}
measures the number of elements pairs $(e, e')$ such that the relative order of $e$ and $e'$ is reversed between permutations $p$ and $q$. Formally,
\[
\dist_\tau (p,q) = \Big|
\big\{ \{e,e'\} : (p^{-1}[e] - p^{-1}[e']) (q^{-1}[e] - q^{-1}[e']) < 0 \big\}
\Big|
\]
where $p^{-1}[e]$ is the position of the element $e$ in the permutation $p$.
The weighted version penalizes inversions proportionally to the average weight of the two elements:
\[
\dist^w_\tau (p,q) = \sum_{\substack{e < e' \\ (p^{-1}[e] - p^{-1}[e']) (q^{-1}[e] - q^{-1}[e']) < 0}} \frac{w(e)+w(e')}{2}.
\]

\paragraph{Ulam distance} measures the minimum number of move operations required to transform one permutation into another, where a move operation removes an element and reinserts it at a different position. We express this using the Indel distance (also known as LCS distance), which counts single-character insertions and deletions. The Indel distance is exactly twice the Ulam distance, since each move corresponds to one deletion followed by one insertion. 

$$\dist_U (p,q) = n-LCS(p,q) = \frac{1}{2}\ed{p,q}$$

In the weighted variant, each insertion or deletion of element $i$ has cost $w(i)$. Thus, the total weighted distance equals the sum of weights of all elements that must be moved to transform one permutation into the other.

\paragraph{Framework.} Our framework relies on analyzing the geometric structure of the metric space through 
the following key quantity. For an element $x \in \mathcal{X}$ and a pair $i,j\in [m]$, 
we consider the \emph{triangle inequality slack}
\begin{equation*}
\slack{i}{j}{x} =  \dist(x,p_i) + \dist(x,p_j) - \dist(p_i, p_j).
\end{equation*}
This quantity is twice the \emph{Gromov product}, 
a fundamental concept in hyperbolic geometry. 
By the triangle inequality, we have $\slack{i}{j}{x} \geq 0$ for all $i, j \in [m]$ and $x \in \mathcal{X}$. 

Finally, for any subset $Q \subseteq P$, we define the \emph{total slack} as
\begin{equation}
\totalslack{x,Q} = \sum_{\substack{i < j, \\ p_i,p_j \in Q}} \slack{i}{j}{x}.
\end{equation}
While the total slack depends on the set $Q$, we omit $Q$ as it is clear from context, and further denote $\opttotalslack = \totalslack{x^*}$ the total slack of the optimal median. A small total slack of $x$ with respect to $Q$ means that while the points in $Q$ are mutually distant from one another, $x$ serves as an effective median for the entire set. 
Note that,
\begin{align*}
    \totalslack{x} = 2\binom{|Q|}{2}\cost{x}{Q} -   \sum_{\substack{i < j, \\ p_i,p_j \in Q}} \dist(p_i,p_j).
\end{align*}
It follows that an optimal median of $Q$ minimizes the total slack, and in particular
\begin{equation}\label{eq:localleqglobal}
    \min_{y \in \metricspace} \totalslack{y} \leq \opttotalslack.
\end{equation}
The inequality holds because $x^*$, which is optimal for $P$, need not be optimal for the subset $Q$.

\paragraph{MPC Model.} 
Throughout this paper, we assume the input consists of 
$m$ permutations on $[n]$.
In the \emph{MPC model}~\cite{DBLP:conf/stoc/AndoniNOY14,DBLP:journals/jacm/BeameKS17,DBLP:conf/isaac/GoodrichSZ11,DBLP:conf/soda/KarloffSV10}, both the number of machines and the local memory per machine are required to be significantly smaller than the input size. 
Accordingly, we fix a parameter $0 < \epsilon < 1$, assume each machine has $\otilda{n^{1-\epsilon}}$ memory\footnote{Throughout, $\widetilde{\mathcal{O}}(f) = O(f \cdot \polylog f)$}, and aim to minimize the total number of machines used by the algorithm.

An MPC algorithm proceeds in a sequence of rounds. 
Within each round, every machine performs computations on the data stored in its memory. 
No communication is allowed during a round; instead, communication occurs only between rounds, with the restriction that the amount of data received by any machine does not exceed its memory capacity. 
Furthermore, any data output by a machine must be computed solely from the data already present locally. 
The input data is initially distributed across the machines.

\subsection{Our contribution}
In this paper we propose a new unified framework for providing approximate solutions to the 1-median problem for various ranking measures. The key idea can be summarized by the following principle:
\emph{If no input point serves as a good median, then there exists a constant-size subset $Q \subseteq P$ such that a good approximation to the 1-median over $Q$ lies within a small distance of the optimal solution for $P$. }  We refer to the approximation over $Q$ as a \emph{local solution} and the optimal 
median for $P$ as the \emph{global median}.

The main advantage of this framework lies in its efficiency: since the subset $Q$ is small, it is computationally efficient to obtain a local solution $y$ that is optimal or near-optimal for $Q$, and the distance between $x^*$ and $y$ can be bounded using the slack expressions $\totalslack{x^*}$ and $\totalslack{y}$.

While Chakraborty, Das, and Krauthgamer~\cite{DBLP:conf/innovations/Chakraborty0K23} also proposed a construction for approximating the median from a small number of input permutations, their approach relies on combinatorial properties specific to the Ulam distance. In contrast, our local-vs-global property is significantly more general, yielding a broader and unified approach that applies to a variety of ranking problems. 

Furthermore, the framework of~\cite{DBLP:conf/innovations/Chakraborty0K23} has inherent bottlenecks that limit its applicability in distributed settings. In contrast, our first contribution is to show that our framework is sufficiently robust to be deployed in distributed settings. We design new algorithms in the \gls{mpc} model that compute approximate medians over such subsets $Q$. This scalability is crucial for modern applications involving massive datasets, where both the number of rankings and the length $n$ of each ranking are large, making it impractical for a single machine to process even one ranking~\cite{aerts2006gene,pomikalek2012building,takanobu2019aggregating}.
Following this we propose the first approximation algorithm breaking the 2-factor barrier for  rank aggregation under Ulam distances in the MPC model, using $\otilda{n^{1-\mpcSpaceRegime}}$ memory per machine, total memory near-linear in $n$: the size of ranking , and a constant number of rounds.

\begin{restatable}{theorem}{mpcUlamResult} \label{thm:mpc.main.ulam.theorem}
For any constants $\distPara > 0$, $ \mpcSpaceRegime \in (0,1/8)$,\ifsecondtime\footnote{ The algorithm we present in this paper works with $ 0 < \mpcSpaceRegime < 1/8 $. However, $ 1/8 < \mpcSpaceRegime < 1 $ is not a barrier to our technique. For a sufficiently small constant $\varepsilon>0$, the algorithm uses local memory sublinear in $n$ and total memory near-linear in $n$.
}$^{,}$\footnote{Smaller $\delta$ yields better approximations but requires sampling more permutations 
from the input, increasing both space complexity and running time by constant factors.}\fi
there exists a polynomial-time \gls{mpc} algorithm that, given a set of $m$ permutations $P \subseteq \permutations$, computes, with high probability, a $(2 - \approxFact + \bigO{\distPara})$-approximation to the $1$-median under the Ulam distance, where $\alpha>0$ is a constant.\ifthirdtime\footnote{Our analysis gives $ \alpha = 0.000007 $. This constant has not been optimized and can be further improved.}\fi
The algorithm uses $\bigO{1}$ communication rounds, has total space $\otilda{n^{1+6\mpcSpaceRegime}}$, and requires $\otilda{n^{1-\mpcSpaceRegime}}$ local memory per machine.

\end{restatable}
\secondtimefalse
\thirdtimefalse

Next, we show the applicability of our framework to weighted rank aggregation by establishing this property for the element-weighted variants of Kendall-tau, Hamming, and Ulam distances, and for Spearman's footrule.
More importantly, we show that the framework extends beyond mere generality and remains sufficiently robust even in the weighted setting.
Building on this framework, we propose the first approximation algorithms breaking the 2-factor barrier for rank aggregation under Spearman's footrule and the element-weighted variants of Hamming and Kendall-tau distances in the MPC model, using $\otilda{n^{1-\mpcSpaceRegime}}$ memory per machine, linear or near-linear total memory, and a constant number of rounds.

\begin{restatable}{theorem}{mpcresult} \label{thm:mpc.main.theorem}
    For any constant $\delta>0$, and $\epsilon \in (0,1)$, there is a polynomial time \gls{mpc} algorithm that, given a set of $m$ permutations $P \subseteq \permutations$, with high probability computes the following: 
        \begin{itemize}
        \item $(1.75+\bigO{\delta})$ approximation of $1$-median under Spearman's footrule distance and element-weighted Hamming distance using $\bigO{1}$ communication rounds, and $\otilda{n^\epsilon}$ processors each with $\bigO{n^{1-\epsilon}}$ local memory. 

        \item $(1.9+\bigO{\delta})$ approximation of $1$-median under element-weighted Kendall-tau distance using $\bigO{1}$ communication rounds, and $\otilda{n^{2\epsilon}}$ processors each with $\bigO{ n^{1-\epsilon} }$ local memory.
        
    \end{itemize}
\end{restatable}
\firsttimefalse

Next, we simplify the analysis of~\cite{DBLP:conf/innovations/Chakraborty0K23}, improve the approximation factor from $1.999$ to $1.968$, and extend the result to the weighted Ulam metric. In contrast to prior work, which is restricted to the unweighted setting, we present the first approximation algorithm for computing a median under the element-weighted Ulam distance.

\begin{restatable}{theorem}{ulamresult}
There is a polynomial time algorithm that, given a set of $m$ permutations $P \subseteq \permutations$, provides a 1.968-approximate solution to the $1$-median problem under the element-weighted Ulam metric with high probability.
\end{restatable}

Note that in our problem the total input size is $O(mn)$. A traditional MPC model only enforces that local memory is sublinear in the input size. In contrast, our results guarantee significantly stronger bounds by requiring local memory to be only sublinear in $n$. This is particularly important since, in most real-world applications, the number of rankings $m$ can be much larger than the size of each ranking $n$.

Finally, we note that our framework, particularly the local-global principle, is 
not specific to permutations and may extend to other metric spaces where a below-$2$ 
approximation for the $1$-median problem remains open. For example, the edit 
distance metric for the well-known median string problem~\cite{hoppenworth2020fine, CDK21}, and 
the Fr\'echet distance or Dynamic Time Warping metrics for the average curve 
problem~\cite{DBLP:conf/swat/BuchinDS20}.

\subsection{Technical overview}\label{subsec:techoverview}

\paragraph{Framework Overview.} Our general framework operates under the assumption that no point in $P$ achieves cost better than $(2-\alpha)\opt$ for some constant $\alpha \in (0,1)$. Otherwise, we can achieve a $(2-\alpha)$-approximation by simply returning the minimum-cost point from the input set.
With this assumption, the framework relies on two key properties:

\begin{property}[\textbf{Universal}]\label{property:universal}
    For every $2 \leq r \leq m$, there exists a subset $Q \subseteq P$ of size $r$ such that 
    \begin{equation}\label{eq:smallq}
        \opttotalslack \leq \alpha \binom{r}{2} \opt.
    \end{equation}
\end{property}

\begin{restatable}[\textbf{Metric-specific}]{property}{SecondProperty} \label{property:metric.specific}
    There exists a constant $r$ such that for every subset $Q \subseteq P$ of size $r$ and every $y \in \metricspace$,
    \begin{equation}\label{eq:distq}
        \dist(x^*,y) \leq \totalslack{x^*} + \totalslack{y}.
    \end{equation}
\end{restatable}

~\cref{property:universal} holds universally for all metric spaces. We prove this through an averaging argument by considering the pairwise distances of points in $P$ and expressing them using the distances from the optimal median $x^*$.
This intuitively suggests that for any two permutations in $Q$, their distance is 
approximately equal to the sum of their respective distances to the optimal median.

However,~\cref{property:metric.specific} must be established individually for each metric. 
When it holds and $y$ is an optimal median of $Q$, Equation~\eqref{eq:localleqglobal} 
implies $\dist(x^*, y) \leq 2\opttotalslack$. This means $x^*$ and $y$ are close, 
so $y$ serves as a good median for $P$, establishing the \emph{local-global property}. For local solutions, where $y$ is not optimal median of $Q$, we demonstrate how $\totalslack{y}$ can be bounded in terms of $\opttotalslack$ for our specific metrics. A detailed analysis of our framework is given in~\cref{sec:framework}.

\paragraph{New Scalable Approximation Algorithms for 1-Median under Ulam Distance.}
\ifarxiv
Building on the above framework, in~\cref{sec:description.of.scalable.cycle.removal} (algorithm description), and~\cref{sec:analysis_for_the_algorithm_for_Ulam_distance} (analysis) we present a new approximation algorithm for the 
1-median under the Ulam distance, which we later show can be adapted to the MPC setting.
\else
Building on the above framework, in~\cref{sec:description.of.scalable.cycle.removal} (algorithm description) we present a new approximation algorithm for the 
1-median under the Ulam distance, which we later show can be adapted to the MPC setting.
\fi
 For the Ulam metric, we assume that the set $Q$ contains five permutations. 
A natural baseline approach for computing the $1$-median of the permutations in $Q$ is to apply a dynamic programming algorithm. 
However, this approach is not space-efficient, as it requires $n^5$ space, which is prohibitive in the \gls{mpc} setting.

An alternative is to adapt the strategy from~\cite{DBLP:conf/innovations/Chakraborty0K23}. 
In this approach, an element is called \emph{bad} if it is not aligned in an optimal alignment between a permutation in $Q$ and the optimal median for at least two permutations in $Q$. 
Let the \emph{bad set} denote the collection of all such elements. 
The algorithm constructs a tournament graph by taking the majority relative order for each pair of elements induced by the permutations in $Q$, then removes cycles and outputs a topological ordering.

The main bottleneck of this approach is the cycle-removal step. 
In~\cite{DBLP:conf/innovations/Chakraborty0K23}, the algorithm repeatedly identifies the shortest cycle and removes all its vertices. 
Since the graph is a tournament, the shortest cycle always has length three, and it can be shown that at least one vertex in each such triangle must be bad, which allows bounding the number of good elements removed.

However, this approach does not extend to the \gls{mpc} model. Because no single machine can store an entire permutation, a $3$-cycle $(a,b,c)$ may not be present on any one machine and therefore may go undetected. 
One possible workaround is to ensure that every triple of elements is collocated on at least one machine. 
But this leads to further complications: a single element may appear on many machines, and different machines may independently select distinct triangles involving the same vertex (e.g., $(a,b,c)$ and $(a,d,e)$). 
This is problematic because the~\cite{DBLP:conf/innovations/Chakraborty0K23} analysis crucially relies on the selected cycles being vertex-disjoint in order to bound the number of good elements removed.

It is nontrivial to reconcile these independently detected cycles while maintaining vertex disjointness in the \gls{mpc} setting. 
This difficulty forms the main obstacle to implementing prior approaches. 
To overcome this, we propose a new framework and algorithm that provides a local solution for the permutations in $Q$.

First, fix an optimal median of $Q$, denoted by $y$.
We partition $y$ into blocks of length $\blSize= n^{1-\varepsilon}$,
for some constant $\varepsilon > 0$.
Our goal is to construct, for each such block of $y$, a substring of comparable
length that provides a good approximation whenever the objective value of the
block is not large. We classify an objective value as large if it is at least $\kappa/\rho$, where $\rho > 0$ is a small constant to be specified later. 
If the objective value is large, we instead use a block consisting only of dummy
elements, which already yields a sufficiently good approximation. Furthermore, when approximating a block, we focus only on the good elements and ignore the bad elements. This is valid because the total number of bad elements is small and their cumulative contribution to the objective is large, allowing them to be handled arbitrarily. For this we devise a four-step process.

\vspace{1mm}
\noindent
\textbf{Windows Decomposition:} For each permutation $\pi \in Q$, fix an optimal alignment between $\pi$ and $y$. Let $q$ be a block of $y$. Define $e^\ell \in q$ to be the leftmost element of $q$ that is matched in the optimal alignment between $\pi$ and $y$, and define $e^r \in q$ to be the rightmost element with this property.
Suppose $e^\ell$ is matched to $\pi[i]$ and $e^r$ is matched to $\pi[j]$.
We define $\pi^q$ to be the substring $\pi[i,j]$.
The collection of substrings $\{\pi^q : \pi \in Q\}$ suffices to reconstruct $q$;
we refer to these substrings as the \emph{constructive blocks}.
The main challenge is to identify these constructive blocks efficiently.

We now describe a window decomposition for the permutations in 
$Q$, beginning with several observations specific to our problem.
A key difficulty is that constructive blocks may have highly variable lengths.
We first consider the case where at least one constructive block is very large,
for example of size greater than $k|q|$ for some sufficiently large constant $k$.
In this situation, the objective value of $q$ is already large, and therefore a
trivial approximation using a block of dummy elements suffices.

Next, consider the case where at least two constructive blocks are very small.
Then many elements of $q$ do not appear in both of these blocks and hence remain
unmatched in at least two permutations from $Q$.
Such elements are therefore \emph{bad elements}.
This implies that the number of bad elements in $q$ is large.
Since the total number of bad elements across all blocks is small, this scenario
cannot occur frequently.
Consequently, in this case as well, a good approximation can be obtained by using
a trivial block consisting only of dummy elements.

From this point onward, we assume that no constructive block is too large and that
at most one constructive block is very small.
Even under this assumption, the constructive blocks may still have different
sizes.
To handle this, we introduce a windowing strategy that, for each permutation
$\pi \in Q$, generates windows of variable sizes starting at various indices.
We ensure that for every constructive block $\pi^q$ whose size is neither too
large nor too small, there exists a window $\pi[i',j']$ such that 
\[
|i-i'| + |j-j'| \le \rho \cdot \mathrm{obj}(q),
\]
where $\mathrm{obj}(q)$ denotes the objective value of the block $q$. 

A crucial aspect of this strategy is that we allow windows of size zero to handle very small constructive blocks, and we consider such zero-length windows at multiple starting indices.
We provide the details of the window decomposition 

\vspace{1mm}
\noindent
\textbf{Block Reconstruction:} Next, using these windows as constructive blocks, we attempt to estimate
the blocks of $y$.
Since we do not know in advance which windows correspond to a particular block, we try various choices of five windows. Specifically, we form \emph{constructive groups} by selecting one window from each string in $Q$ using a structured method rather than enumerating all combinations.

Fix an arbitrary choice of five windows (constructive blocks), forming a constructive group.
We first construct a graph whose vertices correspond to elements that appear in
at least four constructive blocks.
This restriction is justified because we only aim to approximate the good
elements: any element that does not appear in at least four blocks is classified
as a bad element.

For each pair of vertices, corresponding to two elements $c$
and $d$, we orient the edge between them as follows: we add an edge from $c$ to $d$ if $c$ appears before $d$ in at least three of the
five blocks, and an edge from $d$ to $c$ if $d$ appears before $c$ in at least three blocks.
If neither direction has a strict majority, we add no edge and delete both vertices $c$ and $d$.
This is valid because the absence of a majority implies that at least one of $c$
or $d$ is a bad element.
Since the total number of bad elements is small, the number of good elements deleted
in this process is also small.

After this pruning step, the resulting graph is a tournament.
This property is crucial for the next phase, where we remove cycles.
We iteratively find the shortest cycle in the graph and delete all vertices in
that cycle.
Because the graph is a tournament, the shortest cycle has length at most three,
and among the vertices of such a cycle, at least one must be a bad element.
Moreover, since the cycles removed in different iterations are vertex-disjoint,
the total number of good elements deleted remains bounded by the number of bad
elements.

Once the graph becomes acyclic, we compute a topological ordering of the remaining
vertices, which yields a candidate block for the constructive group. For the subsequent dynamic programming step, we ensure that each
candidate block has length at least $\kappa$.
Thus, if the string obtained from the topological ordering has length less than $\kappa$, we
pad it with dummy elements to reach length $\kappa$.
The details of this step are provided in~\cref{sub:local_aggregating_algorithm}.

\vspace{1mm}
\noindent
\textbf{Block Composition via Dynamic Programming:} 
\ifarxiv
In~\cref{sub:combining_blocks}, we combine these candidate blocks, each of length at least $\kappa$,
using dynamic programming to obtain a good approximation of $y$.
\else
We combine these candidate blocks, each of length at least $\kappa$,
using dynamic programming to obtain a good approximation of $y$.
\fi
Since the dynamic program operates on blocks, the total state space is polynomial
in $n/\kappa = n^{\varepsilon}$.
Consequently, the DP can be implemented in the MPC setting within a single round
by storing all blocks on one machine.

An important subtlety is that, as discussed earlier, for some blocks of $y$ the
corresponding constructive blocks may be either very large or very small.
In such cases, we approximate the block of $y$ using a block consisting entirely
of dummy elements.
However, we do not include these dummy-only blocks explicitly in the dynamic
program, as doing so would significantly increase the space requirements.
Instead, they are handled implicitly within the DP formulation.

The output of the dynamic program is a string that provides a good approximation
of $y$ over the good elements.
Moreover, each element appears at most once in this string.
This property is guaranteed by the construction of the graph used in the previous
step, where a vertex corresponding to an element is included only if the element
appears in at least four constructive blocks. 

\vspace{1mm}
\noindent
\textbf{Post Processing:} The DP output may omit some elements and may include dummy elements
As argued earlier, the number of such elements is small and bounded by the number
of bad elements.
To convert the output into a permutation, we replace each dummy element with a
missing element.
One simple approach would be to delete all dummy elements and append the missing
elements at the end, which would still yield a good approximation.
However, our chosen strategy is crucial for enabling an efficient implementation
in the MPC setting. 
The details are given in~\cref{sub:postprocessing}.

\paragraph{New Scalable Approximation Algorithms for 1-Median with Element Weights.}

In the element-weighted setting although exact (or almost exact) $1$-median algorithms may exist for computing the median of 
$Q$, they may not be directly adapted to the massively parallel setting.
Such is the case for the rank aggregation problem under Hamming, Spearman's footrule and Kendall-tau. To this end, we design algorithms that are scalable and are well-suited for implementation in the MPC model

Again, the cornerstone of our approach is to devise a consensus strategy among the permutations in $Q$. We then show that for the relevant $r$, this consensus is close to the optimal median of the entire set $P$. We establish these distance bounds using the total slack term, which provides a unified framework across metric spaces. 
The details and analysis of the algorithms are presented in~\cref{section:application}.

\vspace{-4mm}
\paragraph{Hamming Distance.} Our algorithm constructs a consensus by selecting the majority element at each position where one exists among the three permutations in $Q$. 
The reasoning for selecting the majority as the basis of our consensus strategy is that if the optimal median does not follow the majority, then both input permutations supporting the majority will differ from the optimum, thereby increasing their slack. Since the total slack is assumed to be small (by Equation~\ref{eq:smallq}), this situation cannot occur frequently.
Positions without majority agreement are filled arbitrarily with the remaining elements. The key insight is that such positions are exactly those contributing to the total slack: since they lack a majority, they affect both the local solution and the global optimum equally.
Importantly, this property holds even in the element-weighted setting, making our consensus strategy robust to weights and yielding a $1.75$-approximation.

\vspace{-4mm}
\paragraph{Spearman's Footrule.} We construct a consensus by taking the median element at each position across the three permutations, which yields a pseudo-permutation that may contain duplicates. This consensus is then converted into a valid permutation by sorting the elements and reassigning ranks according to their sorted order. While it is clear that taking the median value at each position minimizes the median cost, it is less obvious that this remains optimal once we enforce the consensus to be a valid permutation. 

Our key observation is that whenever the optimal median permutation selects a value different from the position-wise median, the slack increases for at least a pair of inputs. For example, suppose at position $i$ the three input permutations $\pi_1, \pi_2, \pi_3$ have values $a \geq b \geq c$. If the median instead uses some $b' \neq b$ (say $b'>b$) at index $i$, then the slack of $\pi_2$ and $\pi_3$ at this position increases by $2(b'-b)$. Since the total slack is assumed to be small, such deviations cannot occur frequently. Thus, even under the constraint that the output must be a permutation, taking the position-wise median remains the best consensus strategy. This transformation ultimately yields the closest valid permutation to our consensus under the Spearman's footrule metric, resulting in an overall $1.75$-approximation.

\vspace{-4mm}
\paragraph{Kendall-tau Distance.} The Kendall-tau metric compares inversions between pairs of elements. The consensus is determined by majority vote over element pairs rather than positions. We construct a tournament graph where vertex $a$ precedes vertex $b$ if $a$ appears before $b$ in at least two of the three permutations. The resulting feedback arc set problem, solved via the KWIK-SORT algorithm, yields a permutation that respects the majority preference on most pairs. 
While the KWIK-SORT algorithm was applied to rank aggregation under the Kendall-tau metric in~\cite{ACN08} to obtain an approximate 1-median solution, the weighted setting requires a different approach. Using our framework, a weighted tournament is derived from the three input permutations, with edge weights assigned so that the triangle inequality holds, thereby enabling the use of the KWIK-SORT algorithm.

\vspace{-4mm}
\paragraph{Improved Approximation for Ulam Distance.}

Unlike the previous metrics, consensus under Ulam distance requires subsets $Q$ of size $5$. This follows from a more general property: for any two permutations $x,y$, we show the distance from $x$ to $y$ is at most $\totalslack{x} + \totalslack{y}$. Combined with the universal~\cref{property:universal}, this enables us to improve the metric-specific analysis of~\cite{DBLP:conf/innovations/Chakraborty0K23} and naturally extend the result to the element-weighted variant.

\paragraph{MPC Algorithm for Approximating 1-Median.}
We start by describing how permutations are represented in the \gls{mpc} model and how pairwise distances are computed, primitives common to all metrics.

\vspace{-4mm}
\paragraph{Storing and Computing Distance in MPC:}

A key challenge in the \gls{mpc} model is that each machine has only $\widetilde{\mathcal{O}}(n^{1-\mpcSpaceRegime})$ local memory, which is insufficient to store an entire permutation; consequently, each permutation must be distributed across at least $\widetilde{\mathcal{O}}(n^{\mpcSpaceRegime})$ machines. This sublinear memory constraint makes it nontrivial to compute a local solution and validate its cost even for a small subset $Q$ of permutations, necessitating specialized \gls{mpc} algorithms tailored to different distance metrics. Hamming and Spearman's footrule distances can be computed in constant rounds using $\bigO{n}$ total space via block-wise aggregation. Kendall-tau distance requires $n^{2\mpcSpaceRegime}$ machines to handle inversions between all block pairs, resulting in a total of $ \bigO{n^{1+\mpcSpaceRegime}} $ space. Computing Ulam distance in MPC is more challenging, requiring the $(1+\approxFactMPCED)$-approximation algorithm of~\cite{hajiaghayi2019massively} with $\otilda{n^{1+\mpcSpaceRegime}}$ space. 

\vspace{-4mm}
\paragraph{A Sublinear Total Space Framework:}
We address the constraint on total memory through a sampling-based approach.
Rather than enumerating all possible candidates, which may be as many as $ \bigO{m^r} $, we show that it suffices to consider only $ \bigO{\log(n)} $ randomly selected input permutations and $ \bigO{\log(n)} $ local solutions derived from $ \bigO{\log(n)} $ random subsets of size $ r $. To identify the best candidates, instead of computing their exact costs, we leverage results from~\cite{indyk1999sublinear, indyk2001high}, which demonstrate that, with high probability, a $ (1+\delta) $-approximate solution to the optimal candidate can be found by selecting the permutation with the minimum cost over a random sample of $ \bigO{\log(n)/\distPara^{2}} $ input permutations.

Next, we outline how to compute a local solution of the set $Q$ under Ulam, Hamming, Spearman's footrule, and Kendall-tau distances using constantly many communication rounds.

\vspace{-4mm}
\paragraph{Ulam distance.} To employ our new algorithm for the Ulam distance in the \gls{mpc} setting, we proceed in three phases.

First, we distribute the task of constructing candidate blocks. For each group of five windows, we assign a specific machine to build the tournament graph, remove cycles, and output the resulting block. These results must then be aggregated to find the optimal combination of blocks. A direct approach would be to collect all candidate blocks onto a single machine to run the dynamic programming. However, this is not space-efficient, as the total size of these blocks may exceed the local memory of a single machine. To overcome this, we observe that the dynamic program does not need the actual strings; it only requires their objective values, lengths, and the starting and ending indices of the corresponding windows. This summary information is small enough to fit in the memory of one machine, allowing us to compute the optimal sequence of blocks efficiently.

The solution to the dynamic program implicitly defines a string consisting of $ n^{\mpcSpaceRegime} $ blocks. Each block is either a candidate block stored on one of the distributed machines or a logical block consisting entirely of dummy elements. By backtracking through the solution, we can identify the specific machines holding the chosen candidate blocks and assign new machines to generate the required blocks of dummy elements.

It remains to transform this distributed string into a valid permutation of length $ n $. We first truncate the string by removing the leftmost dummy elements until the total length is exactly $ n $. This is a simple counting task that can be performed using a broadcast tree of constant depth. Next, we must replace each remaining dummy element with a unique unused element. Our strategy is to pair the distinct unused elements with the remaining dummy positions according to their sorted order. That is, the $ k $-th smallest unused element should be placed in the $ k $-th available dummy position.

To implement this matching, we need to compute the global rank of each unused element and each dummy position. Since the data is distributed, each machine first counts the unused elements it holds locally. Then, using a broadcast tree, each machine computes the total count of unused elements residing on all preceding machines. This allows every machine to determine the exact global rank of its elements. We apply the same procedure to rank the dummy positions. Finally, the elements and positions with matching ranks are routed to the same machine to complete the permutation.

\vspace{-4mm}
\paragraph{Hamming distance.} In the algorithm for Hamming distance, (i) a majority element is assigned to each position if it exists among the permutations in $ \smallSubset $. (ii) Positions without majority agreement are filled arbitrarily with remaining elements. The first step can be simulated in \gls{mpc} similar to computing the Hamming distance: we bring the information from the $ i $'th block of each permutation in $Q$ to a single machine and decide the majority. 

The difficulty lies in the second step, as the list of unused elements and unassigned positions can be large and may not fit in a single machine. A simple approach is to assign unused elements in each $ i $ block to unassigned positions in the same block. However, this may not be feasible as the number of unused elements and unassigned positions in a block may differ. To resolve this, we consider the sorted list of unused elements and unassigned positions, and aim to pair them according to their order in the sorted lists. This requires computing the rank of each unused element relative to all unused elements, and the rank of each unassigned position relative to all unassigned positions. To implement this, we let each $ i $'th machine keep the unused elements in the $ i $'th block. The rank of each element can be computed if each machine knows the number of unused elements in all the previous blocks. This information can be made available to all machines by using a broadcast tree of depth $ \max( 1,\frac{2\mpcSpaceRegime-1}{1-\mpcSpaceRegime} ) $. We proceed similarly for unassigned positions. Finally, the elements and positions with the same rank are sent to the same machine.

\vspace{-4mm}
\paragraph{Spearman's Footrule.} In the algorithm for Spearman's footrule distance, we first compute a pseudo-permutation by assigning the median element at each position across the permutations in $ \smallSubset $. The output permutation is obtained by sorting the elements in the pseudo-permutation and reassigning ranks according to their sorted order. The first step can be implemented in \gls{mpc} similar to Hamming distance. The second step pose a similar challenge as in Hamming distance, since the elements are distributed across multiple machines. By sorting all elements in the pseudo-permutation, we reduce this problem to the pairing problem as in Hamming distance.

\vspace{-4mm}
\paragraph{Kendall-tau distance.} For Kendall-tau distance, recall that we need to apply the KWIK-SORT algorithm to solve the feedback arc set problem on the majority graph induced by $ \smallSubset $. The KWIK-SORT algorithm starts with a random vertex as pivot, and partitions the remaining vertices into two sets according to the direction of the edges between them and the pivot. We then recurse on the two sets, and concatenate the results. We follow a constant rounds \gls{mpc} implementation of KWIK-SORT presented in~\cite{im2020fast}. The approach to simulate this in constant \gls{mpc} rounds is to pick multiple pivots in each round. These pivots form a decision tree, which partitions the vertices into multiple sets. We then recurse on each set in parallel. However, if we directly follow the algorithm in~\cite{im2020fast}, we will violate the local space constraint. In particular,~\cite{im2020fast} allows the local space to be $ \otilda{n} $, hence, they can choose $ \bigO{n^{1/2}} $ pivots simultaneously in each round. Our local space may not be sufficient to accommodate this many vertices. We overcome this by choosing only $ \bigO{n^{1-\mpcSpaceRegime}} $ vertices as pivots. Due to the fact that the size of partitioned sets can not be too large~\cite[Lemma 7]{im2020fast}, as long as the number of pivots is $ \Omega(\log(n)) $, the number of rounds remains constant. Another challenge is that in~\cite{im2020fast}, the total space is $ \bigO{n^{2}} $, as they need to store the edges explicitly. We resolve this by observing that we can access the direction of an edge by accessing the positions of two elements in the permutations in $ \smallSubset $. This allows us to store the edges implicitly, keeping the total space at $ \bigO{n} $.
We provide the details of the MPC algorithms in~\cref{sec:mpc_implementation}.

%% file: framework.tex
\section{General Framework Analysis}\label{sec:framework}

In this section, we establish our framework and prove the components that are universal across all metric spaces. We begin by formalizing~\cref{property:universal} in the following lemma.

\begin{lemma}\label{lemma:property1}
Let $P$ be a finite set of points in a metric space $(\metricspace, \dist)$.
Assume that the number of points $p \in P$ for which $\costp{p} \leq (2-\alpha)\opt$ is at most $\delta m$, for $\delta \in [0,1-\alpha]$. Let $Q$ be a random subset of $P$ of size $r\geq 2$. Then, $\mathbb{E}_Q[\opttotalslack] \leq \binom{r}{2}(\alpha+2\delta)\opt$.
\end{lemma}

\begin{proof}
Let $F = \{ p \in P \mid \costp{p} \leq (2-\alpha)\opt \}$ be the set of points from $P$ with good approximation factor. From our assumption on the size of $F$, we have $|F| \leq \delta m$, so we can lower bound the total cost
\begin{equation*}
\sum_{p \in P} \costp{p} = \sum_{p \in F} \costp{p} + \sum_{p \in P \setminus F} \costp{p} \geq \sum_{p \in P \setminus F} \costp{p} \geq (2-\alpha)(1-\delta)m\opt.
\end{equation*}

On the other hand, we can express the same total cost by summing over all pairwise distances,
\begin{equation*}
\sum_{p \in P} \costp{p} = \sum_{p \in P} \frac{1}{m} \sum_{i=1}^m \dist(p, p_i) = \frac{1}{m} \sum_{\substack{i=1}}^m \sum_{\substack{j=1, \\ j \neq i}}^m \dist(p_i, p_j) .
\end{equation*}

Using the definition of triangle inequality slack, we can write,
\begin{align*}
\frac{1}{m} \sum_{\substack{i=1}}^m \sum_{\substack{j=1, \\ j \neq i}}^m \dist(p_i, p_j) &= \frac{1}{m} \sum_{\substack{i=1}}^m \sum_{\substack{j=1, \\ j \neq i}}^m  \left(\dist(x^*, p_i) + \dist(x^*, p_j) - \slackopt{i}{j}\right) \\
&= \frac{2}{m} \sum_{\substack{i=1}}^m \sum_{\substack{j=1, \\ j \neq i}}^m  \dist(x^*, p_i) - \frac{1}{m} \sum_{\substack{i=1}}^m \sum_{\substack{j=1, \\ j \neq i}}^m \slackopt{i}{j} \\
&= 2(m-1)\opt - \frac{1}{m} \sum_{\substack{i=1}}^m \sum_{\substack{j=1, \\ j \neq i}}^m \slackopt{i}{j}
\end{align*}

Combining both bounds,
\begin{equation*}
(2-\alpha)(1-\delta)m\opt \leq 2(m-1)\opt - \frac{1}{m} \sum_{\substack{i=1}}^m \sum_{\substack{j=1, \\ j \neq i}}^m\slackopt{i}{j}.
\end{equation*}

Rearranging,
\begin{equation*}
 \sum_{\substack{i=1}}^m \sum_{\substack{j=1, \\ j \neq i}}^m \slackopt{i}{j} \leq (m\alpha(1-\delta) + 2m\delta-2)m\opt \leq (m\alpha + 2m\delta-2)m\opt.
\end{equation*}

Since each pair $\{i,j\}$ with $i \neq j$ appears twice in the double sum, and $\slackopt{i}{i} = 0$,
\begin{equation*}
\sum_{1 \leq i < j \leq m} \slackopt{i}{j} \leq \frac{1}{2}(m\alpha + 2m\delta-2)m\opt.
\end{equation*}

For a subset $Q$ of size $r$ drawn uniformly at random from $P$, using linearity of expectation,
\begin{equation*}
\mathbb{E}_Q[\opttotalslack] = \mathbb{E}_Q\left[\sum_{\substack{i < j \\ p_i, p_j \in Q}} \slackopt{i}{j}\right] = \sum_{1 \leq i < j \leq m} \slackopt{i}{j} \cdot \Pr[p_i \in Q, p_j \in Q].
\end{equation*}

The probability that both $p_i$ and $p_j$ are selected in a random subset of size $r$ is
\begin{equation*}
\Pr[p_i \in Q, p_j \in Q] = \frac{\binom{m-2}{r-2}}{\binom{m}{r}}.
\end{equation*}

Therefore,
\begin{align*}
\mathbb{E}_Q[\opttotalslack] &= \frac{\binom{m-2}{r-2}}{\binom{m}{r}} \sum_{1 \leq i < j \leq m} \slackopt{i}{j} \leq \frac{\binom{m-2}{r-2}}{\binom{m}{r}} \cdot \frac{1}{2}(m\alpha + 2m\delta-2)m\opt \\
&= \frac{r(r-1)}{2(m-1)} \cdot (m\alpha + 2m\delta-2)\opt \leq \binom{r}{2}(\alpha+2\delta)\opt
\end{align*}

where the last inequality holds for $\delta \leq 1-\alpha$.
\end{proof}

We formalize~\cref{property:metric.specific} as we apply it in our analysis.

\begin{sameproperty}{property:metric.specific} [{\protect\hypersetup{linkcolor=black}\nameref{property:metric.specific}}]\label{property:formal.metric.specific}
    For our framework to apply to a metric space $(\metricspace, \dist)$, we require the existence of constants $C > 0$ and $r \geq 2$, and an algorithm that, given any subset $Q \subseteq P$ of size $r$, produces a solution $y$ such that $\dist(x^*, y) \leq C \cdot \opttotalslack$.
\end{sameproperty}

The challenge is to design such an algorithm and establish this property for each 
specific metric space.

We now combine both properties to design a general framework for the $1$-median problem in the sublinear regime.
Our framework operates by establishing that if only few points in $ P$ achieve a $(2-\alpha)$-approximation for some small constant $\alpha$, then there exist many constant-size subsets $Q \subseteq P$ with small total slack. Following~\cref{property:metric.specific} the algorithm applied on such $Q$ produce a point that is close to the optimal global median.

Hence, we construct a set of candidates $\candidates$ such that with high probability one of the candidates achieves a $(2-\alpha)$-approximation. We employ a sublinear approach that was developed by Indyk~\cite{indyk1999sublinear, indyk2001high} for the discrete $1$-median problem, where the medians are restricted to a specific set of points.
In~\cite{indyk2001high} it was shown that sampling $\bigO{\delta^{-2}\log n}$ points is sufficient to identify a candidate that achieves a $(1+\delta)$-approximation to the best median in the candidate set $\candidates$ with high probability.
See~\cref{alg.space.efficient.general.framework} for the complete pseudocode.

\newcommand{\estimateSet}{S}
\begin{algorithm}[htbp]
    \SetKwInOut{KwIn}{Input}
    \SetKwInOut{KwOut}{Output}
    \KwIn{Set $ \inpset \subseteq \metricspace$, constant $ r $, $\delta>0$}
    \KwOut{Approximate 1-median solution}
    $\candidates \gets $ a set of $ \bigO{\log {n}/\distPara} $ points sampled uniformly at random from $ \inpset $ \nllabel{line.sample.C}\;
    \For{ $ i \in \{1,2,\ldots, \bigO{\log{n}/\distPara} \} $\nllabel{line.sampleq}}{
        $ \smallSubset\gets $ a set of $ r $ permutations sampled uniformly at random from $ \inpset $\;
        $y \leftarrow \texttt{localAlg}(Q)$; \tcp{compute a local solution over $Q$} \nllabel{line.local.alg}
        $\candidates \leftarrow \candidates \cup \{y\}$
    \;
    }
    Sample a set $ \estimateSet $ of $ \bigO{\log{n/\distPara^{2}}} $ permutations uniformly at random from $ \inpset $
    \nllabel{line.costcalc}\;
    \Return{$ x\eqdef \argmin_{x\in \candidates}\cost{x}{\estimateSet} $}\nllabel{line.Indyk.return}\;
    \caption{A space efficient general framework for 1-median approximation.}
    \label{alg.space.efficient.general.framework}
\end{algorithm}

\begin{theorem}\label{thm:generalframework}
    Let $ \tdist{n}$ denote the time required to compute the distance between two permutations of length $n$.  Assume 
   ~\cref{property:metric.specific} holds with constants $C$ and $r$, and there is an algorithm that, given any subset $ \smallSubset\subseteq \inpset $ of size $ r $, in time $ \ltime{n} $ produces a solution $ y $ such that $ \dist(\gmedian,y) \leq C\cdot \opttotalslack $. Then~\cref{alg.space.efficient.general.framework} returns with high probability a $(2-\frac{1}{C\binom{r}{2}+1} + \bigO{\delta})$-approximation to the $ 1 $-median problem on $ \inpset $ using $ \otilda{\ltime{n} + \tdist{n}} $ time.
\end{theorem}

\begin{proof}

    \textbf{Approximation ratio analysis.} 
    Let $ \alpha = (C\binom{r}{2}+1)^{-1} $.
    If the number of points $p \in P$ for which $\costp{p} \leq (2-\alpha)\opt$ is at least $\delta m$, then w.h.p such a point is added to the set of candidates $\candidates$. Otherwise, by~\cref{lemma:property1}, expected total slack of $Q$ for every $Q$ sampled is at most $\binom{r}{2}(\alpha+2\delta)\opt$. 
    Using Markov inequality, in~\cref{line.sampleq}, with high probability, we sample a set $Q$ with at most $\binom{r}{2}(\alpha+3\delta)\opt$. The output $y$ of $\texttt{localAlg}(Q)$ is added to $\candidates$, furthermore, following~\cref{property:metric.specific}, $\dist (x^*,y) \leq C \opttotalslack \leq C\binom{r}{2}(\alpha+3\delta)\opt$.
    Then, by the triangle inequality,
    \begin{align*}
        \costp{y} \leq \frac{1}{m}\sum_{p \in P} (\dist(y,x^*)+\dist(x^*,p)) \leq & C\binom{r}{2}(\alpha+3\delta)\opt+\opt \\\leq (C\binom{r}{2}(\alpha+3\delta)+1)\opt.
    \end{align*}
    Using Indyk's sampling method and our selection of $\alpha$, we obtain~\cref{alg.space.efficient.general.framework} return a $(2-\alpha+\bigO{\delta})$-approximation.

    \textbf{Running time analysis.} 
   ~\cref{line.sampleq} repeats $ \bigO{\log(n)/\distPara} $ times, each executing $ \texttt{localAlg} $ in $ \ltime{n} $ time. 
   ~\cref{line.costcalc} computes the cost of $ \bigO{\log(n)/\delta} $ candidates to set $S$ of size $ \bigO{\log(n)/\delta^2} $ in $ \bigO{\log(n)^{2}\tdist{n}/\delta^3} $ time.  
    For a constant $\delta$, the total running time is then $\otilda{\ltime{n} + \tdist{n}} $.

\end{proof}

%% file: Ulam/local_algorithm.tex
\section{New Approximation Algorithm for Rank Aggregation under Ulam Metric}\label{sec:description.of.scalable.cycle.removal} 

In this section, we present our new approximation rank aggregation algorithm under Ulam metric. Using the framework we developed in~\cref{alg.space.efficient.general.framework}, it suffices to design a scalable algorithm on input $ Q $ consisting of $ r $ permutations, producing a permutation $ \offlineOut $ that is close to $ \gmedian $.

\ifarxiv
Our contribution is a new offline polynomial-time algorithm, $ \scalableCycleRemoval $, which we implement in the MPC model using $ \bigO{1} $ rounds (\cref{sec:mpc.Ulam}).
\else
Our contribution is a new offline polynomial-time algorithm, $ \scalableCycleRemoval $, which we implement in the MPC model using $ \bigO{1} $ rounds.
\fi
This algorithm satisfies a variant of~\cref{property:formal.metric.specific}, as stated in~\cref{lem:bound.gmedian.to.offlineOut}.

\begin{restatable}{lemma}{gmedianOfflineOut}\label{lem:bound.gmedian.to.offlineOut}
    Let $ Q = \{ \pi_{1}, \pi_{2}, \dots, \pi_{5} \} $ be a set of five permutations from $ \inpset $. The output $ \offlineOut $ of $ \scalableCycleRemoval $ with input $ Q $ satisfies
    \begin{align}
        \ed{\offlineOut, \gmedian} \leq 0.0009 \sum_{i=1}^{5}\ed{\gmedian, \pi_{i}} + 266 \totalslack{\gmedian}. \label{eq:scalable.cycle.removal}
    \end{align}
\end{restatable}

Additionally to \cref{lem:bound.gmedian.to.offlineOut}, we require the following structural lemma which, which extends Lemma~\ref{lemma:property1} with additional properties that we exploit in the Ulam case.

\begin{lemma}\label{lemma:structural}
    Let $\alpha \leq 1$ and $\delta \leq \frac{\alpha}{100}$.
    Let $P$ be a finite set of $m$ points in a metric space $(\metricspace, \dist)$, 
    and assume that $
    |\{p \in P : \costp{p} \leq (2-\alpha)\opt\}| \leq \delta m$.
    Let $Q$ be a uniformly random subset of $P$ of size $5$. 
    Then with probability at least $0.008$, both of the following hold:
    \begin{enumerate}[label=\textup{(P\arabic*)}]
        \item $\dist(q, \gmedian) \leq (1+\alpha)\opt$ for all $q \in Q$, and \label{enu.five.prop.middle}
        \item $\totalslack{\gmedian} \leq 510\alpha \cdot \opt$. \label{enu.upper.bound.total.slack}
    \end{enumerate}
\end{lemma}

\begin{proof}
Let 
\begin{align}
    F_0 &= \{p \in P \mid \mathrm{cost}(p,P) < (2-\alpha)\opt\}, \nonumber \\ 
    F_1 &= \left\{p \in P \mid (1-\alpha)\opt \leq \dist(\gmedian, p) \leq (1+\alpha)\opt\right\}. \nonumber
\end{align}
Recall $|F_0| = \delta m$ and further denote $|F_1| = \gamma m$ for some $\gamma \in [0,1]$. Then,
\begin{align*}
    \opt &= \frac{1}{m}\sum_{p \in P} \dist(\gmedian,p) \\
    &\geq \frac{1}{m}\sum_{p \in F_1} \dist(\gmedian,p) + \frac{1}{m}\sum_{p \in P \setminus (F_0 \cup F_1)} \dist(\gmedian,p) \\
    &\geq \gamma(1-\alpha)\mathrm{OPT} + (1-\gamma-\delta)(1+\alpha)\mathrm{OPT}.
\end{align*}
Solving for $\gamma$, we obtain
\[
\gamma \geq \frac{1}{2} - \frac{(1+\alpha)\delta}{2\alpha} \geq \frac{1}{2}-\frac{1}{100},
\]
where the last inequality follows from our assumptions on $\alpha$ and $\delta$.

Consider a random subset $Q$ of five permutations chosen uniformly from $P$. Let $B$ denote the bad event that $Q$ does not satisfy both $\totalslack{\gmedian} \leq 500(\alpha+2\delta)\opt$ and $Q \subseteq F_1$. 

For the first condition, by Lemma~\ref{lemma:property1} and Markov's inequality,
\[
\mathbb{P}\!\left[\totalslack{\gmedian} > 500(\alpha+2\delta)\mathrm{OPT}\right] \leq \tfrac{1}{50}.
\]
Since $\delta \leq \frac{\alpha}{100}$, we have $500(\alpha+2\delta) \leq 510\alpha$, and therefore $
\mathbb{P}\!\left[\totalslack{\gmedian} > 510\alpha\cdot \mathrm{OPT}\right] \leq \tfrac{1}{50}$.

For the second condition, the probability that not all five sampled points are from $F_1$ is $
1-\gamma^5 \leq 1-\left(\frac{1}{2}-\frac{1}{100}\right)^5 < 0.972$.
By the union bound, $\mathbb{P}[B] \leq \frac{1}{50} + 0.972 = 0.992$.
\end{proof}

\begin{restatable}{theorem}{offlineulamresult} \label{thm:offline.main.theorem}
    For any constant $ \distPara>0 $, there is a  polynomial time algorithm that, given a set of $ m $ permutations $ P \subseteq \permutations $, with high probability computes a $ ( 2-\alpha + \bigO{\distPara}) $ approximation of $ 1 $-median under Ulam distance, where $ \alpha> 0 $ is a constant.
\end{restatable}

\begin{proof}[Proof of~\cref{thm:offline.main.theorem}]
    We utilize the framework of~\cref{alg.space.efficient.general.framework} with $ r=5 $ and in place of $ \texttt{localAlg} $ in~\cref{line.local.alg}, we use our algorithm $ \scalableCycleRemoval $. Our proof relies on~\cref{lemma:structural}, which describes a structural property of the input set under Ulam distance. 
    
    Fix $ \alpha = 0.000007 $.
    If the number of permutations in $ p\in \inpset $ for which $ \cost{p}{\inpset} \leq (2-\alpha)\opt $ is at least $ \delta m $, then w.h.p, such a permutation is included in $ \candidates $ in~\cref{line.sample.C} of~\cref{alg.space.efficient.general.framework}.

    Otherwise, \cref{lemma:structural} implies that a random set $ Q $ of $ 5 $ permutations satisfies properties~\ref{enu.five.prop.middle} and~\ref{enu.upper.bound.total.slack} with probability at least $ 0.008 $. As we sample $ \bigO{\log n} $ sets in~\cref{line.sampleq}, w.h.p., such a set $ Q $ is included in the samples.

    In this case, the candidate set $ \candidates $ contains a permutation
    \[
        y = \scalableCycleRemoval(\pi_{1}, \dots, \pi_{5}).
    \]
    From~\cref{lem:bound.gmedian.to.offlineOut} and our choice of $ \alpha $, we have:
    \begin{align}
        \ed{y, \gmedian} &\leq 0.0009 (1+\alpha)\opt + 266* 510 \alpha \opt \nonumber \leq 0.96 \opt. \nonumber
    \end{align}
    Consequently, by the triangle inequality:
    \begin{align}
        \cost{y}{\inpset} \leq \dfrac{1}{m}\left( \sum_{i=1}^{m} \left( \ed{y, \gmedian} + \ed{\gmedian, p_{i}} \right) \right) \leq \ed{y, \gmedian} + \opt \leq 1.96\opt. \nonumber
    \end{align}

    In all cases, the candidate set $ \candidates $ contains a permutation achieving a $ (2-\alpha) $-approximation to the optimal $ 1 $-median. Choosing a candidate from $ \candidates $ using Indyk's sampling method~\cite{indyk2001high} (lines~\ref{line.costcalc}--\ref{line.Indyk.return} of~\cref{alg.space.efficient.general.framework}) 
    returns a permutation with cost at most $(2-\alpha + \bigO{\distPara})\opt $. 

\ifarxiv
As the algorithm $ \scalableCycleRemoval $ runs in time polynomial in $ n $ (complete running time analysis provided in Section~\ref{sub:postprocessing}), it is clear that the overall algorithm runs in polynomial time.
\else
As the algorithm $ \scalableCycleRemoval $ runs in time polynomial in $ n $, it is clear that the overall algorithm runs in polynomial time.
\fi
    
\end{proof}
The remainder of this section describes the algorithm $\scalableCycleRemoval$.

\subsection{Overview of ScalableMedianReconstruct Algorithm}\label{sec:overview_of_scalablepermutationreconstruct_algoirthm} 

Our objective is to compute a good estimation of the median of $ Q $, defined as \[\lmedian = \argmin_{\genStr \in \permutations}\sum_{i=1}^{5}\ed{\pi_{i}, \genStr}. \] 
To outline the algorithm, we first fix an optimal alignment between $ \lmedian $ and each input permutation $ \pi_{i} $ for $ i=1, \dots, 5 $. Consider a partition of $ \lmedian $ into $ n^{\mpcSpaceRegime} $ disjoint contiguous blocks $ \lmedian[\ell_{j}, r_{j}) $ of size $ \blSize = n^{1-\mpcSpaceRegime} $. The fixed alignments induce a corresponding decomposition of each $ \pi_{i} $ into substrings $ \pi_{i}[\optStartPoint{i,j}, \optEndPoint{i,j}) $, such that the block $ \lmedian[\ell_{j}, r_{j}) $ aligns with $ \pi_{i}[\optStartPoint{i,j}, \optEndPoint{i,j}) $. 
We refer to each $ \pi_{i}[\optStartPoint{i}, \optEndPoint{i,j}) $ as a \emph{constructive block}, and to $ \{ \pi_{i}[\optStartPoint{i,j}, \optEndPoint{i,j}) \}_{i=1}^{5} $ as a \emph{constructive group}. 
Each block $\lmedian[\ell_{j}, r_{j})$ can be approximately reconstructed using the corresponding constructive group $\{\pi_{i}[\optStartPoint{i,j}, \optEndPoint{i,j})\}_{i=1}^{5}$. We refer to these approximations as \emph{candidate blocks}.

The challenge lies in efficiently identifying these constructive blocks and combining the candidate blocks to obtain an overall approximation of $y^*$. The $\scalableCycleRemoval$ algorithm accomplishes this through the following four steps:

\ifarxiv
\begin{itemize}
    \item\texttt{\nameref*{sub:windows_decomposition}}: Each string $\pi_i \in Q$ is decomposed into windows, which serve as \emph{constructive blocks} for building the blocks of $y^*$. We then form \emph{constructive groups} by selecting one window from each string using a structured method, rather than enumerating all combinations. The details are provided in~\cref{sub:windows_decomposition}. 
    \item\texttt{\nameref*{sub:local_aggregating_algorithm}}: For each constructive group, a \emph{candidate block} is computed using the $ \localAggregation $ algorithm (\cref{alg.local.cycle.removal}). These candidate blocks serve as primitives for the next step. The details are provided in~\cref{sub:local_aggregating_algorithm}.
    \item\texttt{\nameref*{sub:combining_blocks}}: The candidate blocks from the previous step are combined to form an \emph{intermediate string} $ \offlineIntS $ via a dynamic programming approach. The intermediate string $ \offlineIntS $ consists of $ n^{\mpcSpaceRegime} $ blocks, each being either a string created in the $ \localAggregation $ step or a string of dummy elements of length $ \blSize $. The details are provided in~\cref{sub:combining_blocks}.
    \item\texttt{\nameref*{sub:postprocessing}}: The intermediate string $ \offlineIntS $ has length at least $ n $ and contains no repeated elements, but it may not be a permutation due to missing elements and the presence of dummy elements. To transform $ \offlineIntS $ into a permutation $ \offlineOut $, we replace the dummy elements with the unused elements. The details are provided in~\cref{sub:postprocessing}. 
\end{itemize}
\else
\begin{itemize}
    \item\texttt{Windows Decomposition}: Each string $\pi_i \in Q$ is decomposed into windows, which serve as \emph{constructive blocks} for building the blocks of $y^*$. We then form \emph{constructive groups} by selecting one window from each string using a structured method, rather than enumerating all combinations.
    \item\texttt{Block Reconstruction}: For each constructive group, a \emph{candidate block} is computed using the $ \localAggregation $ algorithm. These candidate blocks serve as primitives for the next step.
    \item\texttt{Block Decomposition via Dynamic Programming}: The candidate blocks from the previous step are combined to form an \emph{intermediate string} $ \offlineIntS $ via a dynamic programming approach. The intermediate string $ \offlineIntS $ consists of $ n^{\mpcSpaceRegime} $ blocks, each being either a string created in the $ \localAggregation $ step or a string of dummy elements of length $ \blSize $.
    \item\texttt{PostProcessing}: The intermediate string $ \offlineIntS $ has length at least $ n $ and contains no repeated elements, but it may not be a permutation due to missing elements and the presence of dummy elements. To transform $ \offlineIntS $ into a permutation $ \offlineOut $, we replace the dummy elements with the unused elements.
\end{itemize}
\fi

Throughout our algorithm, we use a small constant $ \edErr >0 $, the value of which will be specified later. We denote the dummy element by $ \specialChar $.

\ifarxiv
\subsection{Windows Decomposition}\label{sub:windows_decomposition} 

In this phase, we decompose the interval $ [1,n+1) $ into a set of windows $ \windowSet $. Each window is defined as a pair of start and end positions $ \pair{ \startPoint, \endPoint } $ such that $ 1\leq \startPoint \leq \endPoint \leq n+1 $. We aggregate these windows into a set $ \windTupSet \subseteq \windowSet^{5} $ consisting of groups of five windows. Based on this set, we refer to a collection of five substrings $ \{ \pi_{i}[\startPoint{i}, \endPoint{i}) \}_{i=1}^{5} $ as a \emph{constructive group} if $ \{ \pair{ \startPoint{i}, \endPoint{i} } \}_{i=1}^{5} $ belongs to $ \windTupSet $.

We decompose $ [1,n+1) $ into $ \windowSet = \bigcup_{j=1}^{n^{\mpcSpaceRegime}}\windowSet{j} \cup \bigcup_{j=1}^{n^{\mpcSpaceRegime}}\sWindowSet{j} $, where $ \windowSet{j} $ and $ \sWindowSet{j} $ are constructed as follows.
For each $ \upperBoundED $ such that $ n^{\upperBoundED} $ belongs to the set $ \{ 1, 1+ \edErr, (1+\edErr)^{2}, \dots, (1+\edErr)^{\ceil{\log_{1+\edErr}(2n)}} \} $, we select start positions $ \startPoint $ from the range $ [ \ell_{j} - n^{\upperBoundED}, \ell_{j} + n^{\upperBoundED} ] $ that are divisible by $ \gap = \edErr n^{\upperBoundED - \mpcSpaceRegime} $. The set $ \windowSet{j} $ consists of windows starting at such positions $ \startPoint $. For each $ \startPoint $, the end positions are defined as:
\begin{align}
    \endPoint = 
    \begin{cases}
    \startPoint + \blSize\\
    \startPoint + \blSize + (1+\edErr)^{a}, & 0\leq a \leq \log_{1+\edErr}\min(\blSize/\edErr, n^{\upperBoundED}) \\
    \startPoint + \blSize - (1+\edErr)^{a}, & 0 \leq a \leq \ceilenv{\log_{1+\edErr}(\blSize - \edErr\blSize)}
    \end{cases} \label{eq:mpc_ulam_end_point_choices}
\end{align}
The set $ \sWindowSet{j} $ is a collection of degenerate windows of the form $ [\startPoint, \startPoint) $, where $ \startPoint $ is selected from the same set of start positions as above.

We now describe how to group these windows to form $ \windTupSet $, thereby inducing the constructive groups.
\paragraph{A Simple Grouping Strategy.}\label{par:a_simple_windows_decomposition}
A simple approach is to construct $ \windTupSet = \windowSet^{5} $. Observe that $ \card{\windowSet{j}} = \bigO{n^{\mpcSpaceRegime}\log^{2}n} $ and $ \card{\sWindowSet{j}} = \bigO{n^{\mpcSpaceRegime}} $. Consequently, the total size of the window set is $ \card{\windowSet} = \bigO{n^{2\mpcSpaceRegime}\log^{2}n} $, leading to $ \card{\windTupSet} = \bigO{n^{10\mpcSpaceRegime}\log^{10}n} $.

\paragraph{A Selective Grouping Strategy.}
To reduce the size of the set, we employ a more selective method. We define $ \windTupSet $ as the union of sets $ \windTupSet{1}, \windTupSet{2}, \dots, \windTupSet{ n^{\mpcSpaceRegime} } $, where each $ \windTupSet{j} $ is constructed as:
\begin{align}
    \windTupSet{j} = \windowSet{j}^{5} \cup \bigcup_{k=1}^{5} \left( \windowSet{j}^{k-1} \times \sWindowSet \times \windowSet{j}^{5-k} \right). \label{eq.def.wind.tup.set.j}
\end{align}
Here, $ \sWindowSet = \bigcup_{j} \sWindowSet{j} $ denotes the set of all degenerate windows. Since $ \card{\sWindowSet} = \bigO{n^{2\mpcSpaceRegime}} $, the size of each $ \windTupSet{j} $ is bounded by $ \bigO{n^{6\mpcSpaceRegime}\log^{8}n} $. Therefore, the total size is reduced to $ \card{\windTupSet} = \bigO{n^{7\mpcSpaceRegime}\log^{8}n} $.

\subsection{Block Reconstruction}\label{sub:local_aggregating_algorithm} 

For every constructive group $ \{  \pi_{i}[\startPoint{i}, \endPoint{i} )   \}_{i=1}^{5} $, we apply \texttt{\nameref*{sub:local_aggregating_algorithm}} algorithm (\cref{alg.local.cycle.removal}) with input $ \{ \pi_{i}[\startPoint{i}, \endPoint{i}) \}_{i=1}^{5} $ to construct a candidate block. This candidate block, along with the corresponding input intervals, is collected into a set $ C_{j} $. Formally, for each $ j \in \{1, \dots, n^{\mpcSpaceRegime}\} $,

\begin{align}
    C_{j} := &\>\>\>\>\>\left\{ \biggl( \localAggregation\Bigl(\bigl\{ \pi_{i}[\startPoint{i}, \endPoint{i}) \bigr\}_{i=1}^{5}\Bigr), \bigl\{ \pair{ \startPoint{i}, \endPoint{i} } \bigr\}_{i=1}^{5} \biggr)| \bigl\{ \pair{ \startPoint{i}, \endPoint{i} } \bigr\}_{i=1}^{5} \in \windTupSet{j} \right\}. \nonumber
\end{align}

The \texttt{\nameref*{sub:local_aggregating_algorithm}} algorithm proceeds as follows. First, we construct a directed graph $ G=(V,E) $ where the vertex set $ V $ consists of all elements appearing in at least four of the five input substrings. For any pair of distinct vertices $ a, b \in V $, we add a directed edge $ (a, b) $ to $ E $ if and only if $ a $ precedes $ b $ in a strictly greater number of substrings than $ b $ precedes $ a $. 

Next, we prune $ G $ to obtain an acyclic tournament. We first iteratively remove any pair of vertices $ ( a, b ) $ that is not connected by an edge. Once every remaining pair is connected, the graph becomes a tournament. Since a tournament contains cycles if and only if it contains triangles, we iteratively identify triangles and remove their vertices until the graph is acyclic. Finally, we extract a string $ \tp $ from the topological ordering of the remaining vertices. If the length of $\tp $ is less than $ \blSize $, we append $ \specialChar $ elements to extend it to length $ \blSize $. The complete procedure is detailed in~\cref{alg.local.cycle.removal}.

\begin{algorithm}[htbp]
    \SetKwInOut{KwIn}{Input}
    \SetKwInOut{KwOut}{Output}
    \KwIn{A constructive group $ \{ \pi_{i} \}_{i=1}^{5} $.}
    \KwOut{A candidate block $ \tp $.}
    \SetKwProg{procedure}{procedure}{}{}
    \SetKwFunction{BlockReconstruction}{BlockReconstruction}
    \procedure{\BlockReconstruction{$ \pi_{1}, \pi_{2}, \pi_{3}, \pi_{4}, \pi_{5} $}}{
        $ G \gets (V, E) $ where\\
        \hspace*{2.5em}$ V = \{ a \in \Sigma \mid a \text{ appears in at least } 4 \text{ strings} \}, $ \nllabel{line.majority.graph}
        \\
        \hspace*{2.5em}$ E = \{ (a, b) \mid \text{count}(a \text{ before } b) > \text{count}(b \text{ before } a) \} $

    \nllabel{line.starts.removing}
    \While{there exists a pair of vertices $ \{a,b\} \subseteq V $ with no edge between them}{
        $ V\gets V\setminus\{ a,b \} $
    }
    \While{there exists a triangle $ C $ in $ G $}{
        $ V\gets V\setminus V(C) $ \nllabel{line.ends.removing} 
    }
        $ \tp\gets $ topological sort of $ V $ \nllabel{line.topological.sort}

    \uIf{$ \card{\tp} < \blSize $}{
        \KwRet{ $ \tp\specialChar^{\blSize-\card{\tp}} $ } \tcp{pad with $ \specialChar $ to get a string of length $ \blSize $} \nllabel{line.return.padded.string}
    }

    \KwRet{$ \tp $} \nllabel{line.already.long.string}
}
\caption{\texttt{\nameref*{sub:local_aggregating_algorithm}} algorithm}
\label{alg.local.cycle.removal}
\end{algorithm}

\newcommand{\aggSet}{\mathcal{D}} 
\newcommand{\blEd}{\mathrm{BlockED}} 
\subsection{Block Composition via Dynamic Programming}\label{sub:combining_blocks} 
We construct an intermediate string $ \offlineIntS $ by combining substrings constructed from the step \texttt{\nameref*{sub:local_aggregating_algorithm}}. The string $ \offlineIntS $ consists of $ n^{\mpcSpaceRegime} $ blocks, where each $ j $-th block is either a string from a tuple $ \tuple{a} \in C_{j} $, or a block consisting of $ \blSize $ dummy elements. We determine the specific composition of $ \offlineIntS $ by solving a minimization problem over the search space induced by the sets $ C_{1}, C_{2}, \dots, C_{ n^{\mpcSpaceRegime} } $.
To formalize this optimization process, we first specify the domain of valid solutions.
\paragraph{Search space.}\label{par:search_space} 
We define the \emph{search space} as the set of all strings induced by a \emph{valid} sequence of tuples.
Formally, for all $ 1 \leq k \leq n^{\mpcSpaceRegime} $ and indices $ 1 \leq j_{1} < j_{2} < \dots < j_{k} \leq n^{\mpcSpaceRegime} $, consider a sequence of tuples $ (\tuple{a_{1}}, \tuple{a_{2}}, \dots, \tuple{a_{k}}) $, where $ \tuple{a_{t}} \in C_{j_{t}} $. We call this sequence \emph{valid} if for all $ 1 \leq t < k $, denoting $ \tuple{a_{t}} = ( \offlineLocalOut{t}, \{ \pair{ \startPoint{i,t}, \endPoint{i,t} } \}_{i=1}^{5} ) $ and $ \tuple{a_{t+1}} = ( \offlineLocalOut{t+1}, \{ \pair{ \startPoint{i,t+1}, \endPoint{i,t+1} } \}_{i=1}^{5} ) $, we have $ \endPoint{i,t} \leq \startPoint{i,t+1} $ for all $ i=1, \dots, 5 $.

Any such valid sequence uniquely defines a string of $ n^{\mpcSpaceRegime} $ blocks, denoted as $ \offlineIntS{\tuple{a_{1}}, \dots, \tuple{a_{k}}} = \blIntS{1}\blIntS{2}\dots\blIntS{n^{\mpcSpaceRegime}} $. The block $ \blIntS{j_{t}} $ is the string component of the tuple $ \tuple{a_{t}} = (\offlineLocalOut{t}, \{ \pair{ \startPoint{i,t}, \endPoint{i,t} } \}_{i=1}^{5}) $, i.e., $ \blIntS{j_{t}} = \offlineLocalOut{t} $. Any block $ \blIntS{h} $ where $ h \notin \{j_{1}, \dots, j_{k}\} $ is filled with $ \blSize $ $ \specialChar $ elements. This amounts to inserting $ (j_{1}-1)\blSize $ elements before block $ \blIntS{j_{1}} $, and $ (j_{t+1}-j_{t}-1)\blSize $ elements between blocks $ \blIntS{j_{t}} $ and $ \blIntS{j_{t+1}} $.

Furthermore, this valid sequence defines a specific decomposition for each permutation $ \pi_{i} $. The indices of $ \pi_{i} $ are partitioned into:
\begin{align}
    [1, \startPoint{i,1}),\ [\startPoint{i,1}, \endPoint{i,1}),\ [\endPoint{i,1}, \startPoint{i,2}),\ [\startPoint{i,2}, \endPoint{i,2}), \dots,\ [\endPoint{i,k}, n+1). \nonumber
\end{align}
This structure imposes an alignment between the constructed string $ \offlineIntS $ and each $ \pi_{i} $: the explicit blocks $ \blIntS{j_{t}} $ align with the intervals $ \pi_{i}[\startPoint{i,t}, \endPoint{i,t}) $, while the segments of $ \specialChar $ elements align with the prefix $ \pi_{i}[1,\startPoint{i,1}) $, the substrings $ \pi_{i}[\endPoint{i,t}, \startPoint{i,t+1}) $, and the suffix $ \pi_{i}[\endPoint{i,k}, n+1) $. We illustrate this alignment in~\cref{fig:induced.alignment}.

\begin{figure}
    \centering
    \includegraphics[width=\textwidth]{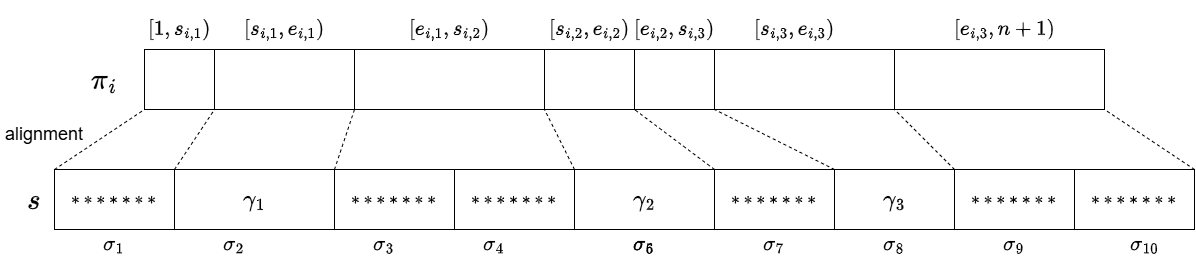}
    \caption{A valid sequence of tuples $ \tuple{a_{1}} \in C_{2} $, $ \tuple{a_{2}} \in C_{6} $, and $ \tuple{a_{3}} \in C_{8} $ induces a string $ \concateStr = \blIntS{1}\blIntS{2}\dots\blIntS{10} $. This structure defines a decomposition of each input permutation $ \pi_{i} $ and establishes the alignment between $ \concateStr $ and $ \pi_{i} $.}
    \label{fig:induced.alignment}
\end{figure}

\paragraph{Objective Function.}\label{par:objective_function}

Let $ \blEd(\concateStr) $ be the sum of distances between the blocks of $ \concateStr $ and the corresponding intervals of each $ \pi_{i} $ defined by this decomposition. Since this represents the cost of one specific alignment, $ \blEd(\concateStr) $ serves as an upper bound on the distance $ \sum_{i=1}^{5}\ed{ \concateStr, \pi_{i} } $. Formally,

\begin{align}
    \blEd(\concateStr) = &\>\>\>\>\>\sum_{t=1}^{k} \sum_{i=1}^{5} \ed{ \blIntS{j_t}, \pi_{i}[\startPoint{i,t}, \endPoint{i,t} ) } \label{eq.ed.mapped.constructed.interval} \\
                          &+ \sum_{i=1}^{5}\left((j_{1}-1)\blSize + \startPoint{i,1}- 1\right) \label{eq.ed.before.first.constructed.interval} \\ 
                          &+ \sum_{t=1}^{k-1} \left( 5(j_{t+1}-j_{t}-1)\blSize + \sum_{i=1}^{5} (\startPoint{i,t+1} - \endPoint{i,t}) \right) \label{eq.ed.inbetween.constructed.interval} \\ 
                          &+ \sum_{i=1}^{5} \biggl( (n^{\mpcSpaceRegime} - j_{k})\blSize + n + 1 - \endPoint{i,k} \biggr). \label{eq.ed.after.last.constructed.interval}
\end{align}
The components of this equation correspond to the alignment segments as follows:
\begin{itemize}
    \item The term~\eqref{eq.ed.mapped.constructed.interval} is the sum of local edit distances between each selected block $ \blIntS{j_{t}} $ and the input interval $ \pi_{i}[\startPoint{i,t}, \endPoint{i,t}) $.
    \item The term~\eqref{eq.ed.before.first.constructed.interval} is the distance between the prefix $ \blIntS{1}\dots \blIntS{j_{1}-1} $ of $ \concateStr $ (consisting of $ \specialChar $ elements) and the prefix $ \pi_{i}[1, \startPoint{i,1}) $.
    \item The term~\eqref{eq.ed.inbetween.constructed.interval} is the distance between the gap substrings $ \blIntS{j_{t}+1}\dots \blIntS{j_{t+1}-1} $ of $ \concateStr $ and the intervals $ \pi_{i}[\endPoint{i,t}, \startPoint{i,t+1}) $.
    \item The term~\eqref{eq.ed.after.last.constructed.interval} is the distance between the suffix $ \blIntS{j_{k}+1}\dots\blIntS{n^{\mpcSpaceRegime}} $ of $ \concateStr $ and the suffix $ \pi_{i}[\endPoint{i,k}, n+1) $.
\end{itemize}

Finally, the string we construct is the solution to the following optimization problem over the defined search space:
\begin{align}
    \min_{\concateStr} \blEd(\concateStr), \label{eq.ed.roundtwo.objective}
\end{align}
where the minimization is performed over all strings $\concateStr$ induced by valid sequences of tuples.

\paragraph{Dynamic Programming Solution.}\label{par:dynamic_programming_solution} 
The minimization problem~\eqref{eq.ed.roundtwo.objective} can be solved using dynamic programming as follows.

For convenience, for a tuple $ \tuple{ a } = (\offlineLocalOut, \{ \pair{ \startPoint{i}, \endPoint{i} } \}_{i=1}^{5}) $, we denote
\[
    \ed_{\tuple{ a }} = \sum_{i=1}^{5}\ed{\offlineLocalOut, \pi_{i}[\startPoint{i}, \endPoint{i})}.
\]

Let $ D[\tuple{ a }] $ be the minimum of $ \blEd( \concateStr ) $, taking over all $ \concateStr $ induced from sequence of tuples in which $ \tuple{ a } $ is the last tuple. The update rule for $ D[\tuple{ a }] $ is given by
\begin{align}
    D[\tuple{ a }] = \min\biggl\{ &\sum_{i=1}^{5}(\startPoint{i}-1) + 5(j-1)\blSize + \ed_{\tuple{ a }}, \nonumber \\
                        &\min_{\substack{\tuple{ b }\in C_{ h }, h < j\\ \tuple{ b } = \left( \offlineLocalOut', \{ \pair{ \startPoint{i}', \endPoint{i}' } \}_{i=1}^{5} \right) \\ \endPoint{i}' \leq \startPoint{i}, \forall i=1, \dots , 5}} \Bigl\{ D[\tuple{ b }] + \sum_{i=1}^{5}(\startPoint{i} - \endPoint{i}') + 5(j-h-1)\blSize + \ed_{\tuple{ a }}  \Bigr\}\biggr\} \nonumber\\
                        & + \sum_{i=1}^{5}(n + 1 -\endPoint{i}^{\tuple{ a }}) + 5 (n^{\mpcSpaceRegime} - j)\blSize. \label{eq.mpc.ulam.dp.update}
\end{align}

Finally, it is clear that the solution to~\eqref{eq.ed.roundtwo.objective} is given by
\[
        \min_{\concateStr} \blEd(\concateStr)= \min_{\substack{1\leq j \leq n^{\mpcSpaceRegime}, \\ \tuple{ a }\in C_{ j }}} D[\tuple{ a }].
\]

In order to trace back the solution, we employ an array $ P[\tuple{ a }] $ to keep track of the previous solution, that is, $ P[\tuple{ a }] = \tuple{ b } $ if $ D[\tuple{ a }] $ is computed from $ D[\tuple{ b }] $. Let $ \offlineIntS $ be the string recovered by backtracking through these pointers. It is evident from the construction of valid sequences of tuples that $ \offlineIntS $ has length at least $ n $ and contains no repetitions.
We provide a formal description for this phase in~\cref{alg.combining.blocks}.
\newcommand{\answer}{\mathtt{answer}}
\begin{algorithm}[htbp]
    \SetKwInOut{KwIn}{Input}
    \KwIn{$ C_{ 1 }, C_{ 2 }, \dots, C_{n^{\mpcSpaceRegime}} $.}
    \For{$ j=1 \to n^{\mpcSpaceRegime} $}{ \nllabel{line.dp.start.first.loop}
            \For{$ \tuple{ a }\in C_{ j } $, $ \tuple{a} = (\offlineLocalOut, \{ \pair{ \startPoint{i}, \endPoint{i} } \}_{i=1}^{5} ) $ }{
                $ \ed_{\tuple{a}} \gets \sum_{i=1}^{5}\ed{\offlineLocalOut, \pi_{i}[\startPoint{i}, \endPoint{i})} $

                $ D[\tuple{ a }] \gets \sum_{i=1}^{5} (\startPoint{i} - 1) + 5(j-1)\blSize + \ed_{\tuple{ a }}  $

                $ P[\tuple{a}] \gets \emptyset $

                \For{$ h = 1\to j-1 $}{
                    \For{$ \tuple{b}\in C_{ h } $, $ \tuple{b} = (\offlineLocalOut' , \{ \pair{ \startPoint{i}', \endPoint{i}' } \}_{i=1}^{5} $}{
                        \uIf{$ \endPoint{i}' \leq \startPoint{i} $, for all $ i=1,2, \dots 5 $, and $ D[\tuple{ a }] > D[\tuple{ b }] + \sum_{i=1}^{5} (\startPoint{i} - \endPoint{i}') + 5(j - h - 1)\blSize  + \ed_{\tuple{ a }} $ \nllabel{line.dp.inner.first.loop}
                        }{
                            $ D[\tuple{ a }] \gets D[\tuple{ b }] + \sum_{i=1}^{5} (\startPoint{i} - \endPoint{i}') + 5(j - h - 1)\blSize + \ed_{\tuple{ a }}  $   

                            $ P[\tuple{ a }] \gets \tuple{ b } $
                        }
                    }
                }
            }
        } \nllabel{line.dp.end.first.loop}

        $ \answer.D \gets \infty $

        \For{$ j=1 \to n^{\mpcSpaceRegime} $}{
            \For{$ \tuple{ a } \in C_{ j } $}{
                \uIf{$ \answer.D < D[\tuple{ a }] + \sum_{i=1}^{5} (n + 1- \endPoint{i}^{\tuple{ a }}) + 5 (n^{\mpcSpaceRegime} - j)\blSize $}{
                    $ \answer.D \gets D[\tuple{ a }] + \sum_{i=1}^{5} (n + 1- \endPoint{i}^{\tuple{ a }}) + 5 (n^{\mpcSpaceRegime} - j)\blSize $

                    $ \answer.P \gets \tuple{ a } $
                }
            }
        }

    \KwRet{$ \answer $}
    \caption{Block Composition via Dynamic Programming}
    \label{alg.combining.blocks}
\end{algorithm}

\subsection{PostProcessing}\label{sub:postprocessing} 
This step transforms the intermediate string $ \offlineIntS $, obtained from \texttt{\nameref*{sub:combining_blocks}} step, into the final permutation $ \offlineOut $.

Interestingly, a naive strategy, simply truncating $ \offlineIntS $ to length $ n $ and replacing remaining $ \specialChar $ elements with unused characters, suffices to preserve the desired approximation ratio.
However, to ensure scalability in the MPC model, we follow a different strategy. First, we scan $ \offlineIntS $ from left to right, removing $ \specialChar $ elements until the length reduces to $ n $. Let $ \offlineIntS' $ be the resulting string. Next, we replace each remaining $ \specialChar $ element in $ \offlineIntS' $ with distinct unused elements from $ \Sigma $ in increasing order. This yields the final output permutation $ \offlineOut $.

\paragraph{Runtime Analysis.}
We finish this section by arguing that $ \scalableCycleRemoval $ algorithm runs in time polynomial in $ n $. In step~\texttt{\nameref*{sub:windows_decomposition}}, constructing the set $ \windTupSet $ of size $ \bigO{n^{7\mpcSpaceRegime}\log^{8}n} $ requires $ \bigO{n^{7\mpcSpaceRegime}\log^{8}n} $ time. The subsequent steps, \texttt{\nameref*{sub:local_aggregating_algorithm}} and~\texttt{\nameref*{sub:combining_blocks}}, run in time polynomial in the size of $ \windTupSet $; since $ |\windTupSet| $ is polynomial in $ n $, these steps are also polynomial in $ n $. Finally, step~\texttt{\nameref*{sub:postprocessing}} requires only $ \bigO{n} $ time. Consequently, the overall runtime of the algorithm is polynomial in $ n $.

\fi

%% file: Ulam/analysis.tex
\section{Analysis for ScalableMedianReconstruct algorithm}\label{sec:analysis_for_the_algorithm_for_Ulam_distance} 

In this section we prove~\cref{lem:bound.gmedian.to.offlineOut}.
\gmedianOfflineOut*

Let $ \lmedian = \argmin_{y \in \permutations}\sum_{i=1}^{5}\ed{y, \pi_{i}} $ be an optimal median of $ Q $. The corner stone in our proof is to bound $\cost{\offlineOut}{Q}$ in terms of $\cost{\lmedian}{Q}$ and the total slack. 

\begin{restatable}{lemma}{estimateLocalMedian}\label{lem:mpc.estimate.local.median}
    Let $ Q = \{ \pi_{1}, \pi_{2}, \dots, \pi_{5} \} $ be a set of any five permutations from $ \inpset $. The output $ \offlineOut $ of $ \scalableCycleRemoval $ with input $ Q $ satisfies
    \begin{align*}
        \sum_{i=1}^{5} \ed{\offlineOut, \pi_{i}} \leq 1.000225\sum_{i=1}^{5}\ed{\lmedian, \pi_{i}} + 66 \totalslack{\gmedian}. \nonumber
    \end{align*}
\end{restatable}

In~\cref{sub:notation_and_terminoloy}, we introduce the notations and terminologies used in the analysis. We provide an overview of the proof of~\cref{lem:mpc.estimate.local.median} in~\cref{sec:overview.of.the.analysis}. The formal proof is detailed in~\cref{sub:an_approximate_interval_for_the_medium_interval},~\cref{sub:property_of_the_cycle_removal_algorithm},~\cref{sub:a_good_approximate_string}, and~\cref{sub:proof_of_cref_lem_bound_gmedian_to_offlineout_}. Finally,~\cref{sec:from_an_optimal_median_of_five_permutations_to_an_optimal_median_of_the_input_set} completes the analysis by proving~\cref{lem:bound.gmedian.to.offlineOut}, using~\cref{lem:mpc.estimate.local.median}.

\subsection{Notations and Terminologies Used in the Analysis}\label{sub:notation_and_terminoloy}

For any permutation $ \perm \in \permutations $, fix an optimal alignment between $ \perm $ and each $ \pi_{i} $ in $ Q $. For each $ i=1,2, \dots , 5 $, let $ \unali{\pi_{i}}^{\perm} $ be the set of elements that are unaligned when comparing $ \perm $ with $ \pi_{i} $. Then $ \card{\unali{\pi_{i}}^{\perm}} = \ed{\perm, \pi_{i}} $.

We define a bad set $ B_{\perm} $ as the set of elements that are unaligned when comparing $ \perm $ with at least two permutations from $ Q $,
\begin{align}
    B_{\perm} = \bigcup_{1 \leq i < j \leq 5} (\unali{\pi_{i}}^{ \perm } \cap \unali{\pi_{j}}^{\perm}). \nonumber
\end{align}

Fix an optimal alignment between $ \lmedian $ and each $ \pi_{i} $ in $ Q $.
Fix a partitioning of $ \lmedian $ into $ n^{\mpcSpaceRegime} $ disjoint contiguous blocks of size $ \blSize=n^{1-\mpcSpaceRegime} $, denoted by $ \lmedian[\ell_{j}, r_{j}) $. Here $ r_{j} - \ell_{j} = n^{1-\mpcSpaceRegime} $, and $ \ell_{j+1}=r_{j} $. This alignment decomposes each $ \pi_{i} $ into $ n^{\mpcSpaceRegime} $ substrings, denoted by $ \pi_{i}[\optStartPoint{i,j}, \optEndPoint{i,j}) $, such that the block $ \lmedian[\ell_{j}, r_{j}) $ is aligned with $ \pi_{i}[\optStartPoint{i,j}, \optEndPoint{i,j}) $.

Recall that $ \edErr>0 $ is a small constant we used in the algorithm $ \scalableCycleRemoval $.

For each $ i=1,2, \dots , 5 $, let $ n^{\upperBoundED{i}} $ be a $ ( 1+\edErr ) $-approximation of $ \ed{\lmedian, \pi_{i}} $, that is, $ \ed{\lmedian, \pi_{i}} \leq n^{\upperBoundED{i}} \leq (1+\edErr) \ed{ \lmedian, \pi_{i} } $.  

Define $ \gap{i} = \edErr n^{\upperBoundED{i} - \mpcSpaceRegime} $, for simplicity, we assume that $ \gap{i} $ is an integer.

We classify the intervals $ [\optStartPoint{i,j}, \optEndPoint{i,j}) $ based on their length relative to the block size $ \blSize $. An interval is called \emph{large} if $ \optEndPoint{i,j} - \optStartPoint{i,j} > \blSize/\edErr $,  \emph{small} if $ \optEndPoint{i,j} - \optStartPoint{i,j} < \gap{i} + \edErr\blSize $ and \emph{medium} otherwise.

An interval $ [ \optStartPoint{i,j}', \optEndPoint{i,j}' ) $ is called an approximation of $ [\optStartPoint{i,j}, \optEndPoint{i,j}) $ if
\begin{equation}\label{eq:approximate.pairs}
    \begin{gathered}
        \optStartPoint{i,j} \leq \optStartPoint{i,j}' \leq \optStartPoint{i,j} + \edErr n^{\upperBoundED{i} - \mpcSpaceRegime} \\
        \optEndPoint{i,j} - \edErr n^{\upperBoundED{i} - \mpcSpaceRegime} - \edErr \ed{\lmedian[\ell_{j}, r_{j}), \pi_{i}[\optStartPoint{i,j}, \optEndPoint{i,j})} \leq \optEndPoint{i,j}' \leq \optEndPoint{i,j}.
    \end{gathered}
\end{equation}

\subsection{Overview of the analysis.} \label{sec:overview.of.the.analysis}
We analyze each step of $ \scalableCycleRemoval $ algorithm.

\paragraph{Analysis of~\nameref*{sub:windows_decomposition}}\label{par:analysis_of_nameref_sub_windows_decomposition_} 
We show that in this step, for any block $ \lmedian[\ell_{j}, r_{j}) $, there is a constructive group $ \{ \pi_{i}[\startPoint{i}, \endPoint{i}) \}_{i=1}^{5} $ such that, for each $ i=1,2, \dots ,5 $, if the interval $ [\optStartPoint{i,j}, \optEndPoint{i,j}) $ is medium, then the interval $ [\startPoint{i}, \endPoint{i}) $ satisfies the approximation condition in~\cref{eq:approximate.pairs}. The details are provided in~\cref{sub:an_approximate_interval_for_the_medium_interval}.

\paragraph{Analysis of~\nameref*{sub:local_aggregating_algorithm}}\label{par:analysis_of_protect_nameref_localaggregation_} 

Consider a target block $ \lmedian[\ell_{j}, r_{j}) $ for which all five aligned intervals $ [\optStartPoint{i,j}, \optEndPoint{i,j}) $ are medium. Recall that the previous step established the existence of a constructive group $ \{ \pi_{i}[\startPoint{i}, \endPoint{i}) \}_{i=1}^{5} $ in which each $ [\startPoint{i}, \endPoint{i}) $ approximates $ [\optStartPoint{i,j}, \optEndPoint{i,j}) $. In this step, we demonstrate that the output $ \mpcLocalOut $ generated by \texttt{\nameref*{sub:local_aggregating_algorithm}} on input $ \{ \pi_{i}[\startPoint{i}, \endPoint{i}) \}_{i=1}^{5} $ shares a long common subsequence with the target segment $ \lmedian[\ell_{j}, r_{j}) $.

A similar guarantee holds when $ [\ell_{j}, r_{j}) $ is aligned with exactly four medium aligned intervals. In this scenario, there are also constructive groups $ \{ \pi_{i}[\startPoint{i}, \endPoint{i}) \}_{i=1}^{5} $ where four of them has $ [\startPoint{i}, \endPoint{i}) $ approximates the corresponding medium aligned intervals, and the remaining one is an empty string of the form $ \pi_{i}[\startPoint{i}, \endPoint{i}) $ with $ \startPoint{i} = \endPoint{i} $. We show that the outputs of \texttt{\nameref*{sub:local_aggregating_algorithm}} on such inputs also share a long common subsequence with $ \lmedian[\ell_{j}, r_{j}) $. We refer to this specific invocation as running \texttt{\nameref*{sub:local_aggregating_algorithm}} on four medium intervals.
The detailed analysis is presented in~\cref{sub:property_of_the_cycle_removal_algorithm}.

\paragraph{Analysis of~\nameref*{sub:combining_blocks}}\label{par:analysis_of_nameref_sub_combining_blocks_} 

We identify a string $ \offlineGoodString = \blIntS{1}\blIntS{2}\dots \blIntS{n^{\mpcSpaceRegime}} $ in the search space of \texttt{\nameref*{sub:combining_blocks}}, such that $ \blEd(\offlineGoodString) $ can be bounded from above.
In this string $ \offlineGoodString $, a block $ \blIntS{j} $ is a $ \blSize $-length string of $ \specialChar $ elements if either there is at least one permutation $ \pi_{i} $ with large aligned interval $ [\optStartPoint{i,j}, \optEndPoint{i,j}) $, or there are at least two permutations $ \pi_{i}, \pi_{i'} $ with small aligned intervals $ [\optStartPoint{i,j}, \optEndPoint{i,j}) $, $ [\optStartPoint{i',j}, \optEndPoint{i',j}) $.
In the remaining case, where for the $ j $-th block there are at least $ 4 $ permutations $ \pi_{i} $ with medium aligned intervals $ [\optStartPoint{i,j}, \optEndPoint{i,j}) $, recall that we have established the existence of a constructive group $ \{ \pi_{i}[\startPoint{i}, \endPoint{i}) \}_{i=1}^{5} $ in which $ [ \startPoint{i}, \endPoint{i} ) $ approximate these $ [\optStartPoint{i,j}, \optEndPoint{i,j}) $. Specifically, $ [\startPoint{i}, \endPoint{i}) $ approximates the medium interval $ [\optStartPoint{i,j}, \optEndPoint{i,j}) $, while for any $ i $ with a small aligned interval, we have $ \startPoint{i} = \endPoint{i} $. We define $ \blIntS{j} $ as the output of \texttt{\nameref*{sub:local_aggregating_algorithm}} on input $ \{ \pi_{i}[\startPoint{i}, \endPoint{i}) \}_{i=1}^{5} $. By construction, the resulting string $ \offlineGoodString $ lies within the search space of \texttt{\nameref*{sub:combining_blocks}}.

We argue that $ \offlineGoodString $ is close to the median $ \lmedian $ of $ Q $. This is a consequence obtained from an upper bound on the sum of distance between each block $ \blIntS{j} $ and its corresponding block $ \lmedian[\ell_{j}, r_{j}) $. The argument proceeds differently for different types of blocks:
\begin{itemize}
    \item From the analysis for \texttt{\nameref*{sub:local_aggregating_algorithm}} algorithm, each block $ \blIntS{j} $ obtained from \texttt{\nameref*{sub:local_aggregating_algorithm}} algorithm shares a long common subsequence with the corresponding block $ \lmedian[\ell_{j}, r_{j}) $.
    \item For the blocks $ \lmedian[\ell_{j}, r_{j}) $ such that there is at least one permutations with the large aligned interval $ [\optStartPoint{i,j}, \optEndPoint{i,j}) $, then the distance between $ \lmedian[\ell_{j}, r_{j}) $ and any string of length $ \blSize $ is dominated by $ 2\blSize \leq \sum_{i=1}^{5}\ed{\lmedian[\ell_{j}, r_{j}), \pi_{i}[\optStartPoint{i,j}, \optEndPoint{i,j})} $.
    \item For the remaining blocks, there are at least two permutations with small aligned intervals. We show that this case does not happen too often.
\end{itemize}
The closeness between $ \offlineGoodString $ and $ \lmedian $ allows us to upper bound the sum of the distances between $ \offlineGoodString $ and each permutation $ \pi_{i} $. More precisely, we derive an upper bound for $ \blEd(\offlineGoodString) $.

The output $ \offlineIntS $ of~\cref{alg.combining.blocks} is a string achieving the minimum $ \blEd $ in the search space, and therefore, $ \blEd(\offlineIntS) \leq \blEd(\offlineGoodString) $. 

We present the detailed analysis in~\cref{sub:a_good_approximate_string}.

\paragraph{Analysis of~\nameref*{sub:postprocessing}}\label{par:analysis_of_nameref_sub_postprocessing_} 

The output in~\texttt{\nameref*{sub:postprocessing}} is obtained from $ \offlineIntS $ by removing dummy elements so that the remaining string is of length $ n $; the remaining dummy elements are replaced by unused elements. These operations do not increase the distance. Let $ \offlineOut $ be the output permutation in \texttt{\nameref*{sub:combining_blocks}}, then, $ \blEd(\offlineOut) \leq \blEd(\offlineGoodString) $. This inequality leads to the upper bound on the sum of distance $ \sum_{i=1}^{5}\ed{\offlineOut, \pi_{i}} $, completing the analysis. We analyze~\texttt{\nameref*{sub:postprocessing}} in~\cref{sub:proof_of_cref_lem_bound_gmedian_to_offlineout_}, which then conclude the proof of~\cref{lem:mpc.estimate.local.median}. 

\subsection{Analysis of~\nameref*{sub:windows_decomposition} Step}\label{sub:an_approximate_interval_for_the_medium_interval} 
\texttt{\nameref*{sub:windows_decomposition}} constructs a set $ \windTupSet := \cup_{j=1}^{n^{\mpcSpaceRegime}}\windTupSet{j} $, where $ \windTupSet{j} $ is defined as
\begin{align}
    \windTupSet{j} := \windowSet{j}^{5} \cup \bigcup_{k=1}^{5} \left( \windowSet{j}^{k-1} \times \sWindowSet \times \windowSet{j}^{5-k} \right). \tag{\ref{eq.def.wind.tup.set.j}}
\end{align}
See~\cref{sub:windows_decomposition} for the full definitions of these sets.

We begin by proving the existence of a pair $ \pair{ \startPoint, \endPoint } \in \windowSet{j} $ that approximates the interval $ [\optStartPoint{i,j}, \optEndPoint{i,j}) $ whenever the interval $ [\optStartPoint{i,j}, \optEndPoint{i,j}) $ is medium.
\begin{lemma}\label{lem:good.candidate.exists}
    Fix a permutation $ \pi_{i} $ and an index $ j $. If the interval $ [\optStartPoint{i,j}, \optEndPoint{i,j}) $ is medium, that is, $ \optStartPoint{i,j} + \gap{i} + \edErr \blSize < \optEndPoint{i,j} \leq \optStartPoint{i,j} + \blSize/\edErr  $, then $ \windowSet{j} $ contains a pair $ \pair{ \startPoint, \endPoint } = \pair{ \optStartPoint{i,j}', \optEndPoint{i,j}' } $ satisfying
    \begin{itemize}
       \item $ \optStartPoint{i,j} \leq \optStartPoint{i,j}' \leq \optStartPoint{i,j} + \edErr n^{\upperBoundED{i} - \mpcSpaceRegime} $,
        \item $ \optEndPoint{i,j} - \edErr n^{\upperBoundED{i} - \mpcSpaceRegime} - \edErr \ed{\lmedian[\ell_{j}, r_{j}), \pi_{i}[\optStartPoint{i,j}, \optEndPoint{i,j})} \leq \optEndPoint{i,j}' \leq \beta _{i,j} $.
    \end{itemize}
\end{lemma}
\newcommand{\edij}{d}
\begin{proof}
    For notational convenience, we omit the subscript $ i,j $ in the proof.

    Under the assumption that $ \ed{\lmedian, \pi} \leq n^{\upperBoundED} $, we have $ \card{\ell - \optStartPoint}\leq n^{\upperBoundED} $. Since $ \windowSet{j} $ contains pairs with start positions from $ \ell - n^{\upperBoundED} $ to $ \ell + n^{\upperBoundED} $ with gap $ \gap = \edErr n^{\upperBoundED - \mpcSpaceRegime} $, one of the start position $ \startPoint = \optStartPoint' $ satisfies $\optStartPoint \leq \optStartPoint' \leq \optStartPoint + \edErr n^{\upperBoundED - \mpcSpaceRegime} $.

    In pairs with start position $ \optStartPoint' $, their ending positions are of the form in~\cref{eq:mpc_ulam_end_point_choices}. We have the following observations:
    \begin{itemize}
        \item There exists an integer $ A_{1}\in [0, \log_{1+\edErr }\min(\frac{\blSize}{\edErr }, n^{\upperBoundED})] $ such that $ \optEndPoint\leq \optStartPoint' + \blSize + (1+\edErr )^{A_{1}} $.

            Indeed, with $ \optEndPoint $ in the assumption of the lemma, we have $ \optEndPoint \leq \optStartPoint + \frac{\blSize}{\edErr } \leq \optStartPoint' + \frac{\blSize}{\edErr } $. Additionally, as $ \ed{\lmedian, \pi} \leq n^{\upperBoundED} $, and $ \blSize = r - \ell $, it follows that $ \optEndPoint \leq \optStartPoint + n^{\upperBoundED} + \blSize \leq \optStartPoint' + \blSize + n^{\upperBoundED } $. Overall, $ \optEndPoint \leq \optStartPoint' + \blSize + \min(\frac{\blSize}{\edErr }, n^{\upperBoundED}) $, and therefore,  if we take $ A_{1} = \log_{ 1+\edErr  }\min(\frac{\blSize}{\edErr }, n^{\upperBoundED}) $, then, $ \optEndPoint \leq \optStartPoint' + \blSize + (1+\edErr )^{A_{1}} $.
        \item There exists an integer $ A_{2}\in [0, \log_{1+\edErr }(\blSize(1-\edErr ))] $ such that $ \optEndPoint \geq \optStartPoint' + \blSize - (1+\edErr )^{A_{2}} $.

            Indeed, with $ \optEndPoint $ in the assumption of the lemma, we have $ \optEndPoint \geq \optStartPoint + \gap + \edErr \blSize \geq \optStartPoint' + \edErr \blSize $. Therefore, if we take $ A_{2} = \log_{1+\edErr }(\blSize(1-\edErr )) $, then, $ \optEndPoint \geq \optStartPoint' + \blSize - (1+\edErr )^{A_{2}} $.
    \end{itemize}
    Set $ A=A_{1} $ if $ \optEndPoint \geq \optStartPoint' + \blSize $, and $ A=A_{2} $ if $ \optEndPoint \leq \optStartPoint' + \blSize $, then $ \card{\optEndPoint - \optStartPoint' - \blSize} \leq (1+\edErr )^{A} $. We can choose an integer $ a\in [0,A) $ such that $ (1+\edErr )^{a} \leq \card{\optEndPoint - \optStartPoint' - \blSize} \leq (1+\edErr )^{a+1} $. 

    Let $ \edij = \ed{\lmedian[\ell, r), \pi[\optStartPoint, \optEndPoint)} $.

    Since $ \card{\optEndPoint - \optStartPoint' - \blSize} \leq \card{\optEndPoint - \optStartPoint - \blSize} + \card{\optStartPoint - \optStartPoint'} \leq \edij + \gap $, it follows that $ (1+\edErr )^{a} \leq \edij + \gap $.

    We show a desired ending position $ \optEndPoint' $ by cases.
    \begin{itemize}
        \item $ \optEndPoint' = \optStartPoint' + \blSize $ if $ \optEndPoint = \optStartPoint' + \blSize $. Obviously, $ \optEndPoint' = \optEndPoint $.
        \item $ \optEndPoint' = \optStartPoint' + \blSize + (1+\edErr )^{a} \leq \optEndPoint $, if $ \optStartPoint'+\blSize \leq \optEndPoint $. Observe that $ \optEndPoint - \optEndPoint' = (\optEndPoint - \optStartPoint' -\blSize) - (1+\edErr )^{a} \leq (1+\edErr )^{a+1} - (1+\edErr )^{a} = (1+\edErr )^{a}\edErr  \leq \edErr \edij + \edErr \gap \leq \edErr \edij + \gap $. Hence, $ \optEndPoint - \edErr \edij - \gap \leq \optEndPoint' \leq \optEndPoint $.
        \item $ \optEndPoint' = \optStartPoint' + \blSize - (1+\edErr )^{a+1} \leq \optEndPoint $, if $ \optStartPoint'+\blSize \geq \optEndPoint $. Observe that $ \optEndPoint - \optEndPoint' = (\optEndPoint - \optStartPoint' -\blSize) + (1+\edErr )^{a+1} \leq -(1+\edErr )^{a} + (1+\edErr )^{a+1} = (1+\edErr )^{a}\edErr  \leq \edErr \edij + \edErr \gap \leq \edErr \edij + \gap $. Hence, $ \optEndPoint - \edErr \edij - \gap \leq \optEndPoint' \leq \optEndPoint $.
    \end{itemize}
    We can conclude that the pair $ \pair{ \startPoint, \endPoint } = \pair{ \optStartPoint', \optEndPoint' } $ satisfies the conditions in the lemma.
\end{proof}

Fix an interval $ [\ell_{j}, r_{j}) $ of $ \lmedian $, we partition the set of aligned intervals $ \{ [\optStartPoint{i,j}, \optEndPoint{i,j}) \}_{i=1}^{5} $ into 
\begin{itemize}
    \item $ \lInt{j} := \{[\optStartPoint{i,j}, \optEndPoint{i,j}): \optEndPoint{i,j}-\optStartPoint{i,j} > \frac{\blSize}{\edErr }  \} $ the set of intervals which are large, for each $ i=1,2, \dots, 5 $, 
    \item $ \nInt{j} := \{[\optStartPoint{i,j}, \optEndPoint{i,j}):\ \gap{i} + \edErr \blSize \leq \optEndPoint{i,j}-\optStartPoint{i,j} \leq \frac{\blSize}{\edErr }  \} $ the set of intervals which are medium, for each $ i=1,2, \dots, 5 $,
    \item $ \sInt{j} := \{[\optStartPoint{i,j}, \optEndPoint{i,j}): \optEndPoint{i,j}-\optStartPoint{i,j} < \gap{i} + \edErr \blSize  \} $ the set of intervals which are small, for each $ i=1,2, \dots, 5 $.
\end{itemize}

By~\cref{lem:good.candidate.exists} and the definition of $ \windTupSet{j} $ (\cref{eq.def.wind.tup.set.j}), the following result is immediate.
\begin{corollary}\label{cor:good.tuple.exists}
    For an interval $ [\ell_{j}, r_{j}) $ in $ \lmedian $ and the aligned intervals $ [\optStartPoint{i,j}, \optEndPoint{i,j}) $ in $ \pi_{i} $, there is a constructive group $  \{ \pi_{1}[\startPoint{1}, \endPoint{1}), \pi_{2}[\startPoint{2}, \endPoint{2}), \dots, \pi_{5}[\startPoint{5}, \endPoint{5}) \} $ such that for each $ i=1,2, \dots ,5 $: 
\begin{itemize}
    \item if $ [\optStartPoint{i,j}, \optEndPoint{i,j}) $ is medium, then $ [\startPoint{i}, \endPoint{i}) $ satisfies the approximation condition in~\cref{eq:approximate.pairs},
    \item if $ [\optStartPoint{i,j}, \optEndPoint{i,j}) $ is small, then $ \pair{ \startPoint{i}, \endPoint{i} }\in \sWindowSet $.
\end{itemize}
\end{corollary}

\subsection{Analysis of~\nameref*{sub:local_aggregating_algorithm} Step}\label{sub:property_of_the_cycle_removal_algorithm} 

For any interval $ [\ell_{j}, r_{j}) $ satisfying $ \lInt{j} = \emptyset $ and $ \card{\nInt{j}}\geq 4 $, let $ \anInt{j} $ be the constructive group guaranteed by~\cref{cor:good.tuple.exists}. We denote the strings in $ \anInt{j} $ as $ \{ \pi_{i}[\optStartPoint{i,j}', \optEndPoint{i,j}') \}_{i=1}^{5} $.

If $ \lInt{j} = \emptyset $ and  $ \card{\nInt{j}}\geq 4 $, we show that the output of \texttt{\nameref*{sub:local_aggregating_algorithm}} algorithm (\cref{alg.local.cycle.removal}) with input $ \anInt{j} $ gives a good approximation to $ \lmedian[\ell_{j}, r_{j}) $.

We introduce the following notations. 

Denote $ \pi_{i,j} = \pi_{i}[\optStartPoint{i,j}, \optEndPoint{i,j}) $, $ \pi_{i,j}' = \pi_{i}[\optStartPoint{i,j}', \optEndPoint{i,j}') $, and $ T_{j} = \{\pi_{i,j}\}_{i=1}^{5} $, $ H_{j} = \{ \pi'_{j,i} \}_{i=1}^{5} $. 

By definition, either $ [\optStartPoint{i,j}', \optEndPoint{i,j}') $ approximates $ [\optStartPoint{i,j}, \optEndPoint{i,j}) $ if $ [\optStartPoint{i,j}, \optEndPoint{i,j}) $ is medium, or $ \optStartPoint{i,j}' = \optEndPoint{i,j}' $ if $ [\optStartPoint{i,j}, \optEndPoint{i,j}) $ is small. In the latter case, $ \pi_{i}[\optStartPoint{i,j}', \optEndPoint{i,j}') $ is an empty string.

By abuse of notation, $ \pi_{i,j}\setminus \pi_{i,j}' $ denotes the set of elements that are in $ \pi_{i,j} $ but not in $ \pi_{i,j}' $.

Let $ G_{j} = G \cap \lmedian[\ell_{j}, r_{j})$ and $ B_{j} = B_{\lmedian} \cap \lmedian[\ell_{j}, r_{j})  $ be the set of \emph{good} elements and the set of \emph{bad} elements in $ \lmedian[\ell_{j}, r_{j}) $, respectively. Also let $ \errGood{j} = G_{j}\cap \bigcup_{i=1}^{5}(\pi_{i,j}\setminus \pi'_{i,j} )$. Finally, let $ B^{H}_{j} $ be the set of bad elements in $ B_{\lmedian} $ that appear in at least four strings in $ H_{j} $.

\begin{lemma}\label{lem:local.cycle.removal}
    Consider an interval $ [\ell_{j}, r_{j}) $ where $ \lInt{j} = \emptyset $ and $ \card{ \nInt{j} } \geq 4 $. Let $ \mpcLocalOut $ be the string returned by~\cref{alg.local.cycle.removal} with input $ \anInt{j} $. Then
    \begin{align}
        \lcs(\mpcLocalOut, \lmedian[\ell_{j}, r_{j})) &\geq \card{G_{j}} - 3\card{B^{H}_{j}} - 4\card{\bar{G_{j}}}. \nonumber
    \end{align}
\end{lemma}
\begin{proof}

    Consider the subset of good elements $ \appGood{j} = G_{j} \setminus \errGood{j} $. Following from the definition of $ G_{j} $ and $ \errGood{j} $, every element $ a \in \appGood{j} $ is aligned with at least $ 4 $ strings in $ H_{j} $. Thus, for two distinct elements $ a,b \in \appGood{j} $, for at least three strings $  \pi_{i,j}'\in H_{i} $, $ a $ and $ b $ are aligned between $ \lmedian[\ell_{j}, r_{j}) $ and $ \pi_{i,j}' $. Hence, by the construction, there is a directed edge from $ a $ to $ b $ if $ a $ appears before $ b $ in $ \lmedian $, for every $ a,b \in \appGood{j} $.

    Observe that, for every pair of vertices $ a,b $ in the majority graph, if there is no edge between $ a $ and $ b $, then at least one the two vertices must be either a bad element in $ B^{H}_{j} $ or a good element in $ \errGood{j} $.

    After removing each pair of vertices with no edge, the number of vertices in $ \appGood{j} $ that is removed is most $ \card{B^{H}_{j}} + \card{\errGood{j}} $. At this point, the remaining graph is a tournament. Note that, the shortest cycle in a tournament is of length $ 3 $. Moreover, in every such cycle, there is at least one vertex in $ B^{H}_{j} $ or a vertex in $ \errGood{j} $. Thus, after iteratively removing all cycles of shortest length, the number of vertices removed in $ \appGood{j} $ is at most $ 2( \card{B^{H}_{j}} + \card{\bar{G_{j}}} )$.

    The final graph is a directed acyclic graph with the number of veritces in $ \appGood{j} $ least $ \card{\appGood{j}} - 3\card{B^{H}_{j}} - 3\card{\errGood{j}} $. As $ \mpcLocalOut $ is a topological ordering of vertices in the graph, it follows that longest common subsequence between $ \mpcLocalOut $ and $ \lmedian[\ell_{j}, r_{j}) $ is greater than $\card{\appGood{j}} - 3\card{B^{H}_{j}} - 3\card{\errGood{j}} \nonumber= \card{G_{j}} - 3\card{B^{H}_{j}} - 4\card{\errGood{j}}$,
where the equality is due to $ \card{\appGood{j}} = \card{G_{j}} - \card{\errGood{j}} $.

\end{proof}

We count the number of elements $ \specialChar $'s padded in~\cref{line.return.padded.string} of~\cref{alg.local.cycle.removal}.

\begin{lemma}\label{cor:ed.padded.topo}
    Let $ \mpcLocalOut $ be obtained from~\cref{line.return.padded.string} of~\cref{alg.local.cycle.removal} over the input $ \anInt{j} $. Then, the number of dummy element $ \specialChar $'s in $ \mpcLocalOut $ is at most $ 3\card{B^{H}_{j}} + 4\card{\errGood{j}} + \card{B_{j}} $.
\end{lemma}
\begin{proof}
    Observe that the vertex set of the majority graph constructed in~\cref{line.majority.graph} of~\cref{alg.local.cycle.removal} contains $ \appGood{j}\cup B^{T}_{j} $. Consequently, the algorithm initializes with a graph of at least $ \card{\appGood{j}} + \card{B^{H}_{j}} $ vertices. Using the identities $ \appGood{j} = G_{j}\setminus \errGood{j} $ and $ \blSize = \card{G_{j}} + \card{B_{j}} $, it follows that the initial size is at least $ \blSize - \card{B_{j}} - \card{\errGood{j}} $.
    As in~\cref{lem:local.cycle.removal}, at most $ 3\card{B^{H}_{j}} + 3\card{\errGood{j}} $ vertices are removed by $\cref{alg.local.cycle.removal}$. 
    Thus, the number of dummy elements $ \specialChar $ in the final output $ \mpcLocalOut $ is at most
    \begin{align}
        ( 3\card{B^{H}_{j}} + 3 \card{\errGood{j}} ) + \card{\errGood{j}} + \card{B_{j}} = 3\card{B^{H}_{j}} + 4\card{\errGood{j}} + \card{B_{j}}. \nonumber
    \end{align}

\end{proof}

\subsection{Analysis of~\nameref*{sub:combining_blocks} Step}\label{sub:a_good_approximate_string}

Let $ \offlineIntS $ be the string obtained in~\texttt{\nameref*{sub:combining_blocks}} step. In this section, we give an upper bound on $ \blEd(\offlineIntS) $.

Our strategy is to first construct a string $ \offlineGoodString $ such that $ \offlineGoodString $ is in the search space of~\texttt{\nameref*{sub:combining_blocks}} and is close to $\lmedian$. This structural similarity enables us to derive an upper bound on $\blEd(\offlineGoodString)$. Consequently, this yields an upper bound on the optimal cost $\blEd(\offlineIntS)$ as well, since $\blEd(\offlineIntS) \leq \blEd(\offlineGoodString)$ by definition of $ \offlineIntS $.

Prior to that, we show some lemmas concerning the intervals $ [\ell_{j}, r_{j}) $ with $ \lInt{j} \neq \emptyset $ or $ \lInt{j} = \emptyset $ and $ \card{\sInt{j}} \geq 2 $.

\begin{lemma}\label{lem:approx.large.strings}
    For each block $ [\ell_{j}, r_{j}) $, if there exists $ i\in [5] $ such that $ \optEndPoint{i,j} - \optStartPoint{i,j} > \blSize/\edErr  $, then setting $ \mpcLocalOut $ to be an arbitrary $ \blSize $-length string satisfies that
    \begin{align}
        \ed{\mpcLocalOut, \lmedian[\ell_{j}, r_{j})} & \leq \dfrac{5}{2}\edErr \sum_{i=1}^{5}\ed{\lmedian[\ell_{j}, r_{j}), \pi_{i}[\optStartPoint{i,j}, \optEndPoint{i,j})}, \nonumber
    \end{align}
    for any $ \edErr  \leq \frac{1}{5} $.
\end{lemma}
\begin{proof}
    Let $ i $ be an index such that $ \optEndPoint{i,j} - \optStartPoint{i,j} > \blSize/\edErr  $. Observe that
    \begin{align}
        \ed{\lmedian[\ell_{j}, r_{j}), \pi_{i}[\optStartPoint{i,j}, \optEndPoint{i,j})} \geq \card{(\optEndPoint{i,j} - \optStartPoint{i,j}) - \blSize} > \dfrac{\blSize}{\edErr } - \blSize = \left(\dfrac{1}{\edErr } - 1\right)\blSize. \nonumber
    \end{align}
    As $ \mpcLocalOut $ is an arbitrary $ \blSize $-length string, it follows that, for any $ \edErr  \leq \frac{1}{5} $,
    \begin{align}
        \ed{\mpcLocalOut, \lmedian[\ell_{j}, r_{j})} \leq 2\blSize \leq \dfrac{5\edErr }{2} \left(\dfrac{1}{\edErr }-1\right)\blSize < \dfrac{5\edErr }{2}\ed{\lmedian[\ell_{j}, r_{j}), \pi_{i}[\optStartPoint{i,j}, \optEndPoint{i,j})}. \nonumber
    \end{align}
    Therefore, $ \ed{\mpcLocalOut, \lmedian[\ell_{j}, r_{j})} \leq \dfrac{5}{2}\edErr \sum_{i=1}^{5}\ed{\lmedian[\ell_{j}, r_{j}), \pi_{i}[\optStartPoint{i,j}, \optEndPoint{i,j})} $.
\end{proof}

\begin{lemma}\label{lem:at.least.two.small.strings}
    The number of blocks $ [\ell_{j}, r_{j}) $ for which $ \card{\sInt{j}}\geq 2 $ is at most $ \frac{\card{B_{\lmedian}}}{(1-6\edErr )\blSize} $.
\end{lemma}
\begin{proof}
    By definition of $ G_{j} $, every element $ a\in G_{j} $ is aligned between $ \lmedian[\ell_{j}, r_{j}) $ with at least $ 4 $ substrings among $ \pi_{i,j} $. Since $ \card{\sInt{j}}\geq 2 $ and hence $ \card{\nInt{j}}\leq 3 $, every such $ a $ is aligned with at least one substring in $ \sInt{j} $. Therefore, each element in $ G_{j} $ corresponds to at least one element in the substrings in $ \sInt{j} $. It follows that $ \card{G_{j}}\leq \sum_{[\optStartPoint{i,j},\optEndPoint{i,j}) \in \sInt{j}}(\optEndPoint{i,j}-\optStartPoint{i,j}) $. For each interval $ [\optStartPoint{i,j}, \optEndPoint{i,j}) \in \sInt{j} $, by definition, $ \optEndPoint{i,j}-\optStartPoint{i,j} < \gap{i} + \edErr \blSize \leq 3\edErr \blSize $, as $ \gap{i} = \edErr n^{\upperBoundED{i}-\mpcSpaceRegime} \leq 2\edErr \blSize $. Thus, $ \card{G_{j}}\leq 6\edErr \blSize $, and
    \begin{align}
        \card{\lmedian[\ell_{j}, r_{j})\setminus G_{j}} \geq \blSize - 6\edErr \blSize. \nonumber
    \end{align}

    Observe that $ \lmedian[\ell_{j}, r_{j})\setminus G_{j} \subseteq B_{\lmedian} $, and $ \cup_{j}\lmedian[\ell_{j}, r_{j})\setminus G_{j} = B_{\lmedian} $, hence, the number of blocks $ [\ell_{j}, r_{j}) $ with $ \card{\sInt{j}}\geq 2 $ is at most $\dfrac{\card{B_{\lmedian}}}{(1-6\edErr )\blSize} $.
\end{proof}

We partition the set of block indexes $ \{ 1,2, \dots, n^{\mpcSpaceRegime} \} $ into four subsets:
    \begin{itemize}
        \item $ \mathcal{L} = \{ j: \lInt{j} \neq \emptyset \} $,
        \item $ \mathcal{N} = \{ j: \lInt{j} = \emptyset,\ \card{\nInt{j}} = 5 \} $,
        \item $ \mathcal{N}' = \{ j: \lInt{j} = \emptyset,\ \card{\nInt{j}} = 4 \} $,
        \item $ \mathcal{S} = \{ j: \lInt{j} = \emptyset,\ \card{\sInt{j}}\geq 2 \} $ (note that $ \mathcal{S} = \{ j: \lInt{j} = \emptyset,\ \card{\lInt{j}}\leq 3 \} $).
    \end{itemize}
    
    \paragraph{A Good Estimate String.} We define a string $ \offlineGoodString = \blIntS{1}\blIntS{2}\dots \blIntS{n^{\mpcSpaceRegime}} $ consisting of $ n^{\mpcSpaceRegime} $ blocks constructed as follows. If $ j \in \mathcal{L} \cup \mathcal{S} $, then $ \blIntS{j} $ is a string of $ \specialChar $ elements of length $ \blSize $. If $ j \in \mathcal{N} \cup \mathcal{N}' $, then $ \blIntS{j} $ is the output of the \texttt{\nameref*{sub:local_aggregating_algorithm}} algorithm (\cref{alg.local.cycle.removal}) on input $ \anInt{j} $. \Cref{lem:bounding.offlineGoodString.blocks} below establishes a bound on the indel distance between $ \offlineGoodString $ and $ \lmedian $. Prior to proving~\cref{lem:bounding.offlineGoodString.blocks}, we demonstrate the following auxiliary lemma.

\begin{lemma}\label{lem:diff.offlineGoodString.medium.goodSymbols}
    \begin{align}
        \sum_{j\in \mathcal{N}\cup \mathcal{N'}} ( \card{\blIntS{j}} - \card{G_{j}} ) \leq 5\card{B_{\lmedian}} + 4 \sum_{j\in \mathcal{{N} \cup \mathcal{N'}}} \card{\errGood{j}}. \nonumber
    \end{align}
\end{lemma}
\begin{proof}
    For each $ j\in \mathcal{N}\cup \mathcal{N'} $, $ \blIntS{j} $ is obtained from either~\cref{line.return.padded.string} or~\cref{line.already.long.string} of~\cref{alg.local.cycle.removal}. Let $ \tp_{j} $ be the string obtained by the topological sorting performed at~\cref{line.topological.sort} of~\cref{alg.local.cycle.removal}. Then either $ \blIntS{j} = \tp_{j}\specialChar^{\blSize-\card{\tp_{j}}} $ (if $ \card{\tp_{j}} < \blSize $), or $ \blIntS{j} = \tp_{j} $. Hence,
    \begin{align}
        \sum_{j\in \mathcal{N} \cup \mathcal{N}'} \card{\blIntS{j}} &= \sum_{j\in \mathcal{N} \cup \mathcal{N'}} \card{\tp_{j}} + \sum_{j \in \mathcal{N} \cup \mathcal{N'}: \card{\tp_{j}}< \blSize} \blSize - \card{\tp_{j}}. \nonumber
    \end{align}
    For each $ j\in \mathcal{N} \cup \mathcal{N'} $ such that $ \card{\tp_{j}} < \blSize $, $ \blSize - \card{\tp_{j}} $ is the number of padded element $ \specialChar $'s in $ \blIntS{j} $. By~\cref{cor:ed.padded.topo}, we have
    \begin{align}
        \blSize - \card{\tp_{j}} &\leq 3\card{B^{H}_{j}} + 4\card{\errGood{j}} + \card{B_{j}}. \nonumber
    \end{align}
    Following the definition of $ B^{H}_{j} $ and $ B_{j} $, we have that $ B^{H}_{j} $'s are mutually disjoint, and $ B_{j} $'s are mutually disjoint, for different $ j $'s. Thus,
    \begin{align}
        \sum_{j\in \mathcal{N} \cup \mathcal{N'}} \blSize - \card{\tp_{j}} &\leq 4\card{B_{\lmedian}} + 4 \sum_{j\in \mathcal{N} \cup \mathcal{N'}} \card{\errGood{j}}. \nonumber
    \end{align}
    Since we construct $ \tp_{j} $ using only the elements in $ G_{j} $ and $ B^{H}_{j} $, it follows that
    \begin{align}
        \sum_{j\in \mathcal{N} \cup \mathcal{N'}} \card{\tp_{j}} &\leq \sum_{j\in \mathcal{N} \cup \mathcal{N'}} \left( \card{G_{j}} + \card{B^{H}_{j}} \right) \leq \sum_{j\in \mathcal{N} \cup \mathcal{N'}} \card{G_{j}} + \card{B_{\lmedian}}, \nonumber
    \end{align}
    Therefore,
    \begin{align}
        \sum_{j\in \mathcal{N} \cup \mathcal{N'}} ( \card{\blIntS{j}} - \card{G_{j}} ) &\leq 5\card{B_{\lmedian}} + 4 \sum_{j\in \mathcal{N} \cup \mathcal{N'}} \card{\errGood{j}}. \nonumber
    \end{align}
\end{proof}

\begin{lemma}\label{lem:bounding.offlineGoodString.blocks}
    Let $ \offlineGoodString $ be a string with $ n^{\mpcSpaceRegime} $ blocks in which its $ j $'th block $ \blIntS{j} $ is, if $ j\in \mathcal{L} \cup \mathcal{S} $, a string of all $ \specialChar $ elements of length $ \blSize $, and if $ j \in \mathcal{N} \cup \mathcal{N}' $, $ \blIntS{j} $ is obtained from \texttt{\nameref*{sub:local_aggregating_algorithm}} algorithm (\cref{alg.local.cycle.removal}) with input $ \anInt{j} $, then
    \begin{align}
        \sum_{j=1}^{n^{\mpcSpaceRegime}} \ed{\blIntS{j}, \lmedian[\ell_{j}, r_{j})} \leq (43.5\edErr  + 24\edErr ^{2})\sum_{i=1}^{5}\ed{\lmedian, \pi_{i}} + \left( \dfrac{1}{1-6\edErr } + 12 \right)\card{B_{\lmedian}}. \nonumber
    \end{align}
\end{lemma}
\begin{proof}
    We decompose the left sum.
    \begin{align}
        &\sum_{j=1}^{n^{\mpcSpaceRegime}} \ed{\blIntS{j}, \lmedian[\ell_{j}, r_{j})} \nonumber \\
        = &\sum_{j\in \mathcal{L}} \ed{\blIntS{j}, \lmedian[\ell_{j}, r_{j})} + \sum_{j\in \mathcal{S}} \ed{\blIntS{j}, \lmedian[\ell_{j}, r_{j})} + \sum_{j\in \mathcal{N}\cup \mathcal{N'}} \ed{\blIntS{j}, \lmedian[\ell_{j}, r_{j})}. \nonumber
    \end{align}
    For each $ j\in \mathcal{L} $, by~\cref{lem:approx.large.strings}, we have
    \begin{align}
        \ed{\offlineGoodString_{j}, \lmedian[\ell_{j}, r_{j})} &\leq \dfrac{5}{2}\edErr \sum_{i=1}^{5}\ed{\lmedian[\ell_{j}, r_{j}), \pi_{i,j}}. \nonumber
    \end{align}
    For $ j\in \mathcal{S} $, by~\cref{lem:at.least.two.small.strings}, $ \card{\mathcal{S}}\leq \frac{\card{B_{\lmedian}}}{( 1-6\edErr  )\blSize} $. Hence
    \begin{align}
        \sum_{j\in \mathcal{S}}\ed{\offlineGoodString_{j}, \lmedian[\ell_{j}, r_{j})} &\leq \sum_{j\in \mathcal{S}}\blSize \leq \dfrac{\card{B_{\lmedian}}}{(1-6\edErr )\blSize}\blSize = \dfrac{\card{B_{\lmedian}}}{(1-6\edErr )}. \nonumber
    \end{align}
    For each $ j\in \mathcal{N} \cup \mathcal{N'} $, by~\cref{lem:local.cycle.removal}, we have
    \begin{align}
        \ed{\blIntS{j},\lmedian[\ell_{j}, r_{j})} & = \card{\blIntS{j}} + \card{\lmedian[\ell_{j}, r_{j})} - 2\lcs(\blIntS{j}, \lmedian[\ell_{j}, r_{j})) \nonumber \\
                                            & \leq \card{\blIntS{j}} + \blSize - 2\card{G_{j}} + 6\card{B_{j}} + 8\card{\errGood{j}}. \nonumber
    \end{align}
    Hence,
    \begin{align}
        \sum_{j\in \mathcal{N}\cup \mathcal{N'}} \ed{\blIntS{j}, \lmedian[\ell_{j}, r_{j})} & \leq \sum_{j\in \mathcal{N}\cup \mathcal{N'}} \left( (\card{\blIntS{j}} - \card{G_{j}}) + (\blSize - \card{G_{j}}) + 6\card{B_{j}} + 8 \card{\errGood{j}} \right) \nonumber \\
                                                                                                            &\leq 5 \card{B_{\lmedian}} + 4 \sum_{j\in \mathcal{N} \cup \mathcal{N'}} \card{\errGood{j}} + 7\card{B_{\lmedian}} + 8 \sum_{j\in \mathcal{N}\cup \mathcal{N'}} \card{\errGood{j}} \nonumber \\ 
                                                                                                            & = 12\card{B_{\lmedian}} + 12 \sum_{j\in \mathcal{N}\cup \mathcal{N'}} \card{\errGood{j}}. \nonumber
    \end{align}
    The second inequality is due to $ \sum_{j\in \mathcal{N}\cup \mathcal{N'}}(\card{\blIntS{j}} - \card{G_{j}}) \leq 5\card{B_{\lmedian}} + 4 \sum_{j\in \mathcal{N} \cup \mathcal{N'}}\card{\errGood{j}} $ (\cref{lem:diff.offlineGoodString.medium.goodSymbols}), $ \sum_{j=1}^{n^{\mpcSpaceRegime}} (\blSize - \card{G_{j}}) = \card{B_{\lmedian}} $, and $ \sum_{j\in \mathcal{N}\cup \mathcal{N'}}\card{B_{j}}\leq \card{B_{\lmedian}} $.

    For $ j\in \mathcal{N}$, by definition of $ \errGood{j} $, we have
    \begin{align}
        \sum_{j\in \mathcal{N}} \card{\errGood{j}} &\leq \sum_{j\in \mathcal{N}} \sum_{i=1}^{5} \left( (\optStartPoint{i,j}' - \optStartPoint{i,j}) + (\optEndPoint{i,j} - \optEndPoint{i,j}') \right) \nonumber \\
                                                        &\leq \sum_{j\in \mathcal{N}} \sum_{i=1}^{5} \left( 2\edErr n^{\upperBoundED{i} - \mpcSpaceRegime} + \edErr \ed{\lmedian[\ell_{j}, r_{j}), \pi_{i,j}} \right). \nonumber \\
    \end{align}
    \newcommand{\smallI}{\hat{i}}
    For each $ j \in \mathcal{N'} $, let $ \smallI $ denote the unique index such that the corresponding interval in $ \anInt{j} $ is empty (i.e., $ \optStartPoint{\smallI,j}' = \optEndPoint{\smallI,j}' $). Then, by the definition of $ \errGood{j} $, we have
    \begin{align}
        \sum_{j\in \mathcal{N'}} \card{\errGood{j}} &\leq \sum_{j\in \mathcal{N'}} \left( (\optEndPoint{\smallI,j} - \optStartPoint{\smallI,j}) + \sum_{i \neq \smallI} \left( (\optStartPoint{i,j}' - \optStartPoint{i,j}) + (\optEndPoint{i,j} -\optEndPoint{i,j}') \right) \right) \nonumber \\
                                                    &\leq \sum_{j\in \mathcal{N'}} \left( \gap{\smallI} + \edErr \blSize + \sum_{i \neq \smallI} \left( 2\edErr n^{\upperBoundED{i} - \mpcSpaceRegime} + \edErr \ed{\lmedian[\ell_{j}, r_{j}), \pi_{i,j}} \right) \right) \nonumber \\
                                                    &\leq \sum_{j\in \mathcal{N}'} \left( \edErr n^{\upperBoundED{\smallI} - \mpcSpaceRegime} + \dfrac{3}{2}\edErr \ed{\lmedian[\ell_{j}, r_{j}), \pi_{\smallI,j}} + \sum_{i \neq \smallI}\left(2\edErr n^{\upperBoundED{i}-\mpcSpaceRegime} + \edErr \ed{\lmedian[\ell_{j}, r_{j}), \pi_{i,j}}\right) \right) \nonumber \\
                                                    & \leq \sum_{j\in \mathcal{N'}} \sum_{i=1}^{5} \left( 2\edErr n^{\upperBoundED{i} - \mpcSpaceRegime} + \dfrac{3}{2}\edErr \ed{\lmedian[\ell_{j}, r_{j}), \pi_{i,j}} \right). \nonumber
    \end{align}
    In the second inequality, we use the assumption that $ \optEndPoint{\smallI,j} - \optStartPoint{\smallI,j} \leq \gap{\smallI} + \edErr \blSize $. The third inequality is due to $ \gap{\smallI}\leq 2\edErr \blSize $, and therefore $ \ed{\lmedian[\ell_{j}, r_{j}), \pi_{\smallI,j}}\geq (r_{j} - \ell_{j}) - (\optEndPoint{\smallI, j} - \optStartPoint{\smallI, j}) \geq \blSize - (\gap{\smallI} + \edErr \blSize) \geq \blSize - 3\edErr \blSize $, which leads to $ \edErr \blSize \leq \frac{3}{2} \edErr \ed{\lmedian[\ell_{j}, r_{j}), \pi_{\smallI,j}} $, for any $ \edErr  \leq \frac{1}{9} $.

    Combining all together, we have
    \begin{align}
        \ed{\offlineGoodString, \lmedian} &\leq \dfrac{5}{2}\edErr \sum_{j\in \mathcal{L}} \sum_{i=1}^{5}\ed{\lmedian[\ell_{j}, r_{j}), \pi_{i,j}} + \dfrac{\card{B_{\lmedian}}}{(1-6\edErr )} + 12\card{B_{\lmedian}} \nonumber \\
                                     &\quad + 12 \sum_{j\in \mathcal{N}\cup \mathcal{N'}} \sum_{i=1}^{5} \left( 2\edErr n^{\upperBoundED{i} - \mpcSpaceRegime} + \dfrac{3}{2}\edErr \ed{\lmedian[\ell_{j}, r_{j}), \pi_{i,j}} \right) \nonumber \\
                                     &\leq \dfrac{39}{2}\edErr \sum_{i=1}^{5}\ed{\lmedian, \pi_{i}} + 24\edErr  \sum_{i=1}^{5} n^{\upperBoundED{i}} + \left( \dfrac{1}{1-6\edErr } + 12 \right)\card{B_{\lmedian}} \nonumber \\
                                     &\leq (43.5\edErr  + 24\edErr ^{2})\sum_{i=1}^{5}\ed{\lmedian, \pi_{i}} + \left( \dfrac{1}{1-6\edErr } + 12 \right)\card{B_{\lmedian}}. \nonumber
    \end{align}
    The last inequality is due to $ n^{\upperBoundED{i}} \leq (1+\edErr )\ed{\lmedian, \pi_{i}} $. This completes the proof.
\end{proof}

As $ \ed{\offlineGoodString, \lmedian} \leq \sum_{j=1}^{n^{\mpcSpaceRegime}} \ed{\blIntS{j}, \lmedian[\ell_{j}, r_{j})} $, the following corollary is immediate from~\cref{lem:bounding.offlineGoodString.blocks}.
\begin{corollary}\label{cor:approximation.local.solution}
    \begin{align}
        \sum_{i=1}^{5} \sum_{j=1}^{n^{\mpcSpaceRegime}}\ed{\blIntS{j}, \pi_{i}[\optStartPoint{i,j}, \optEndPoint{i,j})} &\leq (1 + 217.5\edErr  + 120\edErr ^{2})\sum_{i=1}^{5}\ed{\lmedian, \pi_{i}} + \left( \dfrac{5}{1-6\edErr } + 60 \right)\card{B_{\lmedian}}. \nonumber
    \end{align}
\end{corollary}
\begin{proof}
    By the triangle inequalities, we have
    \begin{align}
        &\>\>\>\>\>\sum_{i=1}^{5} \sum_{j=1}^{n^{\mpcSpaceRegime}}\ed{\blIntS{j}, \pi_{i}[\optStartPoint{i,j}, \optEndPoint{i,j})} \nonumber \\
        &\leq 5 \sum_{j=1}^{n^{\mpcSpaceRegime}} \ed{\blIntS{j}, \lmedian[\ell_{j}, r_{j})} + \sum_{i=1}^{5} \sum_{j=1}^{n^{\mpcSpaceRegime}}\ed{\lmedian[\ell_{j}, r_{j}), \pi_{i}[\optStartPoint{i,j}, \optEndPoint{i,j})}  \nonumber \\
        &\leq \left( 1 +  217.5\edErr  + 120\edErr ^{2}) \right)\sum_{i=1}^{5}\ed{\lmedian, \pi_{i}} + 5\left( \dfrac{1}{1-6\edErr } + 12 \right)\card{B_{\lmedian}}. \nonumber
    \end{align}
    The last inequality is due to~\cref{lem:bounding.offlineGoodString.blocks} and by definition, $ \sum_{j=1}^{n^{\mpcSpaceRegime}}\ed{\lmedian[\ell_{j}, r_{j}), \pi_{i}[\optStartPoint{i,j}, \optEndPoint{i,j})} = \ed{\lmedian, \pi_{i}} $.
\end{proof}

\paragraph{Bounding $ \blEd(\offlineGoodString) $.}\label{par:bounding_bled_offlinegoodstring_} 

The string $ \offlineGoodString = \blIntS{1}\blIntS{2}\dots \blIntS{n^{\mpcSpaceRegime}} $ belongs to the search space defined in the step \texttt{\nameref*{sub:combining_blocks}}, as it is induced by the valid sequence of tuples $ (\tuple{a_{1}}, \tuple{a_{2}}, \dots, \tuple{a_{k}}) $, with $ \tuple{a_{t}} \in C_{j_{t}} $ and indices $ \{ j_{1}, j_{2}, \dots, j_{k} \} = \mathcal{N} \cup \mathcal{N'} $. Denoting the components of the $t$-th tuple as $ \tuple{a_{t}} = \{ \offlineLocalOut, (\startPoint{i,t}, \endPoint{i,t}) \}_{i=1}^{5} $, we observe that for each $ i $: if the interval $ [\optStartPoint{i,j_{t}}, \optEndPoint{i,j_{t}}) $ is medium, then $ [\startPoint{i,t}, \endPoint{i,t}) = [\optStartPoint{i,j_{t}}', \optEndPoint{i,j_{t}}') $ approximates $ [\optStartPoint{i,t}, \optEndPoint{i,t}) $. Otherwise, if the interval is small, we have $ \startPoint{i,t}= \endPoint{i,t} $.

The value $ \blEd(\offlineGoodString) $ provides an upper bound on $ \ed{\offlineGoodString,\pi_{i}} $ corresponding to the following alignment: for every medium interval $ [\optStartPoint{i, j_{t}}, \optEndPoint{i,j_{t}}) $, we align the block $ \blIntS{j} $ with the approximate interval $ \pi_{i}[\optStartPoint{i,j_{t}}', \optEndPoint{i,j_{t}}') $ (incurring cost $ \ed{\blIntS{j}, \pi_{i}[\optStartPoint{i, j_{t}}', \optEndPoint{i,j_{t}}')} $). For all other intervals, the alignment consists of deleting the block from $ \offlineGoodString $ and inserting the corresponding intervals from $ \pi_{i} $.

On the other hand, the sum $ \sum_{j=1}^{n^{\mpcSpaceRegime}} \ed{\blIntS{j}, \pi_{i}[\optStartPoint{i,j}, \optEndPoint{i,j})} $ bounds $ \ed{\offlineGoodString, \pi_{i}} $ via a slightly different alignment: for every medium interval, we align $ \blIntS{j} $ with the optimal interval $ \pi_{i}[\optStartPoint{i,j}, \optEndPoint{i,j}) $ (incurring cost $ \ed{\blIntS{j}, \pi_{i}[\optStartPoint{i,j}, \optEndPoint{i,j})} $), handling other intervals similarly as above.

By triangle inequalities, $ \blEd(\offlineGoodString) $ approximates the sum $ \sum_{i=1}^{5}\sum_{j=1}^{n^{\mpcSpaceRegime}}\ed{\blIntS{j}, \pi_{i}[\optStartPoint{i,j}, \optEndPoint{i,j})} $ with an error of at most $ 2 \sum_{i=1}^{5}\sum_{[\optStartPoint{i,j}, \optEndPoint{i,j}) \text{ is medium}} \left( (\optStartPoint{i,j}'-\optStartPoint{i,j}) + (\optEndPoint{i,j} - \optEndPoint{i,j}') \right) $. As $ \optStartPoint{i,j}' - \optStartPoint{i,j} \leq \edErr n^{\upperBoundED{i}-\mpcSpaceRegime} $ and $ \optEndPoint{i,j} - \optEndPoint{i,j}' \leq \edErr n^{\upperBoundED{i}-\mpcSpaceRegime} + \edErr \ed{ \lmedian[\ell_{j}, r_{j}) , \pi_{i}[\optStartPoint{i,j}, \optEndPoint{i,j})} $, we bound $ \blEd(\offlineGoodString) $ as:
\begin{align}
    \blEd(\offlineGoodString) &\leq \sum_{i=1}^{5} \sum_{j=1}^{n^{\mpcSpaceRegime}}\ed{\blIntS{j}, \pi_{i}[\optStartPoint{i,j}, \optEndPoint{i,j})} + 4 \edErr  \sum_{i=1}^{5} n^{\upperBoundED{i}} + 2\edErr \sum_{i=1}^{5}\ed{\lmedian, \pi_{i}}. \nonumber
\end{align}
Finally, by applying~\cref{cor:approximation.local.solution} and the fact that $ n^{\upperBoundED{i}} \leq (1+\edErr )\ed{\lmedian, \pi_{i}} $, we derive the following upper bound:
\begin{align}
    \blEd(\offlineGoodString) &\leq (1 + 223.5\edErr  + 124\edErr ^{2})\sum_{i=1}^{5}\ed{\lmedian, \pi_{i}} + \left( \dfrac{5}{1-6\edErr } + 60 \right)\card{B_{\lmedian}}. \nonumber
\end{align}

\paragraph{Bounding $ \blEd $ of the Output of~\nameref*{sub:combining_blocks}}\label{par:the_output_of_nameref_sub_combining_blocks_} 
Let $ \offlineIntS $ be the solution string obtained in step~\texttt{\nameref*{sub:combining_blocks}}. Since $ \offlineIntS $ minimizes the objective function over the search space:
\begin{align}
    \offlineIntS = \argmin_{\concateStr{a_{j_{1}}, \dots , a_{j_{k}}}} \blEd( \concateStr{a_{j_{1}}, \dots, a_{j_{k}}} ) , \nonumber
\end{align}
it must satisfy $ \blEd(\offlineIntS) \leq \blEd(\offlineGoodString) $. Therefore,
\begin{align}
    \blEd(\offlineIntS) &\leq (1 + 223.5\edErr  + 124\edErr ^{2})\sum_{i=1}^{5}\ed{\lmedian, \pi_{i}} + \left( \dfrac{5}{1-6\edErr } + 60 \right)\card{B_{\lmedian}}. \label{eq.ed.bound.blEd.offlineIntS}
\end{align}

\subsection{Analysis of~\nameref*{sub:postprocessing} Step and Proof of~\cref*{lem:mpc.estimate.local.median}}\label{sub:proof_of_cref_lem_bound_gmedian_to_offlineout_} 
Let $ \offlineOut $ be the output permutation in the step~\texttt{\nameref*{sub:postprocessing}}. We give an upper bound for $ \sum_{i=1}^{5}\ed{\offlineOut, \pi_{i}} $, finalizing the proof of~\cref{lem:mpc.estimate.local.median}.

Note that, removing element $ \specialChar $'s from $ \offlineIntS $ and replacing remaining $ \specialChar $ by unused elements in $ \Sigma $ to obtain $ \offlineOut $ does not increase the sum $ \blEd $. Then,
\begin{align}
    \sum_{i=1}^{5} \ed{\offlineOut, \pi_{i}} &\leq \blEd(\offlineIntS) \nonumber \\
                                       &\leq (1 + 223.5\edErr  + 124\edErr ^{2})\sum_{i=1}^{5}\ed{\lmedian, \pi_{i}} + \left( \dfrac{5}{1-6\edErr } + 60 \right)\card{B_{\lmedian}}. \nonumber
\end{align}
The last inequality is due to inequality~\eqref{eq.ed.bound.blEd.offlineIntS}.

By choosing $ \edErr = 0.000001 $, and the fact that $ \card{B_{\lmedian}} \leq \totalslack{\lmedian} \leq \totalslack{\gmedian} $, the above inequality becomes
\begin{align}
    \sum_{i=1}^{5} \ed{\offlineOut, \pi_{i}} \leq 1.000225 \sum_{i=1}^{5}\ed{\lmedian, \pi_{i}} + 66 \totalslack{\gmedian}, \nonumber
\end{align}
completing the proof of~\cref{lem:mpc.estimate.local.median}.

\subsection{From A Local Solution to A Global Median}\label{sec:from_an_optimal_median_of_five_permutations_to_an_optimal_median_of_the_input_set} 

In this section, we derive inequality~\eqref{eq:scalable.cycle.removal} of~\cref{lem:bound.gmedian.to.offlineOut} from~\cref{lem:mpc.estimate.local.median}.
Specifically, let $ \offlineOut $ be the output obtained from $ \scalableCycleRemoval $ algorithm with input $ \pi_{1}, \pi_{2}, \dots, \pi_{5} $, we show that
\begin{align}
    \ed{\offlineOut, \gmedian} \leq 0.0009 \sum_{i=1}^{5}\ed{\gmedian, \pi_{i}} + 266 \totalslack{\gmedian}. \tag{\ref{eq:scalable.cycle.removal}}
\end{align}

\begin{proof}[Proof of~\cref{lem:bound.gmedian.to.offlineOut}]
    Our starting point is a characterization of the distance between any two permutations $ x, y $ in terms of the quantities $ \totalslack{x} $ and $ \totalslack{y} $, specifically, $ \ed{x, y} \leq \card{B_{x}} + \card{B_{y}} \leq \totalslack{x} + \totalslack{y} $. 
    This inequality follows from the more general result for weighted Ulam metric, which we prove in~\cref{subsec:ulam} (\cref{lemma:ulam}). It follows,
    \begin{align}
        \ed{\offlineOut, \gmedian} &\leq \card{B_{\offlineOut}} + \totalslack{\gmedian}. \nonumber
    \end{align}
    It remains to bound $ \card{B_{\offlineOut}} $.

    By definition, $ B_{\offlineOut} = \bigcup_{1\leq i< j \leq 5}(\unali{\pi_{i}}^{\offlineOut} \cap \unali{\pi_{j}}^{\offlineOut}) $, thus
    \begin{align}
        \card{B_{\offlineOut}} &\leq \sum_{1\leq i< j \leq 5} \card{\unali{\pi_{i}}^{\offlineOut} \cap \unali{\pi_{j}}^{\offlineOut}} \leq \sum_{1\leq i <j \leq 5} \left( \card{\unali{\pi_{i}}^{\offlineOut}} + \card{\unali{\pi_{j}}^{\offlineOut}} - \ed{\pi_{i}, \pi_{j}} \right) \nonumber \\
                               &= \sum_{i=1}^{5} \Bigl(\ed{\offlineOut, \pi_{i}} + \ed{\offlineOut, \pi_{j}}- \ed{\pi_{i}, \pi_{j}}\Bigr) \nonumber \\ 
                               & (\text{applying~\cref{lem:mpc.estimate.local.median}, and for brevity, let $ \edErr' = 0.000225 $})\nonumber \\
                               &\leq \sum_{1\leq i< j \leq 5} \left( (1+\edErr')( \ed{\lmedian, \pi_{i}} + \ed{\lmedian, \pi_{j}} ) - \ed{\pi_{i}, \pi_{j}} \right) + 4 \cdot 66 \totalslack{\gmedian}    \nonumber \\
                               & = \totalslack{\lmedian} + 4\edErr' \sum_{i=1}^{5} \ed{\lmedian, \pi_{i}} + 264 \totalslack{\gmedian} \nonumber \\
                        & \leq  4 \edErr' \sum_{i=1}^{5} \ed{\gmedian, \pi_{i}} + 265 \totalslack{\gmedian}. \nonumber
    \end{align}
    The last inequality is due to $ \totalslack{\lmedian} \leq \totalslack{\gmedian} $ (see inequality~\eqref{eq:localleqglobal}), and by the fact that $ \sum_{i=1}^{5} \ed{\lmedian, \pi_{i}}\leq \sum_{i=1}^{5}\ed{\gmedian, \pi_{i}} $.

    It follows that
    \begin{align}
        \ed{\offlineOut, \gmedian} &\leq \totalslack{\gmedian} + \card{B_{\offlineOut}} \nonumber \\
                                   &\leq 4\edErr'  \sum_{i=1}^{5} \ed{\gmedian, \pi_{i}} + 266\totalslack{\gmedian}. \nonumber
    \end{align}
    As $ \edErr' = 0.000225 $, we obtain
    $  \ed{\offlineOut, \gmedian} \leq 0.0009 \sum_{i=1}^{5} \ed{\gmedian, \pi_{i}} + 266\totalslack{\gmedian}$.
\end{proof}

%% file: hammingSF.tex
\section{New Approximation Algorithms for
Weighted Rank Aggregation}\label{section:application}

In this section we prove~\cref{property:formal.metric.specific} for each metric space.
We show that for  Hamming, Kendall-tau and Spearman's footrule, $r=3$, and for the Ulam metric $r=5$.

Using~\cref{thm:generalframework} we then obtain the following result.

\begin{theorem}\label{thm:offline.weighted}
   ~\cref{alg.space.efficient.general.framework} provides a:
    \begin{itemize}
        \item $(1.75+O(\delta))$-approximation for the 1-median problem in $(\permutations, \dist^w_H)$ and $(\permutations, \dist_F)$ in linear time.
        \item $(1.9+O(\delta))$-approximation for the 1-median problem in $(\permutations, \dist^w_\tau)$ in $\bigO{n\log n}$ time.
        \item $(1.9677+O(\delta))$-approximation for the 1-median problem in $(\permutations, \dist^w_U)$ in $\bigO{n^{17}}$.
    \end{itemize}
\end{theorem}

\paragraph{Ulam and Kendall-tau analysis.} The analysis for Ulam and Kendall-tau metrics requires additional notation to handle unaligned elements and pairs.
For any $x,p \in \permutations$, we denote by $I_{p}^x$ the set of unaligned characters of $x$ with $p$. For the Ulam metric, these are unaligned characters of $x$ with respect to some alignment between $p$ and $x$, while in the Kendall-tau case, $I_{p}^x$ represents the set of unaligned pairs.

We also establish a relationship between $I_{p}^x \cap I_{q}^x$ and the respective distances, in $\dist^w_\tau$ and $\dist^w_U$. For any set $S$ of elements (characters or pairs), we use $\wnorm{S}$ to denote the sum of weights of elements in $S$.

\begin{proposition}\label{prop:identity}
    For every $p,q,x \in \permutations$,
    \[
    \wnorm{I_{p}^x \cap I_{q}^x} \leq \dist (p,x) + \dist (q,x) - \dist (p,q),
    \]
Where $\dist$ is either $\dist^w_\tau$ or $\dist^w_U$.
\end{proposition}

\begin{proof}
We have $\wnorm{I_{p}^x \cap I_{q}^x} = \wnorm{I_{p}^x} + \wnorm{I_{q}^x} - \wnorm{I_{p}^x \cup I_{q}^x}$.
Whereas, $I_{p}^x \cup I_{q}^x$ is the set of unmatched elements in either $p$ or $q$, hence, $\wnorm{I_{p}^x \cup I_{q}^x} \geq \dist (p,q)$.
\end{proof}

\ifarxiv
\subsection{Hamming}

For the weighted Hamming distance, we design a majority-based algorithm for subsets $Q$ of size 3.~\cref{alg:local-hamming} constructs a consensus permutation by assigning the majority element to each position where one exists among the three input permutations. Positions without majority agreement are filled arbitrarily using the remaining unassigned elements.

\begin{algorithm}[htbp]
    \caption{Majority-based median for weighted Hamming}
    \label{alg:local-hamming}
    \SetKwInOut{KwIn}{Input}
    \SetKwInOut{KwOut}{Output}
    \KwIn{$Q$ of size 3}
    \KwOut{Local solution $y \in \permutations$}
    
    Initialize $U \gets [n]$ \tcp{Set of unassigned elements}
    
    \For{each position $k \in [n]$}{
        \If{there exists element $e$ appearing at position $k$ in at least 2 permutations in $Q$}{
            Set $y[k] \gets e$ \\
            Remove $e$ from $U$
        }
        \Else{
            Mark position $k$ as unresolved
        }
    }
    
    \For{each unresolved position $k$ (in arbitrary order)}{
        Assign $y[k]$ to any element from $U$ \\
        Remove the assigned element from $U$
    }
    
    \KwRet{$y$}
\end{algorithm}

~\cref{alg:local-hamming} operates in linear time and space complexity. 
To verify correctness, observe that its output $y$ constitutes a valid permutation since no element can simultaneously be a majority in two distinct positions. 
We now proceed to demonstrate that the weighted Hamming distance satisfies~\cref{property:formal.metric.specific}.

\begin{lemma}\label{lemma:property2hamming}
For the metric space $(\permutations, \dist^w_H)$, any subset $Q \subseteq \permutations$ of size $3$, and any local solution $y$ obtained using~\cref{alg:local-hamming}, we have $\dist^w_H(x^*, y) \leq \opttotalslack$.
\end{lemma}

\begin{proof}
By definition of weighted Hamming distance:
\begin{align}
\dist^w_H(x^*,y) = 
\frac{1}{2}\sum_{\substack{k=1 \\ x^*[k] \neq y[k]}}^n w(x^*[k])+w(y[k])
\end{align}

Denote $Q = \{ \pi_1, \pi_2, \pi_3 \}$, then, the total slack is:
\begin{align}
\opttotalslack = \sum_{\substack{1 \leq k \leq n, \\ 1 \leq i < j \leq 3}} \left( \dist^w_H(x^*[k],\pi_i[k]) + \dist^w_H(x^*[k],\pi_j[k]) - \dist^w_H(\pi_i[k],\pi_j[k]) \right)
\end{align}

For the slack expression $\dist^w_H(x^*[k],\pi_i[k]) + \dist^w_H(x^*[k],\pi_j[k]) - \dist^w_H(\pi_i[k],\pi_j[k])$ at position $k$, we consider two cases. When $\pi_i[k] = \pi_j[k] = e$ and $x^*[k] \neq e$, this expression equals $w(x^*[k])+w(e)$. When $\pi_i[k], \pi_j[k], x^*[k]$ are three distinct elements, this expression equals $w(x^*[k])$.

By~\cref{alg:local-hamming}, position $k$ is either resolved (has majority) or unresolved (no majority). We can bound the slack as:
\begin{align*}
\opttotalslack \geq \sum_{\substack{k \text{ resolved} \\ x^*[k] \neq y[k]}} w(x^*[k]) + w(y[k]) + \sum_{\substack{k \text{ unresolved} \\ x^*[k] \neq y[k]}} w(x^*[k])
\end{align*}

For unresolved positions where $x^*[k] \neq y[k]$, note that $y[k]$ must appear at some unresolved position $k'$ in $x^*$. This means $y[k'] \neq x^*[k']$ for some unresolved position $k'$. Hence, 
\begin{align*}
\sum_{\substack{k \text{ unresolved} \\ x^*[k] \neq y[k]}} w(x^*[k]) =
\frac{1}{2}\sum_{\substack{k \text{ unresolved} \\ x^*[k] \neq y[k]}} w(x^*[k]) + w(y[k]) 
\end{align*}

Since resolved positions directly contribute $w(x^*[k]) + w(y[k])$ to the slack when $x^*[k] \neq y[k]$, and unresolved positions contribute sufficient slack through the pairing argument, we have:
\begin{align*}
\opttotalslack \geq \dist^w_H(x^*,y)
\end{align*}
\end{proof}

\subsection{Spearman's footrule}

For the Spearman's footrule distance, we design a position-wise median algorithm for subsets $Q$ of size 3.
~\cref{alg:local-spearman} operates in two phases: first, it computes a pseudo-permutation $z$ by taking the median element at each position independently among the three input permutations, second, since $z$ may not constitute a valid permutation due to potential duplicate elements, it converts $z$ to a valid permutation by sorting the element-position pairs and reassigning ranks according to the sorted order of elements.

\begin{algorithm}[htbp]
    \caption{Position-wise median for Spearman's footrule}
    \label{alg:local-spearman}
    \SetKwInOut{KwIn}{Input}
    \SetKwInOut{KwOut}{Output}
    \KwIn{$Q$ of size 3}
    \KwOut{Local solution $y \in \permutations$}
    
    Initialize pseudo-permutation $z$ of length $n$
    
    \For{each position $k \in [n]$}{
        $z[k] \gets \texttt{median}(Q,k)$ \tcp{median element at position $k$}
    }
    
    \tcp{Convert pseudo-permutation to valid permutation}
    Let $S = \{(z[k], k) : k \in [n]\}$ \tcp{Element-position pairs}
    
    Sort $S$ by element value to get $S' = \{(e_1, k_1), (e_2, k_2), \ldots, (e_n, k_n)\}$
    
    \For{$i = 1$ to $n$}{
        $y[k_i] \gets i$ \tcp{Assign rank $i$ to position $k_i$}
    }
    
    \KwRet{$y$}
\end{algorithm}

~\cref{alg:local-spearman} operates in linear time and space complexity as the sorting step can be implemented using counting sort.
Observe that the output $y$ is a valid permutation as the sorting step ensures each element from $[n]$ appears exactly once. The algorithm minimizes position-wise Spearman's footrule distances in the first phase, then finds the closest valid permutation in the second phase.

We use the following simple rearrangement inequality (a special case of more general results in Day~\cite{day1972rearrangement}) to establish our results.

\begin{lemma}[Day~\cite{day1972rearrangement}]\label{lemma:order}
If $a_1 \leq a_2 \leq \ldots \leq a_n$ and $b_1 \leq b_2 \leq \ldots \leq b_n$, then, 
\[
\sum_{k=1}^n |a_k-b_k| = \min_{\pi \in \permutations} \sum_{k=1}^n |a_k-b_{\pi(k)}|.
\]
\end{lemma}

From~\cref{lemma:order}, we obtain that $y$ is the closest permutation to the pseudo-permutation $z$.
\begin{equation}\label{eq:sf}
\dist_F(z,y) = \min_{x \in \permutations} \dist_F(z,x)
\end{equation}

To show that the Spearman's footrule distance satisfies~\cref{property:formal.metric.specific}, we bound the distance from any permutation to $z$.

\begin{lemma}\label{lemma:sf}
For every $x \in \permutations$, $\dist_F(x,z) \leq \frac{1}{2}\totalslack{x}$.
\end{lemma}

\begin{proof}
Denote $Q = \{ \pi_1, \pi_2, \pi_3 \}$, without loss of generality let $\pi_i[k]$ be the median element at position $k$, that is $z[k]=\pi_i[k]$. Also, let $\pi_j[k]$, with $j \neq i$, be the element such that $x[k] \leq \pi_i[k] \leq \pi_j[k]$ or $x[k] \geq \pi_i[k] \geq \pi_j[k]$. Consider again the expression $\dist_F(x,\pi_i[k])+\dist_F(x,\pi_j[k])-\dist_F(\pi_i[k],\pi_j[k])$ at position $k$. By our selection of $i,j$ we get that this expression equals $2\dist_F(x[k],\pi_i[j])=2\dist_F(x[k],z[k])$. Therefore,
\begin{align*}
    \dist_F(x[k],z[k]) \leq \frac{1}{2}\sum_{1 \leq i < j \leq 3}\dist_F(x,\pi_i[k])+\dist_F(x,\pi_j[k])-\dist_F(\pi_i[k],\pi_j[k]).
\end{align*}
Summing over all positions, 
\begin{align*}
    \dist_F (z,x) \leq \frac{1}{2} \sum_{k=1}^n \sum_{1 \leq i < j \leq 3}\dist_F(x,\pi_i[k])+\dist_F(x,\pi_j[k])-\dist_F(\pi_i[k],\pi_j[k]) = \frac{1}{2}\totalslack{x}.
\end{align*}
\end{proof}

Combining~\cref{lemma:sf} and Equation~\ref{eq:sf}, with the triangle inequality, we have 
\begin{align*}
    \dist_F(x^*,y) \leq \dist_F(x^*,z) + \dist_F(z,y) \leq \frac{1}{2}\opttotalslack + \frac{1}{2}\opttotalslack = \opttotalslack.
\end{align*}

\subsection{Kendall-tau}\label{sec:the_median_problem_under_kendall_tau_distance}
For the weighted Kendall-tau distance, we propose an  algorithm that constructs a tournament graph and solves the resulting feedback arc set problem. The algorithm builds a majority graph $G_Q$ with vertex set $[n]$, where a directed edge $(a,b)$ exists if element $a$ precedes element $b$ in at least two of the three input permutations from $Q$. Each edge $(a,b)$ is assigned weight $\frac{1}{2}(w(a) + w(b))$. To solve the feedback arc set problem on $G_Q$, we employ the KWIK-SORT algorithm, which provides a 2-approximation guarantee when edge weights satisfy the triangle inequality~\cite{ACN08}. This algorithm runs in $\bigO{n \log n}$ time and uses linear space~\cite{im2020fast}.

For $x \in \permutations$, define $B_x$ as the set of element pairs that are unaligned when comparing $x$ with at least two permutations from $Q = \{ \pi_1, \pi_2, \pi_3 \}$,
\[
B_x = \bigcup_{1 \leq i < j \leq 3} (I_{\pi_{i}}^x \cap I_{\pi_{j}}^x).
\]

Note that for a pair of elements $\{a,b\}$ where $a$ precedes $b$ in permutation $x$, we have $\{a,b\} \notin B_x$ if and only if $a$ precedes $b$ in the majority of permutations from $Q$.
Also, using~\cref{prop:identity}, we have $\wnorm{B_{x^*}} \leq \opttotalslack$, where $\wnorm{B_{x^*}}$ denotes the sum of weights of elements in $B_{x^*}$.

Since $x$ is a valid permutation, we can transform $G_Q$ into an acyclic graph by redirecting edges of total weight at most $\wnorm{B_x}$. Moreover, every orientation of the edges that transforms $G_Q$ into an acyclic graph corresponds to a permutation.
This is precisely the feedback arc set problem, where the optimal solution has cost at most $\wnorm{B_{x^*}}$. 
As the weights in the graph satisfy the triangle inequality, the algorithm returns a permutation $y$ that is obtained by redirecting a set of edges $B_y$ with $\wnorm{B_y} \leq 2\wnorm{B_{x^*}} \leq 2\opttotalslack$.

\begin{lemma}
    For the metric space $(\permutations, \dist^w_\tau)$, any subset $Q \subseteq \permutations$ of size $3$, and any local solution $y$ obtained using the KWIK-SORT algorithm over $G_Q$, we have $\dist^w_\tau(x^*, y) \leq 3\opttotalslack$.
\end{lemma}

\begin{proof}
If both $x^*$ and $y$ agree with the majority of permutations on pair $\{a,b\}$, then clearly $x^*$ and $y$ agree on pair $\{a,b\}$. Hence, if $x^*$ and $y$ disagree on pair $\{a,b\}$, then $\{a,b\} \in B_y$ or $\{a,b\} \in B_{x^*}$. Therefore,
\begin{align*}
\dist^w_\tau(x^*, y) \leq  \wnorm{B_{x^*} \cup B_y} \leq \wnorm{B_{x^*}} + \wnorm{B_y} \leq 3\wnorm{B_{x^*}} \leq 3\opttotalslack.   
\end{align*}
\end{proof}

\subsection{Ulam}\label{subsec:ulam}
Analogously to the Kendall-tau case, we construct a tournament graph $G_Q$, however in this case we use $Q$ of size $5$. For the weighted Ulam distance, similar to~\cite{DBLP:conf/innovations/Chakraborty0K23}, we solve the resulting feedback vertex set problem and obtain a directed acyclic graph (DAG). 
We return the topological order of the remaining DAG and append the removed feedback vertices arbitrarily to the end of the resulting permutation.
Note that the weights are now assigned to the vertices. To solve this we use the 2-approximating feedback vertex set algorithm of Lokshtanov et al.~\cite{DBLP:journals/talg/LokshtanovMMPPS21}. This is a randomized algorithm that runs in time $\bigO{n^{17}}$.

For $x \in \permutations$, define $B_x$ as the set of elements that are unaligned when comparing $x$ with at least two permutations from $Q = \{ \pi_1, \pi_2, \pi_3, \pi_4, \pi_5 \}$,
\[
B_x = \bigcup_{1 \leq i < j \leq 5} (I_{\pi_{i}}^x \cap I_{\pi_{j}}^x).
\]
Again, using~\cref{prop:identity}, 
\begin{equation}\label{eq:ulambad}
    \wnorm{B_{x}} \leq \sum_{1 \leq i < j \leq 5} \wnorm{I_i \cap I_j} \leq \totalslack{x}.
\end{equation}

If we remove the vertices in $B_x$ from the graph $G_Q$ we get a DAG. To see this, consider any two vertices in the remaining graph, say $a,b$. Since each of $a$ and $b$ is aligned with at least 4 permutations, both $a$ and $b$ are aligned in at least 3 permutations. Hence, the edge between $a$ and $b$ corresponds to the order of $a$ and $b$ in $x$. Since $x$ is a permutation, the remaining graph is a DAG.

We now show that the weighted Ulam distance satisfies~\cref{property:formal.metric.specific}. Similar ideas were used in~\cite{DBLP:conf/innovations/Chakraborty0K23}, however, we extend the proof for every two permutations in $\permutations$.

\begin{lemma}\label{lemma:ulam}
For every $x,y \in \permutations$, $\dist^w_U (x,y) \leq \totalslack{x} + \totalslack{y}$.
\end{lemma}

\begin{proof}
Let $G_x = [n] \setminus B_x$, (resp. $G_y$), is the set of elements that are aligned in at least 4 permutations. For every pair of elements $a,b \in G_x$, in at least three permutations from $Q$ both $a,b$ are matched. Thus, the elements in $G_x$ form a subsequence in $x$ such that the relative order of every $a,b \in G_x$ must be consistent with their majority order across the permutations in $Q$ (resp. for $G_y$).
Consequently, the relative order of every two $a,b \in G_{x} \cap G_y$ must match the majority order in $Q$. It follows that $G_{x} \cap G_y$ corresponds to a common subsequence of $x$ and $y$.

Therefore,
\begin{align*}
    \dist^w_U (x,y) \leq \wnorm{B_x \cup B_y} \leq \wnorm{B_x} + \wnorm{B_y} \leq \totalslack{x} + \totalslack{y}
\end{align*}
The last inequality follows from Equation~\ref{eq:ulambad}.
\end{proof}

Since the local solution is a 2-approximation we get that its cost is at most twice that of an optimal solution for $Q$ which is at most $2\opttotalslack$. Together with~\cref{lemma:ulam}, we obtain~\cref{property:formal.metric.specific} with $r=5$ and $C=3$, namely,
$\dist^w_U(x^*,y) \leq 3\opttotalslack$.

\fi

%% file: MPCimplementation.tex
\section{MPC Implementation}\label{sec:mpc_implementation}
In this section, we present the implementation of our framework (\cref{alg.space.efficient.general.framework}) within the \gls{mpc} model. We work in the sublinear regime, where the local memory per machine is $ \otilda{ n^{1-\mpcSpaceRegime} } $ for a constant $ \mpcSpaceRegime \in (0,1) $. 
The set of elements and the set of positions, each represented by $\{1, 2, \dots, n\}$, is partitioned into $ n^{\mpcSpaceRegime} $ subsets $ \Sigma_{1} \cup \Sigma_{2} \cup \dots \cup \Sigma_{n^{\mpcSpaceRegime}} $, where $ \Sigma_{i} = \{(i-1)n^{\mpcSpaceRegime}+1, (i-1)n^{\mpcSpaceRegime}+2, \dots, in^{\mpcSpaceRegime}\} $. When sampling permutations, we assign $ n^{\mpcSpaceRegime} $ machines for each permutation. The positions of the elements in $ \Sigma_{i} $ and the elements located in the positions in $ \Sigma_{i} $ of each permutation are stored on their $ i $'th machine.

Implementing $ \scalableCycleRemoval $ in MPC yields a $ (2-\alpha) $-approximation algorithm for the Ulam $ 1 $-median problem with $ \bigO{1} $ rounds and $ \otilda{n^{1+6\mpcSpaceRegime}} $ total space.

\mpcUlamResult*

We also demonstrate the applicability of our framework to weighted settings. We implement the algorithms for computing local solutions under element-weighted Hamming, Spearman's footrule, and Kendall-tau distances.

\mpcresult*

We first present MPC implementations of algorithms that produce local solutions for Ulam, weighted Hamming, Spearman's footrule, and weighted Kendall-tau distances. We then implement the complete framework in the MPC model using these local solution algorithms as subroutines.

\subsection{Finding a local solution under Ulam distance in MPC}\label{sec:mpc.Ulam} 
We implement the offline $ \scalableCycleRemoval $ algorithm in the \gls{mpc} model. We proceed in four phases.

\paragraph{Phase 1}\label{par:mpc.local.round.one} 
(\texttt{\nameref*{sub:windows_decomposition}} and \texttt{\nameref*{sub:local_aggregating_algorithm}}) For each $ j\in \{ 1,2, \dots, n^{\mpcSpaceRegime} \} $, and for each tuple of five windows $ \left(\{ \startPoint{1}, \endPoint{1} \}, \{ \startPoint{2}, \endPoint{2} \}, \dots, \{ \startPoint{5}, \endPoint{5} \}\right) \in \windTupSet{j} $, we assign one machine with input consisting of: $ j $, $ \left(\pi_{1}[\startPoint{1}, \endPoint{1}), \pi_{2}[\startPoint{2}, \endPoint{2}), \dots, \pi_{5}[\startPoint{5}, \endPoint{5}) \right) $, and $ (\pair{ \startPoint{1}, \endPoint{1} }, \pair{ \startPoint{2}, \endPoint{2} }, \dots, \pair{ \startPoint{5}, \endPoint{5} }) $. This machine runs~\texttt{\nameref*{sub:local_aggregating_algorithm}} algorithm on its input to obtain a candidate block $ \mpcLocalOut = \localAggregation\left(\{ \pi_{i}[\startPoint{i}, \endPoint{i}\right) \}_{i=1}^{5}) $. The output of this machine is the tuple
\begin{align}
    \left(j, \mpcLocalOut, \pair{ \startPoint{1}, \endPoint{1} }, \pair{ \startPoint{2}, \endPoint{2} }, \dots, \pair{ \startPoint{5}, \endPoint{5} }, \sum_{i=1}^{5}\ed{\mpcLocalOut, \pi_{i}[\startPoint{i}, \endPoint{i})} \right). \nonumber
\end{align}

\paragraph{Phase 2}\label{par:mpc.local.round.two} 
(\texttt{\nameref*{sub:combining_blocks}}) the output of every machine in~\nameref*{par:mpc.local.round.one} is collected in one machine. We group these outputs based on the value of $ j $, forming sets $ C_{ 1 }, C_{ 2 }, \dots, C_{ n^{\mpcSpaceRegime} } $. These sets allows us to run~\cref{alg.combining.blocks}
on this machine exactly the same as in the offline setting. Let $ \answer $ be the output computed by this machine.

\paragraph{Phase 3}\label{par:mpc.local.round.three} 
(Traceback) The output $ \answer $ obtained in~\nameref*{par:mpc.local.round.two} implicitly defines a sequence of $ n^{\mpcSpaceRegime} $ blocks. Each block corresponds to either a candidate block generated by the~\texttt{\nameref*{sub:local_aggregating_algorithm}} algorithm in~\nameref*{par:mpc.local.round.one}, or a dummy block consisting of $ \blSize $ dummy elements. We reconstruct this sequence by backtracking via the array $ P $, starting from $ \answer.P $. This process identifies the machines from~\nameref*{par:mpc.local.round.one} responsible for the computed blocks. Additionally, we allocate new machines to generate the required blocks of dummy elements. The union of these machines, denoted by $ \machineSet{\offlineIntS} $, collectively stores the intermediate string $ \offlineIntS $.

\paragraph{Phase 4}\label{par:mpc.local.round.four} 
(\texttt{\nameref*{sub:postprocessing}}) In this phase, we transform $ \offlineIntS $ into the output permutation $ \mpcOut $ by simulating the \texttt{\nameref*{sub:postprocessing}} step of the offline algorithm. This involves two operations: removing the leftmost $ \specialChar $ elements until the length of $ \offlineIntS $ reduces to $ n $, and replacing the remaining $ \specialChar $ elements with distinct unused elements in increasing order.

We simulate the first operation in constant rounds. We first compute the total number of $ \specialChar $ elements to be removed using a constant-depth broadcast tree. Then, also using a broadcast tree of constant depth, we ensure that each machine in $ \machineSet{\offlineIntS} $ learns exactly how many of its own $ \specialChar $ elements must be removed. Let $ \offlineIntS' $ be the resulting string, distributed across the machine set $ \machineSet{\offlineIntS'} $.

It remains to replace the remaining $ \specialChar $ positions in $ \offlineIntS' $ with unused elements. We employ two sets of auxiliary machines: $ \machine{i}{1} $ to manage unused elements in $ \Sigma_{i} $, and $ \machine{i}{2} $ to manage unassigned positions within the range $ \Sigma_{i} $.

The machines in $ \machineSet{\offlineIntS'} $ send all used elements belonging to $ \Sigma_{i} $ to $ \machine{i}{1} $. Each machine $ \machine{i}{1} $ can then identify the complementary set of unused elements.

We first determine the global position of every block in $ \offlineIntS' $. This is achieved by computing the prefix sum of block lengths via a broadcast tree. With this offset information, machines in $ \machineSet{\offlineIntS'} $ identify the global indices of their $ \specialChar $ elements and send these positions to the corresponding machines $ \machine{i}{2} $.

Finally, we assign the unused elements to the unassigned positions. Although the lists in $ \machine{i}{1} $ and $ \machine{i}{2} $ are locally sorted, the number of elements in machine $ \machine{i}{1} $ may not be the same as the number of positions in machine $ \machine{i}{2} $. To resolve this, we compute the rank of each unused element (relative to all unused elements) and each unassigned position (relative to all unassigned positions) using a constant-depth broadcast tree. An element and a position sharing the same rank $ k $ are then paired: specifically, the element is routed to the machine in $ \machineSet{\offlineIntS'} $ that holds the position of rank $ k $, completing the permutation.

\begin{lemma}\label{lem:mpc.scalable.cycle.removal}
    The $ \scalableCycleRemoval $ algorithm can be implemented in \gls{mpc} model in $ \bigO{1} $ rounds, using $ \bigO{n^{1-\mpcSpaceRegime}} $ memory per machine and $ \bigO{n^{7\mpcSpaceRegime}\log^{8}n} $ machines. 
\end{lemma}

\begin{proof}
    It is clear from the algorithm description that the algorithm runs in $ \bigO{1} $ rounds.

    \vspace{2mm}
     Next, we justify the memory per machine.
In~\nameref*{par:mpc.local.round.one}, the input of each machine
    \begin{align}
        j; \pi_{1}[\startPoint{1}, \endPoint{1}), \pi_{2}[\startPoint{2}, \endPoint{2}), \dots, \pi_{5}[\startPoint{5}, \endPoint{5}) ; \{ \startPoint{1}, \endPoint{1} \}, \{ \startPoint{2}, \endPoint{2} \}, \dots, \left\{ \startPoint{5}, \endPoint{5} \right\}  \nonumber
    \end{align}
    is of size $ \bigO{\blSize} $, since $ [\startPoint{i}, \endPoint{i}) $ is a window of size at most $ \bigO{\blSize} $. The algorithm \texttt{\nameref*{sub:local_aggregating_algorithm}} (\cref{alg.local.cycle.removal}) can be implemented with $ \bigO{\blSize} $ memory. In the output of each machine,
    \begin{align}
    \left(j, \mpcLocalOut, \{ \startPoint{1}, \endPoint{1} \}, \{ \startPoint{2}, \endPoint{2} \}, \dots, \{ \startPoint{5}, \endPoint{5} \}, \sum_{i=1}^{5}\ed{\mpcLocalOut, \pi_{i}[\startPoint{i}, \endPoint{i})} \right), \nonumber
\end{align}
computing the distance $ \sum_{i=1}^{5} \ed{\mpcLocalOut, \pi_{i}[\startPoint{i}, \endPoint{i})} $ can be done with $ \bigO{\blSize} $ memory. Storing the output requires $ \bigO{\log n} $ memory. Thus, each machine in~\nameref*{par:mpc.local.round.one} uses $ \bigO{\blSize} = \bigO{n^{1-\mpcSpaceRegime}} $ memory.

    In~\nameref*{par:mpc.local.round.two}, the machine collects the outputs of all machines in~\nameref*{par:mpc.local.round.one}, each with size $ \bigO{\log n} $. Since there are $ \bigO{n^{7\mpcSpaceRegime}\log^{8}n} $ tuples in $ \windTupSet $ (see~\cref{sub:windows_decomposition}), the number of machines used in~\nameref*{par:mpc.local.round.one} is $ \bigO{n^{7\mpcSpaceRegime}\log^{8}n} $. Thus, the total size of the input collected by the machine in~\nameref*{par:mpc.local.round.two} is $ \bigO{n^{7\mpcSpaceRegime}\log^{9}n} $. With $ \mpcSpaceRegime < 1/8 $, we have $ \bigO{n^{7\mpcSpaceRegime} \log^{9}n} = \bigO{n^{1-\mpcSpaceRegime}} $. Running~\cref{alg.combining.blocks} requires $ \bigO{n^{7\mpcSpaceRegime}\log^{9}n} $ memory. Thus, the machine in~\nameref*{par:mpc.local.round.two} uses $ \bigO{n^{1-\mpcSpaceRegime}} $ memory.

    It is obvious that each machine in~\nameref*{par:mpc.local.round.three} and~\nameref*{par:mpc.local.round.four} uses $ \bigO{n^{1-\mpcSpaceRegime}} $ memory. Overall, the MPC $ \scalableCycleRemoval $ algorithm uses $ \bigO{n^{1-\mpcSpaceRegime}} $ memory per machine.

    \vspace{2mm}
    Next, we argue about the number of machines used. The number of machines in~\nameref*{par:mpc.local.round.one} is $ \bigO{n^{7\mpcSpaceRegime}\log^{8}n} $. We only use only one machine  in~\nameref*{par:mpc.local.round.two}. In~\nameref*{par:mpc.local.round.three} and~\nameref*{par:mpc.local.round.four}, the number of machines used is $ \bigO{n^{\mpcSpaceRegime}} $. Thus, the total number of machines used is $ \bigO{n^{7\mpcSpaceRegime}\log^{8}n} $.

    \vspace{2mm}

\end{proof}
\subsection{Finding a local solution under Hamming distance in MPC} \label{sec:mpc.hamming}

We simulate~\cref{alg:local-hamming} in the \gls{mpc} model. To facilitate this, we allocate an additional $ n^{\mpcSpaceRegime} $ machines to store the local solution. Assigning the majority element to each position (if one exists) is straightforward.

It remains to fill out the positions without majority agreement with the unused elements. For each $ i $, we introduce two additional machines: $ \machine{i}{1} $, which stores unused elements in $ \Sigma_{i} $, and $ \machine{i}{2} $, which stores unassigned positions in $ \Sigma_{i} $. In the first round, $ \machine{i} $ sends unused elements to $ \machine{i}{1} $, and unassigned positions to $ \machine{i}{2} $. 

Since we just need to provide a one-to-one mapping from the list of unused elements to the list of unassigned positions, in an arbitrary manner, we can apply the filling strategy described in \nameref{par:mpc.local.round.four} (\cref{sec:mpc.Ulam}).

\begin{lemma}\label{cor:mpc.hamming}
    \cref{alg:local-hamming} can be implemented in \gls{mpc} model in $ \bigO{1} $ rounds, using $ \otilda{n} $ total space.
\end{lemma}

\subsection{Finding a local solution under Spearman's footrule distance in MPC} \label{sec:mpc.spearman}

We simulate~\cref{alg:local-spearman} in the \gls{mpc} model. Similar to implementation under Hamming distance, we allocate an additional $ n^{\mpcSpaceRegime} $ machines to store the local solution. Computing the median element at each position is straightforward; however, the resulting string $ z $ may contain repeated elements and thus might not be a valid permutation.

To obtain a valid permutation,~\cref{alg:local-spearman} sorts the entries of $ z $ and assigns the sorted entries to positions $ 1,2,\dots, n $ in increasing order. The newly assigned values are mapped back to the original positions. To implement this step in the \gls{mpc} model, we use a standard sorting algorithm in the \gls{mpc} model~\cite{DBLP:conf/isaac/GoodrichSZ11} to sort all entries of $ z $ in increasing order. This sorting ensures that every entry in machine $ \machine{i} $ is less than every entry in machine $ \machine{i+1} $, and can be completed in a constant number of rounds. The final step is to map these sorted entries to the set $ \{1, 2, \dots, n\} $ in increasing order. This assignment is analogous to the process used for Hamming distance, where unused elements are matched to unassigned positions. It can also be accomplished in $ \bigO{1} $ rounds using $ \otilda{n} $ total space.

\begin{lemma}\label{cor:mpc.spearman}
    \cref{alg:local-spearman} can be implemented in \gls{mpc} model in $ \bigO{1} $ rounds, using $ \otilda{n} $ total space.
\end{lemma}

\subsection{Finding a local solution under Kendall-tau distance in MPC} \label{sec:mpc.kendall.tau}
For Kendall-tau distance, we apply the KWIK-SORT algorithm to solve the feedback arc set problem on the majority graph induced by $ \smallSubset $. In this graph, every vertex corresponds to an element, and there is a directed edge from vertex $ a $ to vertex $ b $ if $ a $ precedes $ b $ in at least two permutations from $ \smallSubset $. In our \gls{mpc} implementation, edges do not need to be stored explicitly; determining the direction of an edge between two elements requires only accessing their positions in the three permutations of $ \smallSubset $. The KWIK-SORT algorithm starts with a random pivot, partitions the remaining vertices according to the direction of the edges to the pivot, and recurses on the resulting sets.

To simulate this in constant rounds, we adapt the constant-round MPC implementation of KWIK-SORT from~\cite{im2020fast}, in which multiple pivots are selected per round to form a decision tree that partitions the vertices into multiple sets for parallel recursion. Unlike~\cite{im2020fast}, where $ \bigO{n^{1/2}} $ pivots are chosen per round due to $ \otilda{n} $ local space, our local space constraint allows only $ \bigO{n^{1-\mpcSpaceRegime}} $ pivots. However, since partition sizes cannot be too large~\cite[Lemma 7]{im2020fast}, choosing $ \Omega(\log(n)) $ pivots per round suffices to maintain constant total rounds. Furthermore, by storing edges implicitly via element positions, we keep the total space at $ \bigO{n} $, satisfying the local and total space constraints.

\begin{lemma}\label{cor:mpc.kendall}
    Computing a local solution under Kendall-tau distance can be implemented in \gls{mpc} model in $ \bigO{1} $ rounds, using $ \otilda{n} $ total space.
\end{lemma}

\subsection{Implementing \cref{alg.space.efficient.general.framework} in MPC}\label{sec:mpc.framework}
Let us first review~\cref{alg.space.efficient.general.framework}, which will be simulated by our \gls{mpc} algorithm. The algorithm first samples a set of $ \bigO{\log(n)/\distPara} $ permutations from the input set $ \inpset $, and includes these permutations in a candidate set $ \candidates $. Then, for each subset $ \smallSubset $ of $ r $ permutations sampled uniformly at random from the input set $ \inpset $, we find a local solution from $ \smallSubset $, and add this solution to $ \candidates $. Recall that $ r $ is a constant that will be chosen depending on the specific metric. Finally, we sample a set $ S $ of $ \bigO{\log(n)/\distPara^{2}} $ permutations from $ \inpset $, and return the permutation $ x\in \candidates $ with the smallest $ \cost{x}{S} $.~\cref{alg.mpc.general.framework} presents our \gls{mpc} implementation of this algorithm. 

\begin{algorithm}[htbp]
    \SetKwInOut{KwIn}{Input}
    \SetKwInOut{KwOut}{Output}
    \KwIn{Set $ \inpset = \{p_{1}, p_{2}, \dots, p_{m}\} \subseteq \metricspace $, constants $ r, \distPara $.}
    \KwOut{A permutation.}
        $ \candidates\gets $ a set of $ \bigO{\log(n)} $ permutations sampled uniformly at random from $ \inpset $. \nllabel{line:mpc.sampled.input.set.for.candidates}

        \For{$ i=1,2,\dots, \bigO{\log( n )/\distPara} $, in parallel \nllabel{line:mpc.find.local.solution.start}}{
            $ \smallSubset_{i}\gets $ a set of $ r $ permutations sampled uniformly at random from $ \inpset $.

            $ y_{i}\gets \texttt{localAlg}(\smallSubset_{i}) $ \nllabel{line:mpc.local.algorithm}

            $ C\gets C\cup \{y_{i}\} $
        } \nllabel{line:mpc.find.local.solution}

        Sample a set $ S $ of $ \bigO{\log( n )/\distPara^{2}} $ permutations uniformly at random from $ \inpset $. \nllabel{line:mpc.sampled.input.set.for.evaluation}

        Index every permutation in $ \candidates $ by $ \mpcIdx{x} $ and $ s\in S $ by $ \mpcIdx{s} $. \nllabel{line:mpc.indexing}

        \ForEach{$ (x,s)\in \candidates \times S $, in parallel}{
            Compute $ \dist(x,s) $. \nllabel{line:mpc.compute.distance}
        }

        \ForEach{$ x\in \candidates  $, in parallel}{ \nllabel{line:mpc.compute.cost.start}
            Compute $ \cost{x}{S} = \sum_{s\in S}\dist(x,s) $.
        } \nllabel{line:mpc.compute.cost.end}

        \KwRet{$ \mpcIdx{x} $ such that $ x = \argmin_{x\in \candidates}\cost{x}{S} $}. \nllabel{line:mpc.return.the.best.candidate}
    \caption{A general \gls{mpc} framework for median permutations.}
    \label{alg.mpc.general.framework}
\end{algorithm}

In~\cref{alg.mpc.general.framework}, from~\cref{line:mpc.sampled.input.set.for.candidates} to~\cref{line:mpc.sampled.input.set.for.evaluation}, we simulate the step of establishing a candidate set $ \candidates $ and sampling an evaluation set $ S $ in~\cref{alg.space.efficient.general.framework}. From~\cref{line:mpc.indexing} to~\cref{line:mpc.return.the.best.candidate}, we simulate the step of identifying a candidate that achieves a $ (1+\distPara) $-approximation to the best candidate in $ \candidates $ in~\cref{alg.space.efficient.general.framework}. 

Let $ \mpcLRound $ and $ \mpcLSpace $ be the number of rounds and spaces required to find a local solution from a set of $ r $ permutations. Denote by $ \mpcSDist $ the spaces required to compute the distance between two permutations in $ \permutations $ in $ \bigO{1} $ rounds. We now analyze the round and space complexity of~\cref{alg.mpc.general.framework}.

\paragraph{Space complexity.}
    The candidate set $ \candidates $ is initialized with $ \bigO{\log(n)} $ permutations in~\cref{line:mpc.sampled.input.set.for.candidates}, and after~\cref{line:mpc.find.local.solution}, it is complemented with $ \bigO{\log n} $ local solutions. Storing $ \candidates $ hence acquires a total space of $ \bigO{n\log n} $.

    Each execution of MPC algorithm for producing a local solution in~\cref{line:mpc.local.algorithm} uses a total space of $ \mpcLSpace $. Since from~\cref{line:mpc.find.local.solution.start} and~\cref{line:mpc.find.local.solution}, we run the this algorithm for $ \bigO{\log n} $ subsets of size $ r $ in parallel, we use a total space of $ \mpcLSpace \log n $ in these steps.
    Storing a sampled set of size $ \bigO{\log n/\distPara^{2}} $ in~\cref{line:mpc.sampled.input.set.for.evaluation} then requires $ \bigO{n\log n} $ total space.

    We analyze the space used from~\cref{line:mpc.indexing} to~\cref{line:mpc.return.the.best.candidate}.
    As both $ \candidates $ and $ S $ have $ \otilda{\log n} $ permutations, we can uniquely index every permutation $ x\in\candidates $ and $ s\in S $ using identifiers $ \mpcIdx{x} $ and $ \mpcIdx{s} $ respectively. The length of each index is logarithmic in the set size.

    In parallel, for each pair $ (x,s)\in \candidates \times S $, we assign a total space of $ \mpcSDist $ for $ (x,s) $ and compute their distances in $ \bigO{1} $ rounds. This results in a total space of $ \otilda{\mpcSDist \log^{2} n} $.

    Finally, for each $ x\in \candidates $, in parallel, we can use a broadcast tree of constant depth to compute $ \cost{x}{S} $ using $ \bigO{n \log n} $ space . Using the same argument, it is possible to find the permutation $ x\in \candidates $ with the smallest $ \cost{x}{S} $ using $ \bigO{n \log n} $ space. 

    Therefore, the total space used by~\cref{alg.mpc.general.framework} is 
    $ \otilda{\log^{2}(n) \cdot \mpcSDist + \log(n) \cdot \mpcLSpace + n \log(n)} $.

    \paragraph{Round Complexity.}

    Each execution of \gls{mpc} algorithm for producing a local solution in~\cref{line:mpc.local.algorithm} runs in $ \mpcLRound $ rounds. Since all executions of this algorithm are done in parallel from~\cref{line:mpc.find.local.solution.start} and~\cref{line:mpc.find.local.solution}, these steps take $ \bigO{1} $ rounds.

    All the other steps of the algorithm also run in $ \bigO{1} $ rounds. Therefore, the total round complexity of~\cref{alg.mpc.general.framework} is $ \mpcLRound $.

\begin{proof}[Proof of~\cref{thm:mpc.main.ulam.theorem}]
    We apply~\cref{alg.mpc.general.framework} with $ r=5 $ and in place of $ \texttt{localAlg} $ in~\cref{line:mpc.local.algorithm}, we use the \gls{mpc} implementation of $ \scalableCycleRemoval $ algorithm.

    From~\cref{lem:mpc.scalable.cycle.removal}, we can compute a local solution under Ulam distance in $ \mpcLRound = \bigO{1} $ rounds, using $ \mpcLSpace = \otilda{n^{1+6\mpcSpaceRegime}} $ total space. The Ulam distance between two permutations can be approximated to within a factor of $ 1+\approxFactMPCED $ in $ \bigO{1} $ rounds, with a total space of $ \mpcSDist = \otilda{n^{1+\mpcSpaceRegime}} $ by using the \gls{mpc} algorithm for computing edit distance proposed in~\cite{hajiaghayi2019massively}. Therefore, our algorithm runs in $ \bigO{1} $ rounds and uses a total space of $ \otilda{ n^{1+6\mpcSpaceRegime} } $.

    To analyze the approximation ratio, note that~\cref{alg.mpc.general.framework} is a direct simulation of~\cref{alg.space.efficient.general.framework}, and the local solution algorithm used in~\cref{line:mpc.local.algorithm} produces the same output as the offline $ \scalableCycleRemoval $ algorithm. Hence, our algorithm achieves a $ (2-\alpha) $-approximation for Ulam $ 1 $-median problem.
\end{proof}

\begin{proof}[Proof of~\cref{thm:mpc.main.theorem}]
    We apply~\cref{alg.mpc.general.framework} with $ r=3 $ and the \gls{mpc} implementations of local solution algorithms presented in~\cref{sec:mpc.hamming},~\cref{sec:mpc.spearman}, and~\cref{sec:mpc.kendall.tau} for Hamming, Spearman's footrule, and Kendall-tau distances respectively.

    For weighted Hamming and Spearmans's footrule distances, from~\cref{cor:mpc.hamming} and~\cref{cor:mpc.spearman}, we can compute a local solution in $ \mpcLRound = \bigO{1} $ rounds, using $ \mpcLSpace = \otilda{n} $ total space. The Hamming and Spearman's footrule distances between two permutations can be computed exactly in $ \bigO{1} $ rounds, using $ \mpcSDist = \bigO{n} $ total space. Therefore, our algorithm runs in $ \bigO{1} $ rounds and uses a total space of $ \otilda{n} $ for both Hamming and Spearman's footrule $ 1 $-median problems.

    For weighted Kendall-tau distance, from~\cref{cor:mpc.kendall}, we can compute a local solution in $ \mpcLRound = \bigO{1} $ rounds, using $ \mpcLSpace = \otilda{n} $ total space. The Kendall-tau distance between two permutations can be computed exactly in $ \bigO{1} $ rounds, using $ \mpcSDist = \otilda{n^{1+\mpcSpaceRegime}} $ total space. Therefore, our algorithm runs in $ \bigO{1} $ rounds and uses a total space of $ \otilda{n^{1+\mpcSpaceRegime}} $ for Kendall-tau $ 1 $-median problem.    

    As~\cref{alg.mpc.general.framework} is a direct simulation of~\cref{alg.space.efficient.general.framework}, and the \gls{mpc} algorithms for producing local solutions are straightforward implementations of their offline counterparts, our achieves the same approximation ratios as those in~\cref{thm:offline.weighted}.
\end{proof}